\documentclass{article}[12pt]
\usepackage{graphicx} 

\usepackage[utf8]{inputenc}
\usepackage{authblk}
\usepackage[T1]{fontenc}
\usepackage{listings}
\usepackage{comment} 
\usepackage[english]{babel}
\usepackage[dvipsnames]{xcolor}
\usepackage{enumerate}
\usepackage{hyperref}
\usepackage{soul}

\usepackage{amsmath}
\usepackage{amssymb}
\usepackage{amsthm}
\usepackage[dvipsnames]{xcolor}

\usepackage[shortlabels]{enumitem}
\usepackage{multicol}
\usepackage{lscape}
\usepackage{nicematrix}
\usepackage{cleveref}
\usepackage{cleveref}
\crefname{prop}{Proposition}{Propositions}
\usepackage{changepage}
\usepackage{subcaption}
\usepackage{pdflscape}
\usepackage{geometry}
\usepackage{tikz}
\usepackage{authblk}

\newtheorem{thm}{Theorem}[section]

\newtheorem{lemma}[thm]{Lemma}
\newtheorem{prop}[thm]{Proposition}
\newtheorem{cor}[thm]{Corollary}

\newtheorem{remark}[thm]{Remark}
\theoremstyle{definition}
\newtheorem{definition}[thm]{Definition}

\newcommand{\MM}{\mathcal{M}}
\newcommand{\F}{\mathcal{F}}
\newcommand{\Z}{\mathcal{Z}}
\newcommand{\E}{\mathcal{E}}

\newcommand{\RR}{\mathbb{R}}

\newcommand{\ZZ}{\mathbb{Z}}
\newcommand{\CC}{\mathbb{C}}
\newcommand{\pp}{\bar{\rho}}
\newcommand{\im}{\mathrm{Im}}
\newcommand{\rk}{\mathrm{rank}}
\newcommand{\Flat}{\mathrm{Flat}}
\newcommand{\SL}{S_{L_{12}}}

\newcommand{\mc}[1]{\begingroup\color{blue}#1\endgroup}
\newcommand{\ak}[1]{\begingroup\color{orange}#1\endgroup}
\newcommand{\jg}[1]{\begingroup\color{magenta}#1\endgroup}
\newcommand{\rh}[1]{\begingroup\color{cyan}#1\endgroup}

\newcommand{\nozero}[1]{\begingroup\color{green}#1\endgroup}
\newcommand{\OLD}[1]{\begingroup\tiny\color{gray}#1\endgroup}

\title{Computing phylogenetic invariants for time-reversible models: from TN93 to its submodels}

\author[2,3]{Marta Casanellas}
\author[4]{Jennifer Garbett}
\author[3]{Roser Homs}
\author[5]{Annachiara Korchmaros}
\author[1]{Niharika Chakrabarty Paul}

\affil[1]{\small{Max Planck Institute for the Mathematics in the Sciences}}
\affil[2]{\small{Universitat Polit\`ecnica de Catalunya}}
\affil[3]{\small{Centre de Recerca Matem\`atica}}
\affil[4]{\small{Lenoir-Rhyne University}}
\affil[5]{\small{University of Leipzig}}

\date{}
\begin{document}

\maketitle

\begin{abstract}
Phylogenetic invariants are equations that vanish on algebraic varieties associated with Markov processes that model molecular substitutions on phylogenetic trees. For practical applications, it is essential to understand these equations across a wide range of substitution models. Recent work has shown that, for equivariant models, phylogenetic invariants can be derived from those of the general Markov model by restricting to the linear space defined by the model (namely, the space of mixtures of distributions on the model). Following this philosophy, we describe the space of mixtures and phylogenetic invariants for time-reversible models that are not equivariant. Specifically, we study two submodels of the Tamura-Nei nucleotide substitution model (Felsenstein 81 and 84) using an orthogonal change of basis recently introduced for algebraic time-reversible models.

For tripods, we prove that the algebraic variety of each submodel coincides with the variety of Tamura-Nei intersected with the linear space of the submodel. In the case of quartets, we show that it is an irreducible component of this intersection. Moreover, we demonstrate that it suffices to consider only the binomial equations defining the linear space, which correspond to the natural symmetries of the model in the new coordinates. For each submodel, we explicitly provide equations defining a local complete intersection that characterizes the phylogenetic variety on a dense open subset containing the biologically relevant points.
\end{abstract}

\section{Introduction}

Algebraic phylogenetics is a growing research area at the intersection of algebra, geometry, statistics, and evolutionary biology (see the introductory papers \cite{eriksson2004,AllmanRhodeschapter4}). As already realized in the late eighties, equations that vanish on any distribution arising from a Markov process of nucleotide substitution on a phylogenetic tree can be very helpful in detecting the tree that produced the distribution (see, for example, the quote to this algebraic approach in the seminal book by biologist J. Felsenstein \cite[\S 5]{felsenstein2004}). These equations were called \emph{phylogenetic invariants} and obtaining them has been the main focus of algebraic phylogenetics: the set of distributions arising from a Markov process on a phylogenetic tree $T$ determine an algebraic variety $V_T$ and phylogenetic invariants correspond to polynomials in its ideal $\mathcal{I}(V_T)$. The underlying theory of phylogenetic invariants has been used to propose new methods of phylogenetic reconstruction \cite{chifmankubatko2014,  casfergar23} and model selection \cite{KDGC}, and to establish identifiability of complex evolutionary models \cite{allmanGTR,allman2009}.

Biological constraints on the parameters of the Markov process (namely, the stationary distribution and transition matrices) give rise to different nucleotide substitution models, which in turn lead to different algebraic varieties, $V_T^{\MM}$, for each substitution model $\MM$ and each phylogenetic tree $T$. Phylogenetic invariants have been found for some \emph{equivariant models}, which include the most {general Markov model}, the strand-symmetric model, and the \emph{group-based models}; \cite{Allman2008,Sturmfels2005,Draisma,cassull}. The latter are the simplest models and have been largely studied from the point of view of algebraic geometry because they give rise to toric varieties. However, these models are not exactly the models that biologists use the most frequently. For example, group-based models assume a uniform stationary distribution (which is too restrictive in most biological frameworks), and the general Markov model has too many parameters and becomes intractable from a likelihood perspective. 

Nucleotide substitution models that are widely used in biology assume \emph{time reversibility} and have an arbitrary stationary distribution. In \cite{CHT}, the authors introduce a new technique that allows these models to be tackled more reasonably from the point of view of algebraic geometry by setting up a new framework to find phylogenetic invariants for the so-called \emph{algebraic time-reversible} (ATR) models. So far, only the ATR model of Tamura and Nei (TN93, \cite{TN93_p}) has been studied in depth.
Precisely, the authors provided equations that cut out $V_T^\MM$ on an open set when $T$ is a tree with three or four leaves (\emph{quartet}).
It should be noted that quartets are one of the main objects of study in algebraic phylogenetics because reconstructing quartets correctly allows the reconstruction of larger trees (see, for instance, \cite{wqfm,QFM,zou2019}). {In this paper, we focus on two {nucleotide substitution} submodels of TN93: the F81 \cite{Felsenstein81} and F84 \cite{Felsenstein96} models proposed by Felsenstein in the early eighties. 
{They {generalize} the group-based models Jukes-Cantor and Kimura 2-parameter, respectively, to an arbitrary stationary distribution. In F81, the substitution rate only depends on the base distribution of the ancestral nucleotide and a single parameter. In contrast, F84 has two parameters to emphasize the difference between transitions and transversions.}
These are the only submodels of TN93 that are multiplicatively closed (an essential property for phylogenetic inference, see \cite{Sumner_gtr}), and have been implemented in several phylogenetic packages for tree topology estimation~\cite{Yang1997,Swofford2003}}.

During the last twenty years, some phylogenetic invariants have been found independently for some phylogenetic models; see \cite{Allman2008,Sturmfels2005,cassull,chifmankubatko2014}. However, following the recent philosophy developed for equivariant models, in \cite{casfer2024}, one should be able to obtain phylogenetic invariants for a phylogenetic tree evolving under a submodel of a given model $\mathcal{M}$, by simply imposing symmetries that define the submodel in the equations that cut out $V_T^{\MM}$. This paper aims to study the algebraic varieties defined by algebraic time-reversible models F81 and F84 and to check whether this latter result can be generalized to these models. For equivariant models, symmetries defined by a submodel are easily described by the action of a permutation group. For F81 and F84, the stationary distribution does not satisfy symmetric constraints, which impedes the search for a permutation group that leaves all distributions invariant.
For an evolutionary model $\MM$, \emph{symmetries of the model} are equalities between coordinates that hold for all points in $V_T^{\MM}$, independently of the tree $T$; see \cite{CFK, CasSteel}.
These linear equations are instances of \emph{model invariants} for $n$ leaves, that is, elements of $I(V_T^\MM)$ for any phylogenetic tree $T$ with $n$ leaves. In contrast, equations that hold for a particular tree topology but not for all trees are called \emph{topology invariants}. The most well-known topology invariants arise from imposing rank conditions on a matrix obtained by flattening the distribution vector; see \cite{Allman2008} and Section~\ref{sec:preliminaries}.  

Linear model invariants are relevant in biology because they actually define the space $\mathcal{D}^\MM_n$ of \emph{mixtures} of distributions of Markov processes on trees on $n$ leaves, see \cite{mat08,CFK,CasSteel}, but also are essential for model selection inference. That is, linear model invariants for $\mathcal{M}$ determine whether a distribution may have arisen from a mixture of evolutionary processes on trees or networks evolving under the model $\MM$. From the point of view of algebraic geometry, $\mathcal{D}^\MM_n$ is the linear span of $\cup V_T^{\MM}$, where the union runs over all trees with $n$ leaves.
For models that escape from equivariant models, $\mathcal{D}^\MM_n$ is not fully determined by the natural symmetries, and the extra equations needed are not straightforward to describe or derive. In \cite{CasSteel}, the authors provide a minimum spanning set of linear equations that define the space of mixtures for quartets evolving under F81. The derivation of those equations had to be done ad hoc; the equations had many terms and long coefficients, making their generalization to $n$ leaves particularly challenging. In this paper, we propose a new approach. We use the change of basis given in \cite{CHT} for TN93 to describe the symmetry equations in these coordinates (Propositions \ref{prop:quartet-eqF84}  and \ref{prop:quartet-eq}), and we prove that the extra equations needed to define the mixture space for quartets evolving under F81 or F84 involve at most four monomials and can be obtained by imposing symmetries on certain minors of the flattening matrices; see Theorems \ref{prop:mixtures_quartetF84} and \ref{thm:mixturesF81}. 

In addition to the space of mixtures, we use this approach to provide equations that define $V_T^{\MM}$ locally around the biologically relevant points for tripods and quartets evolving under F81 and F84. For $\MM=$ F81 or F84, providing a minimum set of generators for $I(V_T^\MM)$ can be computationally expensive and unnecessary for biological purposes. For equivariant models, it was proven in \cite{CFM} that rank conditions and equations from tripods can be used to provide a local complete intersection that coincides with $V_T^\MM$ on a Zariski open subset containing the biologically relevant points. We prove the analogs for tripods and quartets evolving under F81 and F84. In addition, in this case, we corroborate the philosophy underlying \cite{casfer2024}: this local complete intersection can be obtained by imposing symmetries on a local complete intersection that defines the variety for TN93. We refer to Theorems~\ref{thm:CIF84} and \ref{thm:CI} for precise statements.

The structure of the paper is as follows. 
In Section~\ref{sec:preliminaries} we introduce algebraic time-reversible models, emphasizing models F81, F84, and TN93 and basic tools and definitions in algebraic phylogenetics.
Section~\ref{sec:tripods} is devoted to the study of tripods, and in Section~\ref{sec:symmetry}, we compute symmetry equations for quartets. Section~\ref{sec:rank} deals with rank constraints arising from flattening matrices of the tensor of joint distributions. In Section~\ref{sec:mixtures}, we derive the defining equations of the spaces of mixtures for tripods and quartets. We provide equations defining the phylogenetic varieties around the biologically relevant points in Section~\ref{sec:complete_intersection} and close the paper with a Discussion section.

\textbf{Acknowledgments:} We are grateful the Women in Algebraic Statistics workshop, held in Oxford, and held by Jane Ivy Coons for providing the invaluable opportunity to collaborate. The workshop was supported by the following sources: St John's College, Oxford; the L'Oreal-UNESCO For Women in Science UK and Ireland Rising Talent Award in Mathematics and Computer Science (awarded to Jane Coons); the Heilbronn Institute for Mathematical Research; the UKRI/EPSRC Additional Funding Programme for the Mathematical Sciences.

M.C. and R.H. were partially supported by projects with references PID2019-103849G-I00 and PID2023-146936NB-I00, financed by 
MCIN/AEI/10.13039/501100011033/FEDER,UE, by the Severo Ochoa and Mar\'ia de Maeztu Program for Centers and Units of Excellence in R\&D (project CEX2020-001084-M), and by the AGAUR project 2021 SGR 00603 Geometry of Manifolds and Applications, GEOMVAP. 
 R.H. is supported by the postdoctoral fellowships programme Beatriu de Pin\'os (ref. 2021BP00119), funded by the Secretary of
Universities and Research (Government of Catalonia).
A.K. is supported by the German Research Foundation (DFG, STA850/49-1).

\section{Preliminaries}\label{sec:preliminaries}

A \emph{phylogenetic tree on $L$} is a tree $T=(V,E)$ with labeled leaves in bijection with the set $L=\{l_1,\dots,l_n\}$ and unlabeled interior nodes.
A phylogenetic tree is \emph{trivalent} if all of its interior nodes have degree three and we say that a phylogenetic tree is \emph{rooted} if we distinguish an interior node as its root $r$ and direct all edges away from it. Phylogenetic trees with $3$ (resp. $4$) leaves are called \emph{tripods} (resp. \emph{quartets}).

We associate a random variable $X_v$, with values in a state space $\Sigma$, with each vertex of $T$. We call $(X_v: v\in V)$ a \emph{Markov process on the phylogenetic tree $T$} if each random variable is conditionally independent of its non-descendants given its parent variable. This process on a rooted phylogenetic tree is a statistical model given by the parameters: 
\begin{itemize}
    \item $\pi^r$, the distribution at the root $r$, and  
    \item a $|\Sigma|\times |\Sigma|$ transition matrix, $M^e$, for each edge, $e:u\rightarrow v\in E$, that encodes state transition probabilities  $M^e_{i,j}=P(X_v=j\mid X_u=i)$.
\end{itemize}


If we model nucleotide substitutions in a DNA sequence, then this Markov process has state space $\Sigma=[4]$, with elements representing the four nucleotides in DNA sequences. In this work, we restrict ourselves to this case and fix notation such that $1$ refers to adenine, $2$ to guanine, $3$ to cytosine, and $4$ to thymine. Thus, $1$ and $2$ are purines and $3$ and $4$ are pyrimidines. By imposing constraints on the distribution of the root and transition matrices, one has different nucleotide substitution models, and according to \cite{CHT}, we have the following definition. 

\begin{definition}\label{def:ATR}
   Let $\Delta^3$ denote the 3-dimensional standard simplex in $\RR^4$ and let $\pi\in\Delta^3$ be a fixed distribution with non-zero entries. We say that a nucleotide substitution model $\mathcal{M}$ is \emph{algebraic $\pi$-time-reversible} (ATR) if the following properties hold:
   \begin{enumerate}
       \item time-reversibility with stationary distribution $\pi$: $\pi_i M_{i,j}=\pi_j M_{j,i}$ for $M\in\mathcal{M}$ and $i,j\in[4]$;
       \item commutativity: $M_1M_2=M_2M_1$, for any $M_1, M_2\in\mathcal{M}$.
   \end{enumerate}
\end{definition}

\begin{remark}\label{rmk:pi-orthog}\rm 
By \cite[Section 3]{CHT}, any $\pi$-time-reversible model is $\pi$-stationary, namely $\pi^t M=\pi^t$ for any transition matrix in the model. From an application point of view, the two properties above are crucial, as they allow estimation of the distribution at the root from data that have reached stationarity, thus eliminating $\pi^r$ from the model parameters.

Note that commutativity in Definition \ref{def:ATR} implies that there exists a basis that simultaneously diagonalizes all matrices in the model. We can choose this basis such that it is orthogonal with respect to the inner product $\langle u,v\rangle_\pi:=\sum_{i\in[4]}\frac{1}{\pi_i}u_iv_i$, $u_i,v_i\in\mathbb{C}$, see \cite[Definition 3.1]{CHT} for a detailed treatment. Such a basis is called a \emph{$\pi$-orthogonal basis}.
\end{remark}

Now, for a fixed root distribution $\pi^r=\pi$, the probability of observing configuration $i_1,\dots,i_n$, $i_j\in[4]$, at the leaves $l_1,\dots,l_n$ of the phylogenetic tree $T$ is a polynomial expression on the entries of the transition matrices. If $e_k:v_k\rightarrow l_k$, $k\in[n]$, is the pendant edge incident to leaf $l_k$, and $i_w\in[4]$ denotes the state of the random variable $X_w$ for any internal node $w$ (including the root $r$) this probability is:
\begin{equation*}
p^T_{i_1\dots i_n}=\sum_{i_r\in [4]} \pi_{i_r}M^{e_1}_{i_{v_1},i_1}\cdots M^{e_n}_{i_{v_n},i_n}\prod_{\begin{array}{c}{\scriptstyle e:u\rightarrow v \in E}\\ {\scriptstyle e\notin\{e_1,\dots,e_n\}}\end{array}} M^e_{i_u,i_v}.
\end{equation*}
\noindent

\noindent We can extend this expression to allow for complex parameters and think of the joint probability vector as a tensor in $\otimes^n\mathbb{C}^4$, via the following definition. 
\begin{definition}\label{def:param}
The \emph{joint probability tensor} associated to the phylogenetic tree $T$ evolving under an ATR model for nucleotides is the tensor 
$$p^T=\sum_{i_1,\dots,i_n\in[4]}p^T_{i_1\dots i_n} e^{i_1}\otimes\cdots\otimes e^{i_n}\in\otimes^n\mathbb{C}^4,$$
\noindent
where $\{e^1,e^2,e^3,e^4\}$ is the standard basis of $\mathbb{C}^4$ and $p^T_{i_1\dots i_n}$ is computed from parameters as above. We will denote by $\psi_T^\MM$, the map that assigns to each choice of parameters $\Theta=(M^e)_{e\in E}$ its corresponding joint probability tensor $p^T$. We will omit the superindex $T$ when there is no ambiguity on the underlying phylogenetic tree.
\end{definition}

\subsection{The Tamura-Nei nucleotide substitution model and its submodels}\label{subsec:submodels}

In \cite{CHT} the authors develop tools to compute phylogenetic invariants for ATR models and apply them to the Tamura-Nei (TN93)\cite{TN93_p} $\pi$-time-reversible nucleotide substitution model. 
Here, we deal with evolutionary models F81 \cite{Felsenstein81} and F84 \cite{Felsenstein96}, the two submodels of TN93 that are multiplicatively closed for fixed $\pi$;
a desirable property for phylogenetic inference, see \cite{Sumner_gtr}. Other commonly used submodels, such as HKY85 \cite{HKY85}, are not multiplicatively closed.

In this section, we introduce the models TN93, F81 and F84. For a fixed stationary distribution $\pi$, the transition matrices of TN93 are of the form

\begin{equation}\label{eq:TN93}
M=\begin{pmatrix}
*_{\mathtt{1}} &\pi_{2}c & \pi_{3}b & \pi_{4}b\\
\pi_{1}c &*_{2} & \pi_{3}b & \pi_{4}b\\
\pi_{1}b &\pi_{2}b & *_{3} & \pi_{4}d\\
\pi_{1}b & \pi_{2}b &\pi_{3}d &*_{4}
\end{pmatrix},
\end{equation}

\noindent
where $*_i$ is chosen to ensure that each row sums to $1$. In \cite[Example 3.5, Section 5]{CHT}, a detailed study of trees evolving under TN93 for a fixed arbitrary stationary distribution $\pi=(\pi_1,\pi_2,\pi_3,\pi_4)$ is provided, using the $\pi$-orthogonal basis
\begin{equation}\label{eq:basis}
    B=\left\{u^1=\begin{pmatrix}{\pi_1}\\ {\pi_2}\\ {\pi_3}\\ {\pi_4}\end{pmatrix}, \,
u^2 =\begin{pmatrix} \pi_1 \pi_{34}\\  \pi_2\pi_{34}\\ -\pi_3 \pi_{12}\\ -\pi_4 \pi_{12}
\end{pmatrix},\,
u^3=\frac{1}{\pi_{34}}\begin{pmatrix}0\\  0 \\ {\pi_3}\pi_4\\-\pi_3\pi_4\end{pmatrix}, \,
u^4=\frac{1}{\pi_{12}}\begin{pmatrix} {\pi_1}\pi_2 \\-\pi_1\pi_2\\ 0 \\0\end{pmatrix}\right\},
\end{equation}
where $\pi_{12}=\pi_1+\pi_2$ and $\pi_{34}=\pi_3+\pi_4$. For any transition matrix $M$ in the model, if the vectors of $B$ are written as columns of a matrix $A$, then $M^t$ diagonalizes as
\begin{equation}\label{eq:Mlambdas}
M^t=A\begin{pmatrix}\lambda_1 &0&0&0\\
0&\lambda_2 &0&0\\
0&0&\lambda_3 &0\\
0&0&0&\lambda_4
\end{pmatrix} A^{-1},
\end{equation}
where $\lambda_1=1$, $\lambda_2=\lambda_1-b$, $\lambda_3=\lambda_1-\pi_{12}b-\pi_{34}d$, $\lambda_4=\lambda_1-\pi_{34}b-\pi_{12}c$ are the eigenvalues of $M$. 

Emulating the case of group-based models, in which Jukes-Cantor (JC69) and Kimura 2-parameters (K2P) can be obtained from Kimura 3-parameters (K3P) by forcing certain eigenvalues to be equal, we study the F81 and F84 submodels of TN93 resulting from imposing $\lambda_2=\lambda_3=\lambda_4$ and $\lambda_3=\lambda_4$, respectively, in \eqref{eq:Mlambdas}. 

When all eigenvalues different from 1 are equal, that is, $\lambda_2=\lambda_3=\lambda_4$, we obtain the \emph{Felsenstein 1981 model} (F81)~\cite{Felsenstein81}. This translates into an equality among all parameters $b,c,d$ of TN93 in \eqref{eq:TN93}, which gives transition matrices of the form

\begin{equation}\label{eq:F81}
M=\begin{pmatrix}
*_{\mathtt{1}} &\pi_{2}b & \pi_{3}b & \pi_{4}b\\
\pi_{1}b &*_{2} & \pi_{3}b & \pi_{4}b\\
\pi_{1}b &\pi_{2}b & *_{3} & \pi_{4}b\\
\pi_{1}b & \pi_{2}b &\pi_{3}b &*_{4}
\end{pmatrix}.
\end{equation}

\noindent Since all non-diagonal entries of each column are equal, the probability of transitioning from one state to a different one does not depend on the initial state, thus justifying the fact that this is also called the \emph{equal input model}; see \cite{CasSteel} for an algebraic treatment of the equal input model in any number of states. Note that by considering the uniform distribution $\pi_1=\pi_2=\pi_3=\pi_4=1/4$ we recover the \emph{fully symmetric} or \emph{Jukes-Cantor model}.   

When only $\lambda_3=\lambda_4$, we instead obtain the \emph{Felsenstein 1984 model} (F84)~\cite{Felsenstein96}. While TN93 has a parameter for transitions (substitutions between purines or between pyrimidines) and one parameter for each type of transversions (substitutions between a purine and a pyrimidine), this model adds a constraint relating both types of transversions. In terms of the transition matrices and the usual notation for F84, if we take $\beta=b$, $\alpha=\pi_{12}c-\pi_{12}b$ (which is also equal to $\pi_{34}d-\pi_{34}b$ because $\lambda_3=\lambda_4$), we get $c=\beta+\alpha/\pi_{12}$, $d=\beta+\alpha/\pi_{34}$, and the F84 transition matrices can be rewritten as in the original paper:

\begin{equation}\label{eq:F84}
M=\begin{pmatrix}
*_{\mathtt{1}} &\pi_{2}\left(\beta+\frac{\alpha}{\pi_{12}}\right) & \pi_{3}\beta & \pi_{4}\beta\\
\pi_{1}\left(\beta+\frac{\alpha}{\pi_{12}}\right) &*_{2} & \pi_{3}\beta & \pi_{4}\beta\\
\pi_{1}\beta &\pi_{2}\beta & *_{3} & \pi_{4}\left(\beta+\frac{\alpha}{\pi_{34}}\right)\\
\pi_{1}\beta & \pi_{2}\beta &\pi_{3}\left(\beta+\frac{\alpha}{\pi_{34}}\right) &*_{4}
\end{pmatrix}.
\end{equation}

\noindent Conversely, any matrix of the F84 model with parameters $\alpha,\beta$ as \eqref{eq:F84} is a TN93 matrix with $\lambda_3=\lambda_4.$

\begin{remark}\label{rmk:generic_pi}\rm
The root distribution $\pi$ is fixed for a given model, and we assume it to be generic throughout the paper. In particular, this ensures that we are not dealing with simpler models such as JC69. See \cite[Remark 5.2]{CHT} for a detailed description of the genericity conditions.
\end{remark}

\subsection{Phylogenetic invariants in a $\pi$-orthogonal basis}\label{subsec:basis}

In this section, we adopt a change of coordinates that eases the study of TN93 and its submodels. 
We refer the reader to \cite{CHT} for further details. At the end of this section, we introduce a further rescaling of these coordinates that was hinted at in \cite{CasSteel}.

Let $A$ be the change of basis matrix from
a $\pi$-orthogonal basis $B=\{u^1,u^2,u^3,u^4\}$ (see Remark \ref{rmk:pi-orthog}) to the standard basis and let $\otimes$ denote the Kronecker product of matrices. Then, if $p$ is a joint probability tensor as in Definition \ref{def:param},
\begin{equation}\label{eq:coord_change}
    \bar{p}=\left(A^{-1}\otimes\cdots\otimes A^{-1}\right)p
\end{equation}
\noindent
gives its coordinates in the $\pi$-orthogonal basis $\{u^{i_1}\otimes\cdots\otimes u^{i_n}\}_{i_1,\dots, i_n}$.


Let $\MM$ denote F81 or F84, the submodels of TN93 presented in the previous sections. Consider the sets of transition matrices of the form \eqref{eq:F81} and \eqref{eq:F84}, and consider their diagonalizations $\Lambda^e=\mathrm{diag}(\lambda_1^e,\lambda_2^e,\lambda_3^e,\lambda_4^e)$, for $e \in E$, in the basis $B$ from \eqref{eq:basis}. We extend the model by allowing $\Lambda^e$ to take values in the complex numbers so that the joint probabilities are homogeneous polynomials on the eigenvalues in \eqref{eq:Mlambdas}.
Then we can parametrize the model for a tree with $n$ leaves as follows: 

\[\begin{array}{rcl}
  \varphi_T^\MM:  \prod_{e\in E}{\mathbb{C}}^4 
    & {\longrightarrow} &\bigotimes^n \mathbb{C}^4 \\
     \left( \Lambda^{e}\right)_{e \in E} & \mapsto  & \sum\limits_{i_1,\dots,i_n\in[4]} \bar{p}_{i_1\dots i_n}^T  u^{i_1}\otimes \dots \otimes u^{i_n} \, .
\end{array}\]
\noindent
Recall from \Cref{subsec:submodels} that 
$\lambda_2^{e}=\lambda_3^{e}=\lambda_4^{e}$ if $\MM=$F81 and $\lambda_3^e=\lambda_4^e$ if $\MM=$F84.

\begin{definition}
    The \emph{phylogenetic variety}, $CV_T^\MM$, of a phylogenetic tree $T$ evolving under $\MM$ is the Zariski closure of the image of $\varphi_T^\MM$.
    We denote the vanishing ideal of $CV_T^\MM$ in $\mathbb{C}[\bar{p}_{1..1},\dots,\bar{p}_{4..4}]$ by $I_T^\MM$ and call its elements \emph{phylogenetic invariants}. If a phylogenetic invariant holds for all trees with $n$ leaves, we call it a \emph{model invariant}. Otherwise, if it holds for some but not all trees, it is a \emph{topology invariant}.
\end{definition}


By tree topology, we refer to the topology of the labeled unrooted tree. In the case of quartets, there are three different topologies for trivalent trees, and they can be described using bipartitions of the leaf set $L$: $T_{12}=l_1l_2|l_3l_4$, $T_{13}=l_1l_3|l_2l_4$ and $T_{14}=l_1l_4|l_2l_3$, as seen in Figure \ref{fig:quartets}. {If $T$ is a (trivalent) quartet, $\varphi_T^\MM(D_1,D_2,D_3,D_4;D)$ stands for the tensor obtained by assigning diagonal matrix $D_i$ at pendant edge $e_i$ and diagonal matrix $D$ at the interior edge.}

\begin{figure}[!htb]
    \centering
    \begin{minipage}{0.3\textwidth}
        \centering
    \begin{tikzpicture}
        \draw (0,0)--(1,1);
        \draw (0,2)--(1,1);
        \draw (1,1)--(2.5,1);
        \draw (2.5,1)--(3.5,0);
        \draw (2.5,1)--(3.5,2);
        \node[left] at (0,0) {$l_2$};
        \node[left] at (0,2) {$l_1$};
        \node[right] at (3.5,0) {$l_4$};
        \node[right] at (3.5,2) {$l_3$};
        \fill (1,1) circle[radius=2pt];
        \fill (2.5,1) circle[radius=2pt];
        \fill (0,0) circle[radius=2pt];
        \fill (0,2) circle[radius=2pt];
        \fill (3.5,0) circle[radius=2pt];
        \fill (3.5,2) circle[radius=2pt];
    \end{tikzpicture}
    \end{minipage}
    \begin{minipage}{0.3\textwidth}
    \centering
    \begin{tikzpicture}
      \draw (0,0)--(1,1);
        \draw (0,2)--(1,1);
        \draw (1,1)--(2.5,1);
        \draw (2.5,1)--(3.5,0);
        \draw (2.5,1)--(3.5,2);
        \node[left] at (0,0) {$l_3$};
        \node[left] at (0,2) {$l_1$};
        \node[right] at (3.5,0) {$l_4$};
        \node[right] at (3.5,2) {$l_2$};
        \fill (1,1) circle[radius=2pt];
        \fill (2.5,1) circle[radius=2pt];
         \fill (0,0) circle[radius=2pt];
        \fill (0,2) circle[radius=2pt];
        \fill (3.5,0) circle[radius=2pt];
        \fill (3.5,2) circle[radius=2pt];
        \end{tikzpicture}
        \end{minipage}
        \begin{minipage}{0.3\textwidth}
        \centering
         \begin{tikzpicture}
         \draw (0,0)--(1,1);
        \draw (0,2)--(1,1);
        \draw (1,1)--(2.5,1);
        \draw (2.5,1)--(3.5,0);
        \draw (2.5,1)--(3.5,2);
        \node[left] at (0,0) {$l_4$};
        \node[left] at (0,2) {$l_1$};
        \node[right] at (3.5,0) {$l_3$};
        \node[right] at (3.5,2) {$l_2$};
        \fill (1,1) circle[radius=2pt];
        \fill (2.5,1) circle[radius=2pt];
         \fill (0,0) circle[radius=2pt];
        \fill (0,2) circle[radius=2pt];
        \fill (3.5,0) circle[radius=2pt];
        \fill (3.5,2) circle[radius=2pt];
    \end{tikzpicture}
        \end{minipage}
        \caption{\label{fig:quartets} Unrooted trivalent trees depicting $T_{12}$, $T_{13}$, and $T_{14}$, from left to right.}
        \label{fig:RunningExample}
\end{figure}
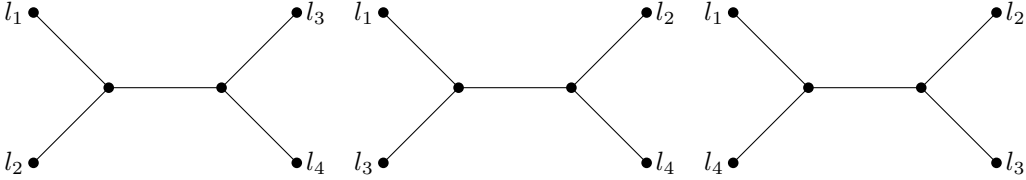

\begin{prop}\label{prop:dim}
Let $T$ be a trivalent phylogenetic tree evolving under a submodel $\MM$ of TN93 with fixed root distribution $\pi$. Then the dimension of the phylogenetic variety $CV_T^\MM$ is $\vert E\vert+1$ if $\MM=$F81 and $2\vert E\vert+1$ if $\MM=$F84.
\end{prop}

\begin{proof}
Let us restrict $\varphi_T^\MM$ to matrices $\Lambda^e$ with $\lambda_1^e=1$, which comprises all Markov matrices. 
We denote by $V_T^\MM$ the Zariski closure of the image of this restriction. 
By \cite{chang1996}, this restricted map has finite fibers {on distributions that satisfy certain generic conditions}. {These conditions are also generic for all models $\MM$} and hence the dimension of $V_T^\MM$ is that of the parameter space. 
If $\pi$ is fixed, each transition matrix has one free parameter for F81 (so $\dim V_T^{F81}=\vert E\vert$) and two free parameters for F84 ($\dim V_T^{F84}=2\vert E\vert$). 
By \cite{CHT} we know that $CV_T^\MM$ is the affine cone over $V_T^\MM$, so $\dim CV_T^\MM=\dim V_T^\MM+1$. 
\end{proof}

When all transition matrices are the identity, we call the corresponding tensor $\rho$ the \emph{no-evolution point} and note that $\bar{\rho} = \varphi_T(\{Id\}_{e\in E})$. Observe that this point does not depend on the substitution model (so we omit the superscript $\MM$). 

If we have a tensor $\bar{p} = \varphi_T^\MM(\{\Lambda^e\}_{e\in E})$, then we can consider the tensor
$$\bar{q}= \varphi_T^\MM(\{Id^{e_i}\}_{i\in [n]};\{\Lambda^e\}_{e\in E\backslash\{e_1,\dots,e_n\}})$$
which has the identity as the transition matrix at the pendant edges $e_i$. Then we have 
$\bar{p}=(\Lambda^{e_1},\dots,\Lambda^{e_n})\cdot \bar{q}$, that is, 
\begin{equation}\label{eq:eigenvalues_p_q}
\bar{p}_{i_1\dots i_n}=\lambda^{e_1}_{i_1}\dots\lambda^{e_n}_{i_n}\bar{q}_{i_1\dots i_n}.
\end{equation}

\noindent In order to obtain a parametrization that is as close as possible to being monomial and monic, we rescale the coordinates via the no-evolution point as
\begin{equation}\label{eq:tildes}
\tilde{p}_{i_1\dots i_n}=\frac{\bar{p}_{i_1\dots i_n}}{\bar{\rho}_{i_1\dots i_n}} 
\end{equation}
\noindent
if $\bar{\rho}_{i_1\dots i_n}\neq 0$ and $\tilde{p}_{i_1\dots i_n}=\bar{p}_{i_1\dots i_n}$ otherwise. The analog of \eqref{eq:eigenvalues_p_q} also holds for coordinates $\tilde{p}$:

\begin{equation}\label{eq:eigenvalues_p_q_tilde}
\tilde{p}_{i_1\dots i_n}=\lambda^{e_1}_{i_1}\dots\lambda^{e_n}_{i_n}\tilde{q}_{i_1\dots i_n}.
\end{equation}

\begin{remark}\label{tilde_rho}\rm
After rescaling, all non-zero coordinates of the no-evolution point are $\tilde{\rho}_{i_1\dots i_n}=1$. Therefore, for star trees with any number of leaves, all non-zero coordinates $\tilde{p}_{i_1\dots i_n}=\lambda^{e_1}_{i_1}\dots\lambda^{e_n}_{i_n}$ are monic and monomial.  
This includes the case of tripods; see \Cref{sec:tripods}.
\end{remark}

\subsection{Flattening matrices}

One of the main techniques for computing phylogenetic invariants of the general Markov model (and its submodels) consists of obtaining rank constraints in flattenings of the joint probability tensor corresponding to bipartitions of the set of leaves.  

\begin{definition}\label{def:flat} Let $A|B$ be a bipartition of the set of leaves $L$. Assume that leaves are ordered so that $A=\{1,\dots,m\}$ and $B=\{m+1,\dots,n\}$. If $p\in \otimes^n \mathbb{C}^4$, the \emph{flattening of $p$ according to the bipartition} $A|B$ is the $4^{|A|}\times 4^{|B|}$ matrix $\Flat_{A|B}(p)$ whose $(i_1,\dots, i_m,i_{m+1},\dots,i_n)$ entry is $p_{i_1\dots i_n}.$ 
For any other order of the set of leaves, $\Flat_{A|B}(p)$ is defined analogously.
\end{definition}

The following result by Allman and Rhodes is key in algebraic phylogenetics.

\begin{thm}[\cite{Allman2008}]\label{thm:ARflat}
    Let $T$ be a phylogenetic tree with leaf set $L$ and let $p\in \im \varphi_T^{\MM}$ be a joint probability tensor as in Definition \ref{def:param}. If $A|B$ is a bipartition of $L$ induced by removing an edge of $T$, then $\rk \left(\Flat_{A|B}(p)\right)\leq 4$. Otherwise, if $A|B$ cannot be induced by removing any edge of $T$, then $\rk\left( \Flat_{A|B}(p)\right)>4$ if the parameters that gave rise to $p$ were sufficiently general.   
\end{thm}

\subsection{The space of mixtures}\label{sec:prelim_mixt}

This section uses the notation introduced in \cite[Section 5]{CasSteel}. As noted in \Cref{def:param}, $\psi_T^\MM$ is the map that sends a collection of parameters $\Theta=(M^e)_{e\in E}$ in a model $\MM$ to the joint probability tensor $p^T$ in $\otimes^n \CC^4$. 

\begin{definition}
    We define the \emph{space of mixtures on $T$} for $\MM$ to be
$$\mathcal{D}_T^\MM=\left\{p=\sum_{i}\lambda_i \psi_{T}^\MM(\Theta_i) \,\Big| \, \sum_i \lambda_i=1\right\}.$$
If $\mathcal{T}_n$ is the set of phylogenetic trees on a leaf set $L$ of cardinality $n$, we define \emph{the space of phylogenetic mixtures} on $L$ as
$$\mathcal{D}_n^\MM=\left\{p=\sum_{i}\lambda_i \psi_{T_i}^\MM(\Theta_i)  \,\Big| \, \sum_i \lambda_i=1 \, , \, T_i \in \mathcal{T}_n \right\}.$$
This space contains all distributions that are a mixture of joint distributions on trees on $L$ and, in particular, it contains all distributions that arise as Markov processes on networks on $L$.
\end{definition}

Spaces of phylogenetic mixtures have been studied in detail for equivariant models \cite{CFK} and the equal input model \cite{CasSteel} (which generalizes F81 to any number of states). Beyond theoretical interest, spaces of mixtures are appealing objects of study because of their applications to model selection; see \cite{KDGC}.

As proved in \cite[Section 5]{CasSteel}, these are affine linear varieties that lie inside the hyperplane
\[H=\left\{p \ \left| \ \sum_{x_1\dots x_n} p_{x_1\dots x_n} =1\right\}\right..\] In the coordinates $\bar{p}$ introduced above, this hyperplane is defined by the equation $\bar{p}_{1\dots 1}=1$. Indeed, if $\textbf{1}$ is the row vector $(1,1,1,1)$, then by \eqref{eq:coord_change} we have $\bar{p}_{1\dots 1}=(\textbf{1}\otimes \textbf{1}\dots\otimes \textbf{1}) p$, which is precisely $\sum_{x_1\dots x_n} p_{x_1\dots x_n}$.

We can view the linear varieties introduced above as the corresponding linear subspaces of $\otimes^n \CC^4$ cut by $H$. 
We do this by introducing the following notation: 
let $\mathcal{E}^\MM_n$ be the subspace of $\otimes^n \CC^4$ where all linear homogeneous model invariants vanish, and let $\mathcal{E}_T^\MM$ be the subspace of $\mathcal{E}^\MM_n$ where all linear homogeneous phylogenetic invariants for a tree $T$ vanish. 
Note that $\mathcal{E}_T^\MM$ is the smallest linear space containing $CV_T$, i.e. $\mathcal{E}_T^\MM=\langle CV_T^\MM \rangle$, and $\E^\MM_n=\langle \bigcup_{T\in\mathcal{T}_n}CV_T^\MM \rangle$.
Then,

\[\mathcal{E}_T^\MM=\langle p^T=\psi_T^\MM(\Theta)\rangle, \quad \mathcal{E}^\MM_n=\langle p^T=\psi_T^\MM(\Theta)\mid\, T \in \mathcal{T}_n\rangle,  \]
\[\mathcal{D}_T^\MM=\mathcal{E}_T^\MM\cap H, \quad \mathcal{D}^\MM_n=\mathcal{E}^\MM_n\cap H.\]


\noindent
{Throughout the paper, we will abuse language and refer to $\E_T^\MM$ and $\E_n^\MM$ as spaces of (phylogenetic) mixtures.}

Although in \cite{CasSteel} these spaces were studied for F81, equations cutting out $\mathcal{E}^\MM_n$ and $\mathcal{E}_T^\MM$ were difficult to derive, and the given ones had long expressions. We will provide simpler equations for F81 and F84 using the change of coordinates from \Cref{subsec:basis}.

Finally, we focus on two special types of linear equations cutting out $\E^\MM$ in coordinates  $\tilde{p}$. We call \emph{model zeros} those model invariants of the form $\tilde{p}_{i_1..i_n}=0$. The next simplest class of model invariants are equalities among coordinates of $p$ that hold regardless of the tree topology. We call them \emph{symmetry equations} as they are analogous to equations for equivariant models that arise from the action of a group \cite{CFK, CFM}. If $\mathcal{F}$ is the set of symmetry equations and $\tilde{p}_{i_1\dots i_n}=\tilde{p}_{j_1\dots j_n}$ is in $\mathcal{F}$, then $\tilde{p}_{\sigma(i_1\dots i_n)}=\tilde{p}_{\sigma(j_1\dots j_n)}$ must be in $\mathcal{F}$ for any $\sigma$ in the set $S_L$ of \emph{permutations} of the leaf set $L$. 

\begin{definition}\label{def:F}
    We call $\Z^\MM_n$ and $\F^\MM_n$ the linear subspaces of $\otimes^n \CC^4$ cut out by model zeros and symmetry equations of the model $\MM$, respectively.
\end{definition}

Note that $\E^\MM_n\subseteq \Z^\MM_n\cap\F^\MM_n$. In the case of TN93, it is proven in \cite[Proposition 5.11]{CHT} that $\E^\MM_n=\Z^\MM_n$ is only cut out by model zeros for $n=3,4$ (no symmetry equations exist beyond those derived from model zeros) and conjectured that this holds for any $n$.

\section{The tripod and symmetry equations}\label{sec:tripods}

Let $T$ be the tripod with leaves $L=\{l_1, l_2,l_3\}$ and edges $E=\{e_1, e_2, e_3\}$. This section is devoted to the computation of equations that cut out the phylogenetic variety $CV_T^\MM$ for substitution models F81 and F84. 

Consider the no-evolution point $\bar{\rho} = \varphi(Id, Id, Id)$. After the rescaling described in \eqref{eq:tildes},  by \cite[Lemma 5.1]{CHT}, its coordinates in the $\pi$-orthogonal basis $B$ defined in \eqref{eq:basis} are
\begin{equation}\label{eq:tripod_rho}
\tilde{\rho}_{111}=\tilde{\rho}_{222}=\tilde\rho_{333}=\tilde{\rho}_{444}=\tilde{\rho}_{\sigma(122)}=\tilde{\rho}_{\sigma(133)}=\tilde{\rho}_{\sigma(144)}=\tilde{\rho}_{\sigma(233)}=\tilde{\rho}_{\sigma(244)}=1
\end{equation}

\noindent
for any $\sigma\in S_L$ and $\tilde{\rho}_{i_1i_2i_3}=0$ otherwise.

\begin{remark}\label{rmk:Z3}\rm 
By \cite{CHT}, $\Z_3^{TN93}$ is cut out by the 45 equations of type $\tilde{p}_{i_1i_2i_3}=0$, if $\tilde{p}_{i_1i_2i_3}$ is not listed in \eqref{eq:tripod_rho}. Since the no-evolution point $\rho$ belongs to the phylogenetic variety $CV_T^\MM$, irrespective of the model, no additional model zeros arise in any submodel. Hence, $\Z_3^{F81}=\Z_3^{F84}=\Z_3^{TN93}$. Therefore, we just call this linear space $\Z_3$. 
\end{remark}

\begin{prop}\label{prop:eqtripodF84}
Let $T$ be a tripod and consider the F84 model on it. The following equations cut out the 12-dimensional linear space $\Z_3\cap\F^{F84}_3$:
\begin{enumerate}
    \item [(a)]
 7 symmetry equations:
    \begin{itemize}
    \item $\tilde{p}_{144}=\tilde{p}_{133}$, $\tilde{p}_{414}=\tilde{p}_{313}$, $\tilde{p}_{441}=\tilde{p}_{331}$;
    \item $\tilde{p}_{244}=\tilde{p}_{233}$, $\tilde{p}_{424}=\tilde{p}_{323}$, $\tilde{p}_{442}=\tilde{p}_{332}$;
    \item $\tilde{p}_{444}=\tilde{p}_{333}$.
\end{itemize}
\item[(b)] 45 model zeros $\tilde{p}_{i_1i_2i_3}=0$ (for each coordinate that is not listed in \eqref{eq:tripod_rho}).
\end{enumerate}
\end{prop}

\begin{proof}
By \Cref{tilde_rho} \Cref{rmk:Z3}, we know that $\tilde{p}_{ijk}=\lambda^{e_1}_i\lambda^{e_2}_j\lambda^{e_3}_k$ if this coordinate is listed in \eqref{eq:tripod_rho} and zero otherwise. The equalities in (a) follow from the fact that $\lambda^e_{3}=\lambda^e_{4}$ for any $e\in E$, and from the expression of the coordinates we see that there are no other possible symmetries. Hence, $\Z_3\cap\F^{F84}$ is a linear space of codimension $45+7=52$ in $\CC^{64}$ and it has dimension 12.
\end{proof}

By \Cref{prop:eqtripodF84}, we can think of the variety $CV_T^{F84}$ for the tripod to be living in the 12-dimensional affine space 
$\mathbb{C}\langle \tilde{p}_{111},\tilde{p}_{122},\tilde{p}_{212},\tilde{p}_{221},\tilde{p}_{144},\tilde{p}_{414},\tilde{p}_{441},\tilde{p}_{244},\tilde{p}_{424},\tilde{p}_{442},\tilde{p}_{222},\tilde{p}_{444}\rangle$. 
By \Cref{prop:dim}, $CV_T^{F84}$ is a 7-dimensional variety
and we have computed its defining equations below with \texttt{Macaulay2} \cite{M2}. These equations can also be calculated with theoretical tools from toric geometry, as in \cite[Proposition 5.3]{CHT}.

\begin{prop}
     If $T$ is a tripod and we consider the F84 model, then $CV_T^{F84}\subset \Z_3\cap\F^{F84}$ is a 7-dimensional variety of degree 27 defined by the following 3 quadrics and 12 cubics:
    \begin{gather*}
    \tilde{p}_{222}\tilde{p}_{441}-\tilde{p}_{221}\tilde{p}_{442},\quad \tilde{p}_{222}\tilde{p}_{414}-\tilde{p}_{212}\tilde{p}_{424},
      \quad \tilde{p}_{144}\tilde{p}_{222}-\tilde{p}_{122}\tilde{p}_{244}, \\ 
      \tilde{p}_{244}\tilde{p}_{424}\tilde{p}_{442}-\tilde{p}_{222}\tilde{p}_{444}^2, \quad
      \tilde{p}_{144}\tilde{p}_{424}\tilde{p}_{442}-\tilde{p}_{122}\tilde{p}_{444}^2,\quad
       \tilde{p}_{244}\tilde{p}_{414}\tilde{p}_{442}-\tilde{p}_{212}\tilde{p}_{444}^2,\quad \\
       \tilde{p}_{244}\tilde{p}_{424}\tilde{p}_{441}-\tilde{p}_{221}\tilde{p}_{444}^2,\quad \tilde{p}_{144}\tilde{p}_{414}\tilde{p}_{441}-\tilde{p}_{111}\tilde{p}_{444}^2,\quad
       \tilde{p}_{122}\tilde{p}_{414}\tilde{p}_{441}-\tilde{p}_{111}\tilde{p}_{424}\tilde{p}_{442},\\
       \tilde{p}_{144}\tilde{p}_{212}\tilde{p}_{441}-\tilde{p}_{111}\tilde{p}_{244}\tilde{p}_{442},
       \quad
       \tilde{p}_{122}\tilde{p}_{212}\tilde{p}_{441}-\tilde{p}_{111}\tilde{p}_{222}\tilde{p}_{442},\quad
       \tilde{p}_{144}\tilde{p}_{221}\tilde{p}_{414}-\tilde{p}_{111}\tilde{p}_{244}\tilde{p}_{424}, \\
       \tilde{p}_{122}\tilde{p}_{221}\tilde{p}_{414}-\tilde{p}_{111}\tilde{p}_{222}\tilde{p}_{424},\quad
       \tilde{p}_{144}\tilde{p}_{212}\tilde{p}_{221}-\tilde{p}_{111}\tilde{p}_{222}\tilde{p}_{244},\quad
       \tilde{p}_{122}\tilde{p}_{212}\tilde{p}_{221}-\tilde{p}_{111}\tilde{p}_{222}^2.
    \end{gather*}
\end{prop}


\begin{prop}\label{prop:eqtripod}
Let $T$ be a tripod and consider the F81 model on it. The algebraic variety $CV_T^{F81}\subset\CC^{64}$ is a 4-dimensional variety of degree 3 cut out by the following  equations: 
\begin{enumerate}
    \item [(a)]
 14 symmetry equations
    \begin{itemize}
    \item $\tilde{p}_{144}=\tilde{p}_{133}=\tilde{p}_{122}$;
    \item $\tilde{p}_{414}=\tilde{p}_{313}=\tilde{p}_{212}$;
    \item $\tilde{p}_{441}=\tilde{p}_{331}=\tilde{p}_{221}$;
    \item $\tilde{p}_{244}=\tilde{p}_{424}=\tilde{p}_{442}=\tilde{p}_{233}=\tilde{p}_{323}=\tilde{p}_{332}=\tilde{p}_{444}=\tilde{p}_{333}=\tilde{p}_{222}$.
\end{itemize}
\item[(b)] 45 equations of type $\tilde{p}_{i_1i_2i_3}=0$ (for each coordinate that is not listed in \eqref{eq:tripod_rho}.
\item[(c)] $\tilde{p}_{144}\tilde{p}_{414}\tilde{p}_{441}-\tilde{p}_{111}\tilde{p}_{444}^2=0.$ 
\end{enumerate}
\end{prop}

\begin{proof}

By \Cref{tilde_rho}, we know that $\tilde{p}_{ijk}=\lambda^{e_1}_i\lambda^{e_2}_j\lambda^{e_3}_k$ for the indices listed in \eqref{eq:tripod_rho} and zero otherwise. The desired equalities follow from the fact that $\lambda^e_{2}=\lambda^e_{3}=\lambda^e_{4}$ for any $e\in E$.

From these equalities, we consider the variety $CV_T$ for the tripod to be living in the affine space 
$\mathbb{C}^5=\mathbb{C}\langle\tilde{p}_{111},\tilde{p}_{122},\tilde{p}_{212},\tilde{p}_{221},\tilde{p}_{222}\rangle$ (cut out by the linear equations above). 
By \Cref{prop:dim}, the dimension of $CV_T$ is $4$, and we need to check that equation (c) holds. It is straightforward to verify that the subtracted terms are equal to $\lambda^{e_1}_1\lambda^{e_2}_1\lambda^{e_3}_1(\lambda^{e_1}_4\lambda^{e_2}_4\lambda^{e_3}_4)^2.$ 
\end{proof}

\begin{remark}
    \rm Note the analogy of the F81 model with the Jukes-Cantor model on the tripod: as proven in \cite[\S 4.5]{ASCB2005} the variety for the Jukes-Cantor model is also a cubic hypersurface inside its ambient space (and, in Fourier coordinates, the cubic equation coincides with the one given above).
\end{remark}

\section{Symmetry equations for quartets}\label{sec:symmetry}

Now we consider quartets on the set of leaves $L=\{l_1, l_2,l_3,l_4\}.$ This section focuses on studying the spaces $\Z_4^\MM$ and $\F^\MM_4$ defined in \Cref{def:F} of model zeros and symmetry equations for models F81 and F84. 

We start by describing the coordinates in the basis $B$ of the no-evolution point $\rho=\psi(Id,Id,Id,Id;Id)$ (which lies in $CV_T^\MM$ for any quartet and any model). This is the list of the non-zero coordinates up to permutations of indices: 
\begin{equation}\label{eq:rho_quartet}
\begin{array}{cccc}
\pp_{1111}=1
    &
    \pp_{1122}=\frac{1}{\pi_{12}\pi_{34}}
    &
    \pp_{1133}=\frac{\pi_{34}}{\pi_3\pi_4}
    &
    \pp_{1144}=\frac{\pi_{12}}{\pi_1\pi_2}
    \\
    \pp_{1222}=\frac{\pi_{34}-\pi_{12}}{\pi_{12}^2\pi_{34}^2}
    & \pp_{1233}=-\frac{1}{\pi_3\pi_4}
    &
    \pp_{1244}=\frac{1}{\pi_1\pi_2}
   &\pp_{1333}=\frac{\pi_{34}(\pi_4-\pi_3)}{\pi_3^2\pi_4^2}
    \\    \pp_{1444}=\frac{\pi_{12}(\pi_2-\pi_1)}{\pi_1^2\pi_2^2}
    &
    \pp_{2222}=
\frac{\pi_{12}^3+\pi_{34}^3}{\pi_{12}^3\pi_{34}^3}
    & \pp_{2233}=
\frac{1}{\pi_3\pi_4 \pi_{34}}
   & \pp_{2244}=
    \frac{1}{\pi_1\pi_2 \pi_{12}}
    \\
    \pp_{2333}=\frac{\pi_3-\pi_4}{\pi_3^2\pi_4^2}
 &\pp_{2444}=\frac{\pi_2-\pi_1}{\pi_1^2\pi_2^2}
 &
 \pp_{3333}
    =\frac{\pi_3^3+\pi_4^3}{\pi_3^3\pi_4^3}
   &
    \pp_{4444}=
    \frac{\pi_1^3+\pi_2^3}{\pi_1^3\pi_2^3}.
\end{array}
\end{equation}

\begin{prop}\label{prop:quartet-zeroes}
    If $\MM$ is F84 or F81, the following 172 linear equations are the only model zeros for any quartet evolving under $\MM$:
    \begin{itemize}
        \item [(a)] ${\tilde{p}}_{i_1i_2i_3i_4} = 0$ if exactly one $i_j=3$ or $4$, for $j\in\{1,2,3,4\}$;
        \item [(b)] $\tilde{p}_{\sigma(1112)}=0$ for any $\sigma\in S_L$.
    \end{itemize}
\end{prop}
\begin{proof}
    Since these equations are model zeros of TN93 by \cite[Proposition 5.11]{CHT}, they are also model zeros of $\MM$. 
    Recall that we are assuming $\pi$ generic (\Cref{rmk:generic_pi}). Therefore,
    all index configurations appearing in \eqref{eq:rho_quartet} yield non-zero coordinates $\tilde{\rho}_{i_1i_2i_3i_4}$ of the no-evolution point for strictly positive $\pi$ 
    (so these coordinates cannot give rise to model zeros). 
    There are exactly six coordinates which are zero for the no-evolution point. These are not listed in the statement but are all index permutations of $3344$. For each of these permutations, we can find a tree topology $T$ for which the expression of such a coordinate is non-zero for a point $q=\psi_T^\MM (Id,Id,Id,Id,M)$. For instance, for the coordinate $\tilde{p}_{3344}$ we can consider the tree $T=T_{12}$ so that $\tilde{q}_{3344}=(\pi_{12}\pi_{34}/\pi_1\pi_2\pi_3\pi_4)(\lambda_1-\lambda_2)$ for TN93 (see \cite[Lemma SM2.1]{CHT}). In both F81 and F84, $\lambda_1$ and $\lambda_2$ are not generically equal. Hence, $\tilde{q}_{3344}$ is not a model zero. 
\end{proof}

\begin{remark}\label{rmk:Z4}\rm
Similarly to \Cref{rmk:Z3}, $\Z^{TN93}_4=\Z^{F84}_4=\Z^{F81}_4$ and we call this space just $\Z_4$.
\end{remark}

The remainder of this section will be devoted to finding symmetry equations. To this aim, we first compute equalities among coordinates $\tilde{q}_{i_1i_2i_3i_4}$ of points of the form 
\[q=\psi_T^\MM(Id,Id,Id,Id;M)\]
for a quartet tree $T$. As $\tilde{p}_{i_1i_2i_3i_4}=\lambda^{e_1}_{i_1}\lambda^{e_2}_{i_2}\lambda^{e_2}_{i_3}\lambda^{e_4}_{i_4}\tilde{q}_{i_1i_2i_3i_4}$, equalities among coordinates $\tilde{q}_{i_1i_2i_3i_4}$ translate into equalities among $\tilde{p}_{i_1i_2i_3i_4}$ as long as certain index positions are preserved. 
More precisely, we can define a multi-degree as follows.

\begin{definition}\label{def:multidegree}
     We define the \emph{multi-degree} of $\tilde{p}_{i_1i_2i_3i_4}$ as the 4-tuple $(m_1,m_2,m_3,m_4)$ in $\ZZ^4$, where $m_j=\max\{k\in[4]:\lambda_k=\lambda_{i_j}\}$, and we denote it by $\mathrm{mdeg}(\tilde{p}_{i_1i_2i_3i_4})$.
\end{definition}

Note that this definition depends on the model. For example, for $\tilde{p}_{3124}$ if we consider F84 we have  $\mathrm{mdeg}(\tilde{p}_{3124})=(4,1,2,4)$ but if we are considering F81 then $\mathrm{mdeg}(\tilde{p}_{3124})=(4,1,4,4)$. Then, if we have an equation of type $\tilde{p}_{i_1i_2i_3i_4}=\tilde{p}_{j_1j_2j_3j_4}$, it must also hold for any point $q$ and it must be homogeneous with respect to this multidegree. The converse also holds, as stated in the following lemma.

\begin{lemma}\label{lemma:q_to_p}
Let $\tilde{q}_{i_1i_2i_3i_4}=\tilde{q}_{j_1j_2j_3j_4}$ for any point $q\in CV_T^{\mathcal{M}}$ of type $q=\psi_T^\MM(Id,Id,Id,Id;M)$. Then, $\tilde{p}_{i_1i_2i_3i_4}-\tilde{p}_{j_1j_2j_3j_4}$ is in the ideal of $CV_T^{\mathcal{M}}$ if and only if $\mathrm{mdeg}(\tilde{p}_{i_1i_2i_3i_4})=\mathrm{mdeg}(\tilde{p}_{j_1j_2j_3j_4})$.
\end{lemma}


In the proof of \Cref{prop:quartet-eqF84} below we provide an explicit description of the point $q$ for $T=T_{12}$ for model F84 (and the coordinates of the corresponding points for $T_{13}$ and $T_{14}$ can be obtained by permuting the indices according to the corresponding permutations of the leaves). 
Note that given a tree $T$, those permutations that swap elements within a cherry or swap cherries do not modify the corresponding coordinates. This suggests the following notation. 
\begin{definition}
    The subgroup $S_{L_{12}}:=\langle (12),(34),(13)(24),(14)(23)\rangle$ of $S_L$ is the set of \emph{compatible permutations with tree topology $T_{12}$}. We define it analogously for $T_{13}$ and $T_{14}$.
\end{definition}
\noindent For $T=T_{kl}$, we have $\tilde{q}_{\sigma(i_1i_2i_3i_4)}=\tilde{q}_{i_1i_2i_3i_4}$ for any $\sigma\in S_{L_{kl}}.$

We are now in the position to describe the spaces, $\F_4^{\MM}$, of symmetry equations for quartets for models F84 and F81.

\begin{prop}\label{prop:quartet-eqF84} 
The linear space, $\F_4^{F84}$, of symmetry equations for quartets is cut out by the following 29 symmetry equations:
\begin{itemize}
    \item[(a)] 22 equations with 1's in the indices:
    \begin{itemize}
        \item[(a.1)] 6 equations: $\tilde{p}_{\sigma(1144)}=\tilde{p}_{\sigma(1133)}$ for all $\sigma\in S_L$,
        \item[(a.2)] 12 equations: $\tilde{p}_{\sigma(1244)}=\tilde{p}_{\sigma(1233)}$ for all $\sigma\in S_L$,
        \item[(a.3)] 4 equations: $\tilde{p}_{\sigma(1444)}=\tilde{p}_{\sigma(1333)}$ for all $\sigma\in S_L$;  
    \end{itemize}
    \item[(b)] Four equations with "triples" of 3's and 4's:
    $\tilde{p}_{\sigma(2444)}=\tilde{p}_{\sigma(2333)}$ for all $\sigma\in S_L$;  
    \item[(c)] Three equations with pairs of 3's and 4's: 
    $\tilde{p}_{3344}=\tilde{p}_{4433}$, $\tilde{p}_{3434}=\tilde{p}_{4343}$, $\tilde{p}_{3443}=\tilde{p}_{4334}$. 
\end{itemize}

\end{prop}

\begin{proof} 
According to \Cref{lemma:q_to_p}, any symmetry equation must be homogeneous with respect to the multi-degree for F84 introduced in \Cref{def:multidegree}. In addition, it must also hold for any $q=\psi_T^{F84}(Id,Id,Id,Id;M)$ for any quartet $T$.
Let $M$ have eigenvalues $\lambda_1$, $\lambda_2$, and $\lambda_4$ (with multiplicity 2). From \cite[SM1.2]{CHT} and \eqref{eq:rho_quartet}, substituting $\lambda_3$ by $\lambda_4$ yields the following expressions for $T=T_{12}$: 

\begin{center}
\begin{tabular}{ll}
\begin{tiny}(1)\end{tiny}  $\tilde{q}_{1111}=\lambda_1$ & \begin{tiny}(10)\end{tiny} $\tilde{q}_{\sigma(2333)}=\tilde{q}_{\sigma(2444)}=\lambda_4$ for $\sigma\in S_L$  \\  

\begin{tiny}(2)\end{tiny}  $\tilde{q}_{\sigma(1122)}=\lambda_1$ for $\sigma\in\SL$ & \begin{tiny}(11)\end{tiny} $\tilde{q}_{\sigma(2323)}=\tilde{q}_{\sigma(2424)}=\lambda_4$ for $\sigma\in\SL$  \\

\begin{tiny}(3)\end{tiny} $\tilde{q}_{\sigma(1133)}=\tilde{q}_{\sigma(1144)}=\lambda_1$ for $\sigma\in\SL$ & \begin{tiny}(12)\end{tiny} $\tilde{q}_{2233}=\tilde{q}_{3322}=\frac{\pi_{34}}{\pi_{12}}\lambda_1+\left(1-\frac{\pi_{34}}{\pi_{12}}\right)\lambda_2$ \\

\begin{tiny}(4)\end{tiny} $\tilde{q}_{\sigma(1212)}=\lambda_2$ for $\sigma\in\SL$ & \begin{tiny}(13)\end{tiny} $\tilde{q}_{2244}=\tilde{q}_{4422}=\frac{\pi_{12}}{\pi_{34}}\lambda_1+\left(1-\frac{\pi_{12}}{\pi_{34}}\right)\lambda_2$ \\

\begin{tiny}(5)\end{tiny} $\tilde{q}_{\sigma(1222)}=\lambda_2$ for $\sigma\in S_L$ & \begin{tiny}(14)\end{tiny} $\tilde{q}_{3344}=\tilde{q}_{4433}=
    \frac{\pi_{12}\pi_{34}}{\pi_1\pi_2\pi_3\pi_4}(\lambda_1 -\lambda_2)$  \\

\begin{tiny}(6)\end{tiny} $\tilde{q}_{\sigma(1233)}=\tilde{q}_{\sigma(1244)}=\lambda_2$ for $\sigma\in\SL$ & \begin{tiny}(15)\end{tiny} $\tilde{q}_{2222}=
    \frac{\pi_{12}\pi_{34}}{\pi_{12}^3+\pi_{34}^3}\lambda_1 + \frac{(\pi_{12}-\pi_{34})^2}{\pi_{12}^3+\pi_{34}^3}\lambda_2$ \\

\begin{tiny}(7)\end{tiny} $\tilde{q}_{\sigma(1313)}=\tilde{q}_{\sigma(1414)}=\lambda_4$ for $\sigma\in\SL$ & \begin{tiny}(16)\end{tiny} $\tilde{q}_{3333}=\frac{\pi_3\pi_4\pi_{34}^2}{\pi_3^3+\pi_4^3}\lambda_1 + \frac{\pi_{12}\pi_{34}\pi_3\pi_4}{\pi_3^3+\pi_4^3}\lambda_2+\frac{\pi_{34}(\pi_3-\pi_4)^2}{\pi_3^3+\pi_4^3}\lambda_4$ \\

\begin{tiny}(8)\end{tiny} $\tilde{q}_{\sigma(1323)}=\tilde{q}_{\sigma(1424)}=\lambda_4$ for $\sigma\in\SL$ & \begin{tiny}(17)\end{tiny} $\tilde{q}_{4444}=\frac{\pi_1\pi_2\pi_{12}^2}{\pi_1^3+\pi_2^3}\lambda_1 + \frac{\pi_{12}\pi_{34}\pi_1\pi_2}{\pi_1^3+\pi_2^3}\lambda_2+\frac{\pi_{12}(\pi_1-\pi_2)^2}{\pi_1^3+\pi_2^3}\lambda_4$ \\

\begin{tiny}(9)\end{tiny} $\tilde{q}_{\sigma(1333)}=\tilde{q}_{\sigma(1444)}=\lambda_4$ for $\sigma\in S_L$ &  \\
\end{tabular}
\end{center}

\noindent and the other coordinates are zero. 

Note that equal coordinates of $\tilde{q}$ are listed separately, such that every entry only contains coordinates that share multi-degree up to permutation.
We need to check that (a),(b), and (c) are a complete list of homogeneous symmetry equations that hold for $q$. 
Indeed, from all equalities among coordinates 
listed above, only those in (3),(6)-(11), and (14) involve coordinates with the same multi-degree for $F84$. 
From these, equations (9) and (10) hold independently of the tree topology and hence yield symmetry equations (a.3) and (b). Although the rest depend on the tree topology, there exist some combinations that ensure that the equations hold for any $\sigma\in S_L$: (3) and (7) yield symmetry equations (a.1); and (6) and (8) yield symmetry equations (a.2). On the other hand, (11) does not produce symmetry equations because they do not hold for any quartet: there exists $\sigma\in S_L$ turning the coordinates of (11) into (12) and (13), which are not equal. Finally, there is only the special case of (14) to be dealt with. Precisely because $\tilde{p}_{3344}$ and $\tilde{p}_{4433}$ are topology zeros for $T_{13}$ and $T_{14}$, we have that $\tilde{q}_{\sigma(3344)}=\tilde{q}_{\sigma(4433)}$ holds for all $\sigma\in S_L$. This proves (c).
\end{proof}

\begin{prop}\label{prop:quartet-eq}
The linear space of symmetry equations for quartets $\F_4^{F81}$ is cut out by the following 60 symmetry equations:

\begin{itemize}
    \item[(a)]  44 equations with 1's in the indices:
    \begin{itemize}
        \item[\textbullet] 12 equations: $\,\,\tilde{p}_{\sigma(1144)}=\tilde{p}_{\sigma(1133)}=\tilde{p}_{\sigma(1122)}$ for all $\sigma\in S_L$,
        \item[\textbullet] 32 equations: $\tilde{p}_{\sigma(1222)}=\tilde{p}_{\sigma(1233)}=\tilde{p}_{\sigma(1244)}=\tilde{p}_{\sigma(1333)}=\tilde{p}_{\sigma(1444)}=\tilde{p}_{\sigma(1332)}=\tilde{p}_{\sigma(1323)}=\tilde{p}_{\sigma(1442)}=\tilde{p}_{\sigma(1424)}$ for all $\sigma\in S_L$;
    \end{itemize}

    \item[(b)] 7 equations with "triples" of 3's and 4's: 
    \begin{itemize}
        \item[\textbullet]$\tilde{p}_{2333}=\tilde{p}_{3233}=\tilde{p}_{3323}=\tilde{p}_{3332}=\tilde{p}_{2444}=\tilde{p}_{4244}=\tilde{p}_{4424}=\tilde{p}_{4442}$ 
    \end{itemize}
    \item[(c)] 9 equations for pairs of 2's, 3's and 4's: 
    \begin{itemize}
        \item[\textbullet] $\tilde{p}_{3344}=\tilde{p}_{4433}$, $\tilde{p}_{3434}=\tilde{p}_{4343}$, $\tilde{p}_{3443}=\tilde{p}_{4334}$; 
    \item[\textbullet]
    $\tilde{p}_{2233}=\tilde{p}_{3322}$, $\tilde{p}_{2323}=\tilde{p}_{3232}$, $\tilde{p}_{2332}=\tilde{p}_{3223}$;
 \item[\textbullet] $\tilde{p}_{2244}=\tilde{p}_{4422}$, $\tilde{p}_{2424}=\tilde{p}_{4242}$, $\tilde{p}_{2442}=\tilde{p}_{4224}$.
    \end{itemize}
    \end{itemize}
\end{prop}

\begin{proof}
Analogous to \Cref{prop:quartet-eqF84}, after further substituting $\lambda_2$ by $\lambda_4$.
\end{proof}

\section{Equations derived from rank conditions on flattenings}\label{sec:rank}

In this section, we consider the tree $T_{12}=l_1l_2|l_3l_4$ but all the results
are analogous for the other quartets $T_{13}$ and $T_{14}$ after permuting the indices of coordinates $\tilde{p}_{i_1i_2i_3i_4}$ accordingly.

Let us denote by $A$ the $16\times 16$ flattening matrix $\Flat_{12\mid 34}(\bar{p})$. By \Cref{thm:ARflat}, we know that $\rk(A)\leq 4$ for tensors arising on $T=l_1l_2|l_3l_4$. But because $\MM$ is a submodel of TN93 and we are using the same $\pi$-orthogonal basis as in \cite{CHT}, we know much more than that: we have blocks of ranks 1, 2, and 3, see \cite[Lemma 5.13 and Appendix B]{CHT}. 
Although these rank conditions on flattenings would give rise to equations of degrees two, three, and four, below we prove that we can obtain lower degree equations by adding the symmetry equations already found in the previous section. We are especially interested in linear equations because of the definition of the space of mixtures. Therefore, the rest of the section is devoted to deriving linear equations from rank conditions of flattenings and symmetry equations.

In Tables \ref{tab:flat} and \ref{tab:flatF84}, we display the matrix $A$ for F81 and F84, respectively, with its entries expressed in coordinates $\tilde{p}$ and the no-evolution point $\rho$; see \eqref{eq:tildes} and \eqref{eq:rho_quartet}. These entries have been obtained from the flattening matrix under the TN93 model \cite[Table 2]{CHT} by imposing the additional $60$ symmetries in \Cref{prop:quartet-eq} for F81 and the $29$ symmetries in \Cref{prop:quartet-eqF84} for F84. In Appendices \ref{app:flat} and \ref{app:flatF84}, we display these flattening matrices with a detailed description of which blocks must have rank 1, 2, or 3 according to the results of \cite{CHT}. 

\textbf{Notation.} Throughout the rest of the paper, we will use the following notation: $A^{ij,\dots,kl}_{ab,\dots,cd}$ denotes the submatrix of $A$ consisting of rows labeled $(a,b),\dots,(c,d)$ and columns $(i,j),\dots,(k,l)$.

\begin{remark}\label{rmk:primeideal}\rm
Note that $CV_T^\MM$ is an irreducible variety, and hence its vanishing ideal $I^\MM_T$ is prime. Therefore, when computing phylogenetic invariants, we can factor out all variables $\tilde{p}_{i_1i_2i_3i_4}$ that do not belong to the ideal. In the case of $T_{12}$, those variables are the ones not listed in \Cref{prop:quartet-zeroes} except for $\tilde{p}_{3434},\tilde{p}_{4343},\tilde{p}_{3443},\tilde{p}_{4334}$ (see \cite[Remark 5.12]{CHT}).
\end{remark} 

\subsection{Rank conditions for F84}

\begin{prop}\label{prop:4quartetEq84}
Let $T=T_{12}$ be a quartet and consider the F84 model on it. The following equations hold for any point in $CV_T \subseteq \F_4^{F84}\cap\Z_4$: 
\begin{itemize}
    \item[(a)] $\tilde{p}_{3434}=\tilde{p}_{3443}=0$;
    \item[(b)] $\tilde{p}_{2424}=\tilde{p}_{2323}$, $\tilde{p}_{2442}=\tilde{p}_{2332}$, $\tilde{p}_{4224}=\tilde{p}_{3223}$, $\tilde{p}_{4242}=\tilde{p}_{3232}$.
\end{itemize}
\end{prop}  
\begin{proof}
Since $\tilde{p}_{3434}=\tilde{p}_{3443}=0$ are topology invariants for $T_{12}$ in the TN93 model, equations in (a) hold for F84. We now consider submatrices $A^{14,24}_{14,24}$ and $A^{13,23}_{13,23}$ of rank 1 in Table \ref{tab:flatF84rk1}. Then 

\begin{eqnarray*}
0&=&\bar{\rho}_{1233}^2\left|A^{14,24}_{14,24}\right|-\bar{\rho}_{1244}^2\left|A^{13,23}_{13,23}\right|
    =\bar{\rho}_{1233}^2\bar{\rho}_{1144}\bar{\rho}_{2244}\tilde{p}_{1414}\tilde{p}_{2424}-\bar{\rho}_{1244}^2\bar{\rho}_{1133}\bar{\rho}_{2233}\tilde{p}_{1414}\tilde{p}_{2323}\\
    &=&\frac{1}{\left(\pi_1\pi_2\pi_3\pi_4\right)^2}\tilde{p}_{1414}\left(\tilde{p}_{2424}-\tilde{p}_{2323}\right),
\end{eqnarray*}

\noindent where the last equality is obtained by substituting the coordinates of the no-evolution point in \eqref{eq:rho_quartet}.
By \Cref{rmk:primeideal} we can factor $\tilde{p}_{1414}$ out, and hence $\tilde{p}_{2424}=\tilde{p}_{2323}$ in $CV_T^{F84}$.
The rest of equations in (b) can be proven analogously by considering submatrices $A^{1j,j2}_{1j,2j}$, $A^{1j,2j}_{1j,j2}$ and $A^{1j,j2}_{1j,j2}$, respectively, with $j\in\{3,4\}$. 
\end{proof}

\begin{remark}\label{rem:nonzero_F84}
\rm Note that the equations in \Cref{prop:4quartetEq84} have been obtained by imposing only the rank conditions on flattenings and the symmetry equations defining $\F_4^{F84}$. The rest of the proof only uses that $\tilde{p}_{1414}$ is different from zero. As a consequence of Corollary \ref{cor:mixt_tree}, 
 the ideal of $CV_T^{F84}$ is generated in degree 1 by the linear polynomials that correspond to the above equations.
\end{remark}

\subsection{Rank conditions for F81}
\begin{prop}\label{prop:4quartetEq}
Let $T=T_{12}$ be a quartet and consider the F81 model on it. The following equations hold for any point in $CV_T \subseteq \CC^{64}$: 
$$ \tilde{p}_{3434}=0=\tilde{p}_{3443},\, \tilde{p}_{2424}=\tilde{p}_{2442}=\,\tilde{p}_{2332}=\tilde{p}_{2323}=
\tilde{p}_{2444}.$$
 
\end{prop}  

\begin{proof}
For the chosen split topology, $\bar{p}_{3434}=0=\bar{p}_{3443}$ for the TN93 model, as shown in~\cite[Remark 5.12]{CHT}. Therefore, this still holds for the F81 model. Hence $\tilde{p}_{3434}=0=\tilde{p}_{3443}$. 

To prove the remaining equalities, we will show that each listed $\tilde{p}_{ijkl}$ is equal to $\tilde{p}_{2444}$ using rank 1 conditions from Table~\ref{tab:flatrank2} and (\ref{eq:rho_quartet}).  First, notice that 
$$
0=\left|A_{14,44}^{14,42}\right|=\pp_{1144}\tilde{p}_{1414}\pp_{2444}\tilde{p}_{2444}-\pp_{1244}\tilde{p}_{1444}\pp_{1444}\tilde{p}_{4414}=\frac{\pi_{12}(\pi_2-\pi_1)}{\pi_1^3\pi_2^3}(\tilde{p}_{1414}\tilde{p}_{2444}-\tilde{p}_{1444}\tilde{p}_{4414}),
$$

\noindent and rearranging gives
\begin{equation}\label{eq:54A}
\tilde{p}_{1414}\tilde{p}_{2444}=\tilde{p}_{1444}\tilde{p}_{4414}.
\end{equation}

\noindent Again using rank 1 conditions from Table~\ref{tab:flatrank2} and (\ref{eq:rho_quartet}), along with (\ref{eq:54A}) and 
Remark \ref{rmk:primeideal} gives 
$$
0=\left|A_{13,ij}^{13,23}\right|=\pp_{1133}\tilde{p}_{1414}\pp_{2233}\tilde{p}_{23ij}-\pp_{1233}\tilde{p}_{1444}\pp_{1233}\tilde{p}_{4414}=\frac{1}{\pi_3^2\pi_4^2}\tilde{p}_{1414}(\tilde{p}_{23ij}-\tilde{p}_{2444}),
$$

\noindent for $i=2, j=3$ and for $i=3,j=2$, which implies that $\tilde{p}_{2323}=\tilde{p}_{2444}$ and $\tilde{p}_{2332}=\tilde{p}_{2444}$, and
$$
0=\left|A_{14,kl}^{14,42}\right|=\pp_{1144}\tilde{p}_{1414}\pp_{2244}\tilde{p}_{24lk}-\pp_{1244}\tilde{p}_{1444}\pp_{1244}\tilde{p}_{4414}=\frac{1}{\pi_1^2\pi_2^2}\tilde{p}_{1414}(\tilde{p}_{24lk}-\tilde{p}_{2444}),
$$

\noindent for $k=4,l=2$ and for $k=2,l=4$, which implies  $\tilde{p}_{2424}=\tilde{p}_{2444}$ and $\tilde{p}_{2442}=\tilde{p}_{2444}$.

\end{proof}

\begin{prop}\label{prop:non-binomialEq}
Let $T=T_{12}$ and consider the F81 model on it. The following non-binomial linear equations hold for any point in $CV_T \subseteq \CC^{64}$:
\begin{enumerate}
    \item [(a)] $\pi_{34}^3\tilde{p}_{2244}-(\pi_{12}^3+\pi_{34}^3)\tilde{p}_{2222}+\pi_{12}^3\tilde{p}_{2233}=0;$
    \item [(b)] $\pi_{12}^2(\pi_3^3+\pi_4^3)\tilde{p}_{3333}-\pi_{12}^2\pi_{34}(\pi_3-\pi_4)^2\tilde{p}_{2444}-\pi_3\pi_4\pi_{12}^2\pi_{34}\tilde{p}_{2233}+\pi_1\pi_2\pi_3^2\pi_4^2\pi_{34}^2\tilde{p}_{3344}=0;$
     \item [(c)] $\pi_{34}^2(\pi_1^3+\pi_2^3)\tilde{p}_{4444}-\pi_{12}\pi_{34}^2(\pi_1-\pi_2)^2\tilde{p}_{2444}-\pi_1\pi_2\pi_{12}\pi_{34}^2\tilde{p}_{2244}+\pi_1^2\pi_2^2\pi_3\pi_4\pi_{12}^2\tilde{p}_{3344}=0;$
    \item [(d)] $\pi_{34}^2\tilde{p}_{2244}-\pi_{12}^2\tilde{p}_{2233}-(\pi_{34}^2-\pi_{12}^2)\tilde{p}_{2444}=0;$ 
    \item [(e)] 
$\pi_1^2\pi_2^2\pi_3\pi_4\pi_{12}\tilde{p}_{3344}-\pi_{34}(\pi_1^3+\pi_2^3)(\tilde{p}_{4444}-\tilde{p}_{2444})=0.$

\end{enumerate}
Moreover, the following quadratic equation also holds for any point in $CV_T$: 
    $$(\pi_3^3+\pi_4^3)\tilde{p}_{1111}\tilde{p}_{3333}-(\pi_3^3+\pi_4^3-\pi_3\pi_4\pi_{34}^2)\tilde{p}_{1111}\tilde{p}_{2444}-\pi_3\pi_4\pi_{34}^2\tilde{p}_{1144}\tilde{p}_{4411}=0.$$
\end{prop}

\begin{proof} 
{Recall that $\pi$ is generic (see \Cref{rmk:generic_pi}), hence throughout the proof we can assume that all coordinates of the no-evolution point $\bar{\rho}_{i_1i_2i_3i_4}$ in \eqref{eq:rho_quartet} are non-zero.}

\begin{itemize}

\item[(a)] First, define the following,
\begin{equation}\label{eq:abc}
a=\pi_1\pi_2\pi_{34}^2, \ \ b=-\pi_{12}^2\pi_{34}^2, \ \ c=\pi_3\pi_4\pi_{12}^2.
\end{equation}

\noindent Note that using (\ref{eq:rho_quartet}) along with (\ref{eq:abc}), and the fact that $\pi_{12}+\pi_{34}=1$ gives
\begin{equation}\label{eq:helpfulzeros}
\begin{array}{c}
a\pp_{1144}+b\pp_{1122}+c\pp_{1133}=0 \\
a\pp_{1244}+b\pp_{1222}+c\pp_{1233}=0.
\end{array}
\end{equation}

\noindent Notice that 
\begin{equation}\label{eq:53aA}
a\left|A_{11,12,44}^{11,22,12}\right|+b\left|A_{11,12,22}^{11,22,12}\right|+c\left|A_{11,12,33}^{11,22,12}\right|=0,
\end{equation}

\noindent since each of the minors involved arises from a rank condition in Table \ref{tab:flatrank2}. Expanding each matrix in (\ref{eq:53aA}), via the third row and regrouping terms based on the column involved gives 
\begin{equation}\label{eq:53aB}
\begin{array}{l}
\left|A_{11,12}^{22,12}\right|\tilde{p}_{4411}(a\pp_{1144}+b\pp_{1122}+c\pp_{1133})  - \left|A_{11,12}^{11,12}\right|(a\pp_{2244}\tilde{p}_{2244}+b\pp_{2222}\tilde{p}_{2222}+c\pp_{2233}\tilde{p}_{2233}) \\
+  \left|A_{11,12}^{11,22}\right|(a\pp_{1244}+b\pp_{1222}+c\pp_{1233})=0.
\end{array}
\end{equation}

\noindent Using (\ref{eq:helpfulzeros}) in (\ref{eq:53aB}) and computing the remaining minor gives 
$$
(\pp_{1122}\tilde{p}_{1111}\tilde{p}_{1414})(a\pp_{2244}\tilde{p}_{2244}+b\pp_{2222}\tilde{p}_{2222}+c\pp_{2233}\tilde{p}_{2233})=0,
$$

\noindent and, {by \Cref{rmk:primeideal},} we get
$$
a\pp_{2244}\tilde{p}_{2244}+b\pp_{2222}\tilde{p}_{2222}+c\pp_{2233}\tilde{p}_{2233}=0.
$$

\noindent Plugging in (\ref{eq:rho_quartet}) and (\ref{eq:abc}) and clearing the denominators gives (a).  

\item[(b)] Define 
\begin{equation}\label{eq:d}
d=(\pi_4-\pi_3)\pi_{12}^2\pi_{34}
\end{equation}

\noindent and note that (\ref{eq:rho_quartet}), (\ref{eq:abc}), and (\ref{eq:d}) give the equality 
\begin{equation}\label{eq:helpfulzero2}
c\pp_{1333}+d\pp_{1233}=0.
\end{equation}

\noindent Now, we further note the following:
\begin{equation}\label{eq:53bA}
a\left|A_{11,12,13,44}^{11,12,13,33}\right|+b\left|A_{11,12,13,22}^{11,12,13,33}\right|+c\left|A_{11,12,13,33}^{11,12,13,33}\right|+d\left|A_{11,12,13,23}^{11,12,13,33}\right|=0,
\end{equation}

\noindent as each of the minors involved arises from a zero rank condition in Table \ref{tab:flatrank3b}. Now, we expand each matrix in (\ref{eq:53bA}) along the fourth row and regroup terms based on their column. This yields
\begin{equation}\label{eq:53bB}
\begin{array}{l}
-\left|A_{11,12,13}^{12,13,33}\right|\tilde{p}_{4411}(a\pp_{1144}+b\pp_{1122}+c\pp_{1133})+\left|A_{11,12,13}^{11,13,33}\right|\tilde{p}_{4414}(a\pp_{1244}+b\pp_{1222}+c\pp_{1233})\\
-\left|A_{11,12,13}^{11,12,33}\right|\tilde{p}_{4414}(c\pp_{1333}+d\pp_{1233})
+\left|A_{11,12,13}^{11,12,13}\right|(a\tilde{p}_{3344}+b\pp_{2233}\tilde{p}_{2233}+c\pp_{3333}\tilde{p}_{3333}+d\pp_{2333}\tilde{p}_{2444})=0.
\end{array}
\end{equation}

\noindent By plugging (\ref{eq:helpfulzeros}) and (\ref{eq:helpfulzero2}) into (\ref{eq:53bB}) and computing the remaining minor, we get 
$$
(\pp_{1122}\pp_{1133}\tilde{p}_{1111}\tilde{p}_{1414}^2)(a\tilde{p}_{3344}+b\pp_{2233}\tilde{p}_{2233}+c\pp_{3333}\tilde{p}_{3333}+d\pp_{2333}\tilde{p}_{2444})=0.
$$

\noindent {Again by \Cref{rmk:primeideal},}
$$
a\tilde{p}_{3344}+b\pp_{2233}\tilde{p}_{2233}+c\pp_{3333}\tilde{p}_{3333}+d\pp_{2333}\tilde{p}_{2444}=0.
$$

\noindent Finally, plug (\ref{eq:rho_quartet}), (\ref{eq:abc}), and (\ref{eq:d}) into the above equation and rearrange to obtain (b). 

\item[(c)] Now, define 
\begin{equation}\label{eq:e}
e=\pi_{12}\pi_{34}^2(\pi_1-\pi_2)
\end{equation}

\noindent and note that (\ref{eq:rho_quartet}), (\ref{eq:abc}), and (\ref{eq:e}) give the equality
\begin{equation}\label{eq:helpfulzero3}
a\pp_{1444}+e\pp_{1244}=0.
\end{equation}

\noindent By combining the minors in Table \ref{tab:flatrank4}, we have 
\begin{equation}\label{eq:53cA}
a\left|A_{11,14,12,44}^{11,14,12,44}\right|+b\left|A_{11,14,12,22}^{11,14,12,44}\right|+c\left|A_{11,14,12,33}^{11,14,12,44}\right|+e\left|A_{11,14,12,24}^{11,14,12,44}\right|=0.
\end{equation}

\noindent We expand each matrix in (\ref{eq:53cA}) using the fourth row and regroup terms to obtain 
\begin{equation}\label{eq:53cB}
\begin{array}{l}
-\left|A_{11,14,12}^{14,12,44}\right|\tilde{p}_{4411}(a\pp_{1144}+b\pp_{1122}+c\pp_{1133})+\left|A_{11,14,12}^{11,12,44}\right|\tilde{p}_{4414}(a\pp_{1444}+e\pp_{1244}) \\
-\left|A_{11,14,12}^{11,14,44}\right|\tilde{p}_{4414}(a\pp_{1244}+b\pp_{1222}+c\pp_{1233})\\
+\left|A_{11,14,12}^{11,14,12}\right|(a\pp_{4444}\tilde{p}_{4444}+b\pp_{2244}\tilde{p}_{2244}+c\tilde{p}_{3344}+e\pp_{2444}\tilde{p}_{2444})=0.
\end{array}
\end{equation}

\noindent By plugging (\ref{eq:helpfulzeros}) and (\ref{eq:helpfulzero3}) into (\ref{eq:53cB}) and computing the remaining minor, we get the following
$$
(\pp_{1122}\pp_{1144}\tilde{p}_{1111}\tilde{p}_{1414}^2)(a\pp_{4444}\tilde{p}_{4444}+b\pp_{2244}\tilde{p}_{2244}+c\tilde{p}_{3344}+e\pp_{2444}\tilde{p}_{2444})=0,
$$

\noindent {and by \Cref{rmk:primeideal},}
$$
a\pp_{4444}\tilde{p}_{4444}+b\pp_{2244}\tilde{p}_{2244}+c\tilde{p}_{3344}+e\pp_{2444}\tilde{p}_{2444}=0.
$$

\noindent Plugging in (\ref{eq:rho_quartet}), (\ref{eq:abc}), and (\ref{eq:e}) gives (c).

\item[(d)] Consider $A_{11,12,33}^{11,22,12}$ from Table~\ref{tab:flatrank3a}. Computing the determinant by expanding row (1,1) yields
\begin{eqnarray*}
\left|A_{11,12,33}^{11,22,12}\right| &=& \tilde{p}_{1111}\left|A_{12,33}^{22,12}\right|-\pp_{1122}\tilde{p}_{1144}\left|A_{12,33}^{11,12}\right| \\
&=& \tilde{p}_{1111}(\pp_{1233}\pp_{1222}\tilde{p}_{4414}\tilde{p}_{1444}-\pp_{2233}\pp_{1122}\tilde{p}_{1414}\tilde{p}_{2233}) \\
&&
+\pp_{1122}\tilde{p}_{1144}\pp_{1133}\tilde{p}_{4411}\tilde{p}_{1414}\pp_{1122}.
\end{eqnarray*}

\noindent Using $\left|A_{11,12,33}^{11,22,12}\right|=0$, {factoring out $\tilde{p}_{1414}$ and applying (\ref{eq:54A}) give}
\begin{equation}\label{eq:rank2D1}
\tilde{p}_{1111} (\pp_{2233}\pp_{1122}\tilde{p}_{2233}-\pp_{1233}\pp_{1222}\tilde{p}_{2444})=\pp_{1122}\tilde{p}_{1144}\pp_{1133}\tilde{p}_{4411}\pp_{1122}.
\end{equation}
Consider $A_{11,12,44}^{11,22,12}$ from Table~\ref{tab:flatrank3a}. Proceeding as we did for (\ref{eq:rank2D1}) gives
\begin{equation}\label{eq:rank2D2}
\tilde{p}_{1111} (\pp_{2244}\pp_{1122}\tilde{p}_{2244}-\pp_{1244}\pp_{1222}\tilde{p}_{2444})=\pp_{1122}\pp_{1144}\tilde{p}_{1144}\tilde{p}_{4411}\pp_{1122}.\end{equation}
Multiplying (\ref{eq:rank2D1}) by $\pp_{1144}$ and (\ref{eq:rank2D2}) by $\pp_{1133}$ makes the right-hand-side of both equations equal. 
{Factoring out ${\tilde{p}_{1111}}$ gives}
$$\pp_{1133}(\pp_{2244}\pp_{1122}\tilde{p}_{2244}-\pp_{1244}\pp_{1222}\tilde{p}_{2444})
= \pp_{1144}(\pp_{2233}\pp_{1122}\tilde{p}_{2233}-\pp_{1233}\pp_{1222}\tilde{p}_{2444}).$$
Therefore, simplifying the resulting equation using (\ref{eq:rho_quartet}) gives (d).

\item[(e)] This equation follows from (\ref{eq:54A}) and rank 3 conditions on the flattening matrix. Consider $A_{11,12,14,33}^{11,12,14,44}$ and $A_{11,12,14,44}^{11,12,14,44}$ from Table~\ref{tab:flatrank4}. Observe that the first rows are equal; therefore, computing the determinants by expanding the first column $(1,1)$ and then column $(1,4)$ yields
\begin{eqnarray*}
  0&=&\pp_{1144}\left\vert A_{11,12,14,33}^{11,12,14,44}\right\vert-\pp_{1133}\left\vert A_{11,12,14,44}^{11,12,14,44}\right\vert
=\tilde{p}_{1111}\left(\pp_{1144}\left\vert A_{12,14,33}^{12,14,44}\right\vert-\pp_{1133}\left\vert A_{12,14,44}^{12,14,44}\right\vert\right)\\
&=&\tilde{p}_{1111}\left(\pp_{1144}^2\tilde{p}_{1414}\left\vert A_{12,33}^{12,44}\right\vert
    -\pp_{1133}\left(\pp_{1144}\tilde{p}_{1414}\left\vert A_{12,44}^{12,44}\right\vert
    -\pp_{1444}\tilde{p}_{4414}\left\vert A_{12,14}^{12,44}\right\vert\right)\right).  
\end{eqnarray*}
Observe that $\left\vert A_{12,14}^{12,44}\right\vert=\pp_{1122}\pp_{1444}\tilde{p}_{1414}\tilde{p}_{1444}$. Thus, diving both sides by $\tilde{p}_{1111}\tilde{p}_{1414}$ gives 
\begin{equation*}
\pp_{1144}^2\left\vert A_{12,33}^{12,44}\right\vert-\pp_{1133}\left(\pp_{1144}\left\vert A_{12,44}^{12,44}\right\vert-\pp_{1122}\pp_{1444}^2\tilde{p}_{1444}\tilde{p}_{4414}\right)=0.
\end{equation*}
Using (\ref{eq:54A}) and diving both sides by $\tilde{p}_{1414}$ yields
\begin{eqnarray*}
0&=&\pp_{1144}^2\begin{vmatrix}
    \pp_{1122} & \pp_{1244}\\
    \pp_{1233}\tilde{p}_{2444} & \tilde{p}_{3344} 
\end{vmatrix}
-\pp_{1133}\left(\pp_{1144}\begin{vmatrix}
    \pp_{1122} & \pp_{1244}\\
    \pp_{1244}\tilde{p}_{2444} & \pp_{4444}\tilde{p}_{4444} 
\end{vmatrix}-\pp_{1122}\pp_{1444}^2\tilde{p}_{2444}\right)\\
&=& \tilde{p}_{2444}(-\pp_{1144}^2\pp_{1244}\pp_{1233}+\pp_{1133}\pp_{1244}^2\pp_{1144}+\pp_{1133}\pp_{1444}^2\pp_{1122})+ \tilde{p}_{3344}\pp_{1144}^2\pp_{1122} \\
&& - \tilde{p}_{4444}\pp_{1133}\pp_{1144}\pp_{1122}\pp_{4444}.
\end{eqnarray*}
Therefore, simplifying the resulting equation using (\ref{eq:rho_quartet}) and $(\pi_1^3+\pi_2^3)=(\pi_1^2+\pi_2^2-\pi_1\pi_2)\pi_{12}$ gives (e).

\end{itemize}

Finally, in order to prove the quadratic equation, consider $A_{11,12,33,13}^{11,12,33,13}$ in Table~\ref{tab:flatrank3b}. Computing the determinant by expanding using column (1,1) and then column (1,2) yields 
\begin{eqnarray*}
    \left| A_{11,12,33,13}^{11,12,33,13}\right|
    &=& \tilde{p}_{1111}\left|A_{12,33,13}^{12,33,13}\right|+\pp_{1133}\tilde{p}_{4411}\left|A_{11,12,13}^{12,33,13}\right|\\
    &=& \tilde{p}_{1111} \left(\tilde{p}_{1414}\pp_{1122}\left|A_{33,13}^{33,13}\right| - \tilde{p}_{4414}\pp_{1233}\left|A_{12,13}^{33,13}\right| \right) - \tilde{p}_{1414}\tilde{p}_{4411}\pp_{1122}\pp_{1133}\left|A_{11,13}^{33,13}\right|\\
    &=& \tilde{p}_{1111}\left[\pp_{1122}\tilde{p}_{1414}(\pp_{3333}\tilde{p}_{3333}\pp_{1133}\tilde{p}_{1414}-\pp^2_{1333}\tilde{p}_{4414}\tilde{p}_{1444})\right.\\
     &&-
\left.\pp_{1233}^2\tilde{p}_{4414}\tilde{p}_{1444}\pp_{1133}\tilde{p}_{1414}\right] - \pp_{1133}^3\pp_{1122}\tilde{p}_{4411}\tilde{p}^2_{1414}\tilde{p}_{1144}\\
    &=& \tilde{p}_{1414}[\pp_{1122}\pp_{3333}\pp_{1133}\tilde{p}_{1111}\tilde{p}_{3333}\tilde{p}_{1414}-(\tilde{p}_{1111}\tilde{p}_{4414}\tilde{p}_{1444})(\pp_{1122}\pp_{1333}^2\\
    && +\pp_{1233}^2\pp_{1133})-\pp_{1133}^3\pp_{1122}\tilde{p}_{4411}\tilde{p}_{1414}\tilde{p}_{1144}].
\end{eqnarray*}

\noindent Using $\left| A_{11,12,33,13}^{11,12,33,13}\right|=0$ together with (\ref{eq:54A}), factoring out ${\tilde{p}_{1414}}$, and simplifying the resulting equation with (\ref{eq:rho_quartet}) gives the desired equation.

\end{proof}
\begin{remark}\label{rem:nonzero_F81}
\rm Note that the equations in Propositions \ref{prop:4quartetEq} and \ref{prop:non-binomialEq} have been obtained by imposing only the rank conditions on flattenings and the symmetry equations defining $\F_4^{F81}$. Throughout the proof of Proposition \ref{prop:non-binomialEq}, we have only used that the coordinates $\tilde{p}_{1111}$ and $\tilde{p}_{1414}$ are different from zero.
\end{remark}

\section{The space of phylogenetic mixtures}\label{sec:mixtures}

We derive the spaces of the phylogenetic mixtures of F81 and F84 from $\E^{TN93}_n$, which was computed in \cite{CHT} for $n=3,4$. {If $\Z_n$ is the space of model zeroes of TN93, then  $\E^{TN93}_n=\Z_n$ for $n=3,4$, see \cite[Proposition 5.11]{CHT}.
 The main results of the section can be summarized as follows:}


\begin{thm}
Let 
$\F_n^\MM$ be the space of symmetry equations for the model $\MM$. Then
\begin{itemize}
    \item $n=3$: 
    $\E^{F84}_3=\F_3^{F84}\cap\Z_3$, $\E^{F81}_3=\F_3^{F81}\cap\Z_3$ with dimensions 
    $12$ and $5$ respectively;
    \item $n=4$: 
    $\E^{F84}_4=\F_4^{F84}\cap\Z_4$, $\E^{F81}_4=\F_4^{F81}\cap\Z_4\cap\mathcal{G}$ with dimensions 
    $55$ and $15$ respectively,
\end{itemize} 
\noindent
where $\mathcal{G}$ is the linear space cut out by the 9 non-binomial linear equations of \Cref{prop:modelEq12|34}.
\end{thm}
\begin{proof}
Combine \Cref{prop:mixtures_quartetF84}, \Cref{prop:eqtripod} and \Cref{thm:mixturesF81}.
\end{proof}

\subsection{Mixtures for F84}
We first focus on the F84 model. We will prove that the model zeroes and symmetry equations are sufficient to define the space of phylogenetic mixtures for both tripods and quartets.

\begin{thm}\label{prop:mixtures_quartetF84}
The space of mixtures $\E^{F84}_n$ for trees evolving under the F84 model coincides with the space $\F^{F84}_n\cap \Z_n$ of model zeros and symmetry equations for $n=3,4$. 
\end{thm} 
\begin{proof}
For $n=3$, \texttt{Macaulay2} computations show that the vanishing ideal $I_T^{F84}$ is generated in degree 1 by the equations that cut out $\F^{F84}_3\cap \Z_3$.

Now consider $n=4$. Recall that $\F^{F84}_4\cap\Z_4$ is a 55-dimensional linear space containing $\E^{F84}_4$, see \Cref{prop:quartet-eqF84}. To prove equality, it is enough to show that there exist 55 linearly independent points in $CV_{T_{12}}^{F84}\cup CV_{T_{13}}^{F84}\cup CV_{T_{14}}^{F84}$.

Consider the diagonal matrix $D=\mathrm{diag}(\lambda_1,\lambda_2,\lambda_4,\lambda_4)$, and for $i\in\lbrace 1,2,4\rbrace$, let $D_i$ be the specialization to $\lambda_i=1$ and $\lambda_j=0$, for $j\neq i$.
Fix $T=T_{12}$ and consider points 
 \[{p}^{\,\mathbf{i}}={\varphi}_T(D_{i_1},D_{i_2},D_{i_3},D_{i_4};Id),\]
\noindent
where $\mathbf{i}=(i_1,i_2,i_3,i_4)$, $i_j\in\lbrace 1,2,4\rbrace$. Note that they can be seen as points in the 4-leaf star tree 
and its coordinates have the following expression: 
$$p^\mathbf{i}_{j_1j_2j_3j_4}=\delta_{i_1j_1}\delta_{i_2j_2}\delta_{i_3j_3}\delta_{i_4j_4}\bar{\rho}_{j_1j_2j_3j_4},$$
\noindent
where $j_k\in[4]$, $\delta_{ij}$ is the Kronecker delta and coordinates $\bar{\rho}_{j_1j_2j_3j_4}$ are displayed in \eqref{eq:rho_quartet}.
 
This construction provides $3^4=81$ points $p^{\,\mathbf{i}}$, 36 of which are the zero point. For the remaining 45 points, coordinate $p^\mathbf{i}_{i_1i_2i_3i_4}$ is always non-zero. If $4$ is not in $\mathbf{i}$, then this is the single non-zero coordinate of $p^{\,\mathbf{i}}$. Otherwise, replacing all occurrences of $4$ in $\mathbf{i}$ with $3$ gives an additional non-zero coordinate $p^\mathbf{i}_{j_1j_2j_3j_4}$. Note that replacing some but not all $4$'s does not result in further non-zero coordinates. 
Crucially, the non-zero coordinates of the $p^{\,\mathbf{i}}$ are distinct across different points. Thus, the set of non-zero $p^{\,\mathbf{i}}$ spans a 45-dimensional linear subspace of $\E_4^{F84}$.




Now we consider the following 4 points, where the interior edge is no longer the identity: 

 \[p^{(2,2,4,4;1)}={\varphi}_{T_{12}}(D_2,D_2,D_4,D_4;D_1),\quad p^{(4,4,2,2;1)}={\varphi}_{T_{12}}(D_4,D_4,D_2,D_2;D_1)\]
 \[p^{(4,4,4,4;1)}={\varphi}_{T_{12}}(D_4,D_4,D_4,D_4;D_1),\quad p^{(4,4,4,4;2)}={\varphi}_{T_{12}}(D_4,D_4,D_4,D_4;D_2).\]

Note that the only two non-zero coordinates of $p^{(2,2,4,4;1)}$ are in the same position than those of $p^{(2,2,4,4)}$ but they are linearly independent points for generic $\pi$: 
$$p^{(2,2,4,4)}_{2233}=1/\pi_{34}\pi_3\pi_4,\quad p^{(2,2,4,4)}_{2244}=1/\pi_{12}\pi_1\pi_2,$$
$$p^{(2,2,4,4;1)}_{2233}=\pi_3\pi_4\pi_{34}^2/\pi_{12},\quad p^{(2,2,4,4;1)}_{2244}=\pi_1\pi_2\pi_{12}^2/\pi_{34}.$$

Linear independence of $p^{(4,4,2,2;1)}$ and $p^{(2,2,4,4;1)}$ is justified analogously; details for $p^{(4,4,4,4;1)}$, $p^{(4,4,4,4;2)}$ and $p^{(4,4,4,4)}$ appear in \Cref{tab:4444}. This amounts to $49$ linearly independent points in $\E_4^{F84}$.

Finally, we can prove analogously that the following 6 additional points arising from tree topologies $T_{13}$ and $T_{14}$ are also linearly independent:
\begin{align*}
    &p^{(2,4,2,4;1)}_{T_{13}}={\varphi}_{T_{13}}(D_2,D_4,D_2,D_4;D_1),\quad p^{(2,4,2,4;1)}_{T_{14}}={\varphi}_{T_{14}}(D_2,D_4,D_2,D_4;D_1),\\
    & p^{(4,2,4,2;1)}_{T_{13}}={\varphi}_{T_{13}}(D_4,D_2,D_4,D_2;D_1),\quad
    p^{(4,2,4,2;1)}_{T_{14}}={\varphi}_{T_{14}}(D_4,D_2,D_4,D_2;D_1),\\
    & 
    p^{(4,4,4,4;1)}_{T_{13}}={\varphi}_{T_{13}}(D_4,D_4,D_4,D_4;D_1),\quad
    p^{(4,4,4,4;1)}_{T_{14}}={\varphi}_{T_{14}}(D_4,D_4,D_4,D_4;D_1).
\end{align*}

\begin{table}[h]
    \centering
    \begin{tabular}{c|c|c}
 Point & positions of non-zero coordinates & non-zero coordinates\\
 \hline
 $p^{(4,4,4,4)}$    &  $(3,3,3,3)$ & $(\pi_3^3+\pi_4^3)/\pi_3^3\pi_4^3$\\
    &  $(4,4,4,4)$ & $(\pi_1^3+\pi_2^3)/\pi_1^3\pi_2^3$\\ 
    \hline
 $p^{(4,4,4,4;1)}$    &  $(3,3,3,3)$ & $(\pi_{34}/\pi_3\pi_4)^2$\\
 &  $(3,3,4,4)$ & $\pi_{12}\pi_{34}/\pi_1\pi_2\pi_3\pi_4$\\
  &  $(4,4,3,3)$ & $\pi_{12}\pi_{34}/\pi_1\pi_2\pi_3\pi_4$\\
  & $(4,4,4,4)$ & $(\pi_{12}/\pi_1\pi_2)^2$\\
  \hline
 $p^{(4,4,4,4;2)}$    &  $(3,3,3,3)$ & $\pi_{12}\pi_{34}/\pi_3^2\pi_4^2$\\
 &  $(3,3,4,4)$ & $-\pi_{12}\pi_{34}/\pi_1\pi_2\pi_3\pi_4$\\
  &  $(4,4,3,3)$ & $-\pi_{12}\pi_{34}/\pi_1\pi_2\pi_3\pi_4$\\
  & $(4,4,4,4)$ & $\pi_{12}\pi_{34}/\pi_1^2\pi_2^2$\\
  \hline
     $p^{(4,4,4,4;1)}_{T_{13}}$    &  $(3,3,3,3)$ & $(\pi_{34}/\pi_3\pi_4)^2$\\
 &  $(3,4,3,4)$ & $\pi_{12}\pi_{34}/\pi_1\pi_2\pi_3\pi_4$\\
  &  $(4,3,4,3)$ & $\pi_{12}\pi_{34}/\pi_1\pi_2\pi_3\pi_4$\\
  & $(4,4,4,4)$ & $(\pi_{12}/\pi_1\pi_2)^2$\\
  \hline
   $p^{(4,4,4,4;1)}_{T_{14}}$    &  $(3,3,3,3)$ & $(\pi_{34}/\pi_3\pi_4)^2$\\
 &  $(3,4,4,3)$ & $\pi_{12}\pi_{34}/\pi_1\pi_2\pi_3\pi_4$\\
  &  $(4,3,3,4)$ & $\pi_{12}\pi_{34}/\pi_1\pi_2\pi_3\pi_4$\\
  & $(4,4,4,4)$ & $(\pi_{12}/\pi_1\pi_2)^2$\\

\end{tabular}
    \caption{Linearly independent points obtained by considering matrices $D_4$ at pendant edges for different tree topologies and different matrices at the interior edge.}
    \label{tab:4444}
\end{table}

This makes a total of $55$ linearly independent points in $CV_{T_{12}}^{F84}\cup CV_{T_{13}}^{F84}\cup CV_{T_{14}}^{F84}$.

\end{proof}

In the case of quartets, we are able to describe the space of mixtures on a given tree $T$, $\E_T^\MM$, using additional \emph{topology} symmetry equations.

\begin{cor}\label{cor:mixt_tree}
The space of mixtures $\mathcal{E}_T^{F84}$ for $T=T_{12}$ has dimension $49$ ($256-201-6=49$) and it is cut out by the equations defining $\Z_4$ in \Cref{prop:quartet-zeroes}, $\F^{F84}_4$ in \Cref{prop:quartet-eqF84} and the $6$ equations in \Cref{prop:4quartetEq84}.
\end{cor}
\begin{proof}
In the proof of \Cref{prop:mixtures_quartetF84}, we obtained $45$ points from the star tree and 4 from $T_{12}$. This yields a total of 49 linearly independent points in $CV_T^{F84}$.
\end{proof}

\subsection{Mixtures for F81}

Now we focus on the F81 model. Note that the case of tripods for this model was already fully understood in \Cref{prop:eqtripod} and $\E_3^{F81}$ is again completely defined by model zeroes and symmetry equations. However, these equations will not be enough to cut out the space of phylogenetic mixtures for quartets. Unlike in the case of the F84 model, the dimensions of the spaces of mixtures are known: $\dim \E_4^{F81}=15$ and $\dim \E_T^{F81}=13$  for a given quartet $T$ \cite{CasSteel}. We will start by describing the latter space and then provide the $9$ additional non-binomial linear equations that define the whole space of mixtures for quartets.

\begin{prop} The linear span $\mathcal{E}_T^{F81}=\langle CV_T^{F81}\rangle$ for the quartet {$T=T_{12}$} has dimension $(256-243)=13$ and is defined by the equations defining $\F^{F81}_4\cap \Z_4$, together with the six equations in Proposition \ref{prop:4quartetEq}, and the five linear equations in Proposition \ref{prop:non-binomialEq}.
\end{prop}


\begin{proof}
Let $Y\subseteq \otimes^4\CC^4$ be the linear space defined by the linearly independent equations defining $\F_4^{F81}\cap \Z_4$ ($172$ model zeros and $60$ symmetry equations), the six equations in Proposition \ref{prop:4quartetEq}, and the five linear equations in Proposition \ref{prop:non-binomialEq}. It is straightforward to check that all these equations are linearly independent and $Y$ has dimension $13=256-172-60-6-5$.

We have seen that these linear equations vanish on $CV_T^{F81}$ and hence they vanish on $\E_T^{F81}$, so that $\E_T^{F81} \subseteq Y$. The dimensions of the two spaces agree by \cite{CasSteel}, and this completes the proof.
\end{proof}

In Proposition~\ref{prop:non-binomialEq}, we highlighted linear and quadratic equations that are satisfied for tree $T_{12}$. In the following proposition we give nine linear equations that vanish on all $CV_T^{F81}$ (the equation in Proposition~\ref{prop:non-binomialEq}(d) is one of them), which  therefore give rise to model invariants.
\begin{prop}\label{prop:modelEq12|34}
The following equations are model equations.
\begin{enumerate}
        \item[(a)] $\pi_{34}^2\tilde{p}_{2244}=\pi_{12}^2\tilde{p}_{2233}+(\pi_{34}^2-\pi_{12}^2)\tilde{p}_{2444}$;
        \item[(b)] $\pi_{34}^2\tilde{p}_{2424}=\pi_{12}^2\tilde{p}_{2323}+(\pi_{34}^2-\pi_{12}^2)\tilde{p}_{2444}$;
        \item[(c)] $\pi_{34}^2\tilde{p}_{2442}=\pi_{12}^2\tilde{p}_{2332}+(\pi_{34}^2-\pi_{12}^2)\tilde{p}_{2444}$;
        \item[(d)] $\pi_1\pi_2\pi_3\pi_4\tilde{p}_{3344}=\pi_{12}^2(\tilde{p}_{2233}-\tilde{p}_{2444})$;
        \item[(e)] $\pi_1\pi_2\pi_3\pi_4\tilde{p}_{3434}=\pi_{12}^2(\tilde{p}_{2323}-\tilde{p}_{2444})$;
        \item[(f)] $\pi_1\pi_2\pi_3\pi_4\tilde{p}_{3443}=\pi_{12}^2(\tilde{p}_{2332}-\tilde{p}_{2444})$;
        \item[(g)] $\pi_1\pi_2\pi_3\pi_4(\tilde{p}_{3344}+\tilde{p}_{3434}+\tilde{p}_{3443})=(\pi_{12}^3+\pi_{34}^3)(\tilde{p}_{2222}-\tilde{p}_{2444})$
  
        \item[(h)]  $\pi_{12}(\pi_3^3+\pi_4^3)(\tilde{p}_{3333}-\tilde{p}_{2444})=\pi_3\pi_4\pi_{34}(\pi_{12}^3+\pi_{34}^3)(\tilde{p}_{2222}-\tilde{p}_{2444})$

        \item[(i)] $\pi_{34}(\pi_1^3+\pi_2^3)(\tilde{p}_{4444}-\tilde{p}_{2444})=\pi_1\pi_2\pi_{12}(\pi_{12}^3+\pi_{34}^3)(\tilde{p}_{2222}-\tilde{p}_{2444})$
    \end{enumerate}
\end{prop}
\begin{proof}
We first show that the nine equations hold in the $T_{12}$ tree topology and then argue that they must hold for the $T_{13}$ and $T_{14}$ tree topologies.  

\begin{itemize}
\item[(a)] This is the equation given in Proposition \ref{prop:non-binomialEq}(d).

\item[(b)-(c)] By Proposition~\ref{prop:4quartetEq}, $\tilde{p}_{2424}=\tilde{p}_{2323}=\tilde{p}_{2444}$ and $\tilde{p}_{2442}=\tilde{p}_{2332}=\tilde{p}_{2444}$ for the $T_{12}$ topology yielding (b) and (c).

\item[(d)] 
We multiply the equation in Proposition \ref{prop:non-binomialEq} (e) by $\pi_{34}$ and add it to the equation in Proposition \ref{prop:non-binomialEq} (c).  In the resulting equation, combining like terms, simplifying using that $\pi_{12}+\pi_{34}=1$, and dividing both sides by $\pi_1\pi_2\pi_{12}$ gives
\begin{equation}\label{eq:modeleqd}
\pi_{34}^2\tilde{p}_{2444}-\pi_{34}^2\tilde{p}_{2244}+\pi_1\pi_2\pi_3\pi_4\tilde{p}_{3344}=0.
\end{equation}
\noindent Solving the equation in Proposition \ref{prop:non-binomialEq} (d) for the $\tilde{p}_{2244}$ term, and plugging it into (\ref{eq:modeleqd}), collecting like terms and rearranging gives (d).

\item[(e)-(f)] By Proposition \ref{prop:4quartetEq}, $\tilde{p}_{3434}=0=\tilde{p}_{3443}$, $\tilde{p}_{2323}=\tilde{p}_{2444}$, and $\tilde{p}_{2332}=\tilde{p}_{2444}$ yielding (e) and (f).  

\item[(g)] First, we multiply the equation in Proposition \ref{prop:non-binomialEq} (d) by $\pi_{34}$ and subtract it from the equation in Proposition \ref{prop:non-binomialEq} (a).  In the resulting equation, combining like terms, using that $\pi_{12}+\pi_{34}=1$, and rearranging gives
\begin{equation}\label{eq:modelg1}
\pi_{12}^2\tilde{p}_{2233}=(\pi_{12}^3+\pi_{34}^3)\tilde{p}_{2222}-\pi_{34}(\pi_{34}^2-\pi_{12}^2)\tilde{p}_{2444}.
\end{equation}

\noindent Using that $\tilde{p}_{3434}=0=\tilde{p}_{3443}$ for $T_{12}$ by Proposition \ref{prop:4quartetEq}, followed by Part (d) of this Proposition, plugging in (\ref{eq:modelg1}), then simplifying using that $\pi_{12}+\pi_{34}=1$ gives 
\begin{eqnarray*}
\pi_1\pi_2\pi_3\pi_4(\tilde{p}_{3344}+\tilde{p}_{3434}+\tilde{p}_{3443}) &=& \pi_1\pi_2\pi_3\pi_4\tilde{p}_{3344} \\
&=& \pi_{12}^2(\tilde{p}_{2233}-\tilde{p}_{2444}) \\
&=& (\pi_{12}^3+\pi_{34}^3)\tilde{p}_{2222}-\pi_{34}(\pi_{34}^2-\pi_{12}^2)\tilde{p}_{2444}-\pi_{12}^2\tilde{p}_{2444} \\
&=& (\pi_{12}^3+\pi_{34}^3)(\tilde{p}_{2222}-\tilde{p}_{2444}),
\end{eqnarray*}

\noindent which is (g).








\item[(h)] Using Proposition \ref{prop:non-binomialEq} (b) (combining like terms and simplifying using that $\pi_{12}+\pi_{34}=1$) followed by Part (d) of this proposition, using that $\pi_{12}+\pi_{34}=1$, and then using Part (g) of this Proposition (with $\tilde{p}_{3434}=0=\tilde{p}_{3443}$ for $T_{12}$ by Proposition \ref{prop:4quartetEq}) yields
\begin{eqnarray*}
\pi_{12}^2(\pi_3^3+\pi_4^3)(\tilde{p}_{3333}-\tilde{p}_{2444}) &=& \pi_3\pi_4\pi_{12}^2\pi_{34}(\tilde{p}_{2233}-\tilde{p}_{2444})-\pi_1\pi_2\pi_3^2\pi_4^2\pi_{34}^2\tilde{p}_{3344} \\
&=& \pi_1\pi_2\pi_3^2\pi_4^2\pi_{34}\tilde{p}_{3344}-\pi_1\pi_2\pi_3^2\pi_4^2\pi_{34}^2\tilde{p}_{3344} \\
&=& \pi_1\pi_2\pi_3^2\pi_4^2\pi_{12}\pi_{34}\tilde{p}_{3344} \\
&=& \pi_3\pi_4\pi_{12}\pi_{34}(\pi_{12}^3+\pi_{34}^3)(\tilde{p}_{2222}-\tilde{p}_{2444}).
\end{eqnarray*}

\noindent Dividing both sides by $\pi_{12}$ gives (h).





\item[(i)] Using Proposition \ref{prop:non-binomialEq} (e) followed by Part (g) of this Proposition (with $\tilde{p}_{3434}=0=\tilde{p}_{3443}$ for $T_{12}$ by Proposition \ref{prop:4quartetEq}) yields
$$
\pi_{34}(\pi_1^3+\pi_2^3)(\tilde{p}_{4444}-\tilde{p}_{2444})=\pi_1^2\pi_2^2\pi_3\pi_4\pi_{12}\tilde{p}_{3344}=\pi_1\pi_2\pi_{12}(\pi_{12}^3+\pi_{34}^3)(\tilde{p}_{2222}-\tilde{p}_{2444})
$$
\noindent which is (i).

\end{itemize} 
We now show that all nine equations hold on the $T_{13}$ and $T_{14}$ tree topologies.  For the $T_{13}$ topology, $\tilde{p}_{ijkl}$ is replaced with $\tilde{p}_{ikjl}$ and for the $T_{14}$ topology, $\tilde{p}_{ijkl}$ is replaced with $\tilde{p}_{iljk}$.  By permuting the indices, we see that the equalities from Proposition~\ref{prop:4quartetEq} become 
\begin{equation}\label{eq:T13T14}
\tilde{p}_{3344}=0=\tilde{p}_{3443}, \ \tilde{p}_{2244}=\tilde{p}_{2442}=\tilde{p}_{2332}=\tilde{p}_{2233}=\tilde{p}_{2444}
\end{equation}

\noindent for both $T_{13}$ and $T_{14}$. By permuting indices, our proof of (a) for $T_{12}$ yields (b) for the $T_{13}$ and $T_{14}$ topologies.  Using (\ref{eq:T13T14}) yields (a) and (c) for the $T_{13}$ and $T_{14}$ topologies. Permuting indices in our proof of (d) for $T_{12}$ gives (e) for the $T_{13}$ and $T_{14}$ topologies, and using (\ref{eq:T13T14}) gives (d) and (f).  For (g), in the final calculation we use that $\tilde{p}_{3344}=0=\tilde{p}_{3443}$ for $T_{13}$ and $T_{14}$, Part (e) of this Proposition in place of Part (d), and then plug in (\ref{eq:modelg1}) with the appropriate permutations of indices to yield (g) for $T_{13}$ and $T_{14}$. Since (h) and (i) remain unchanged when changing tree topologies, our previous proofs hold for $T_{13}$ and $T_{14}$.  

\end{proof}

\begin{remark}\rm Again, note that the equations in Proposition \ref{prop:modelEq12|34} have been obtained by imposing only the rank conditions on flattenings and the symmetry equations defining $\F^{F81}_4$. For later use, we observe that throughout the proof of Proposition \ref{prop:non-binomialEq}, we have only used that the coordinates $\tilde{p}_{1111}$ and $\tilde{p}_{1414}$ are different from zero.
\end{remark}

%

\begin{thm}\label{thm:mixturesF81} The space of mixtures $\E_4^{F81}$ for quartets under the F81 model has dimension $15$ and is defined by the defining equations of $\F^{F81}_4\cap\Z_4$, along with the nine equations of Proposition \ref{prop:modelEq12|34}.
\end{thm}

\begin{proof}
Note that the $232$ linearly independent equations defining $\F^{F81}_4\cap\Z_4$ are also linearly independent from the $9$ equations in Proposition \ref{prop:modelEq12|34}. Thus, if $X$ is a linear space defined by all these equations, it has dimension $15=256-232-9$. Now, $X$ contains each $CV_T^{F81}$ due to Propositions \ref{prop:quartet-eq} and \ref{prop:modelEq12|34}.  Thus, since the space of mixtures is $\mathcal{E}_4^{F81}=\langle CV_{T_{12}}^{F81}\cup CV_{T_{13}}^{F81} \cup CV_{T_{14}}^{F81} \rangle$ and $X$ contains this linear span, we get $\E^{F81}_4\subseteq X$. The dimensions of these two linear spaces agree by \cite{CasSteel}, and this concludes the proof.
\end{proof}

\section{Local complete intersection}\label{sec:complete_intersection}

 From the point of view of applications, many of the generators of $I_T^\MM$ are non-informative. Therefore, in this section, we are interested in computing the \emph{relevant} phylogenetic invariants; in the sense that they cut out the variety locally at points of biological interest. More precisely, we want a complete intersection (CI) $X_T^\MM$ that coincides with the phylogenetic variety $CV_T^\MM$ in an open set containing the no-evolution point. 
Thus, we obtained a complete intersection for $n=3$ and a local complete intersection at the no-evolution point for $n=4$.
Moreover, we show that the equations of the submodels can be recovered from the equations of the bigger model after imposing only the symmetry equations.

\subsection{Complete intersections for tripods}
In the case of tripods, or any star tree, coordinates $\tilde{p}$ provide a monic monomial parametrization of the model, recall \Cref{tilde_rho}. 
Therefore, we have theoretical and computational tools from toric geometry to help us find the equations that cut out both $CV_T^\MM$. By \Cref{prop:eqtripod}, for F81, the phylogenetic variety for the tripod is a hypersurface in its linear span, hence it is already a complete intersection.
On the other hand, for F84, we have a 7-dimensional phylogenetic variety $CV_T^{F84}$ in the 12-dimensional linear space $\E_3^{F84}$ but $I_T^{F84}$ is minimally generated by $15$ polynomials.

\begin{prop}\label{prop:CItripod}
     Let $T$ be a tripod that evolves under the F84 model. The variety $X_T^{F84}\subset \E_3^{F84}$ defined by the following 3 quadrics and 2 cubics 
    \begin{align*}
&\tilde{p}_{222}\tilde{p}_{441}-\tilde{p}_{221}\tilde{p}_{442},\qquad
\tilde{p}_{222}\tilde{p}_{414}-\tilde{p}_{212}\tilde{p}_{424},\qquad
\tilde{p}_{222}\tilde{p}_{144}-\tilde{p}_{122}\tilde{p}_{244}, \\  
&\tilde{p}_{144}\tilde{p}_{414}\tilde{p}_{441}-\tilde{p}_{111}\tilde{p}_{444}^2,  \qquad 
\tilde{p}_{442}\tilde{p}_{424}\tilde{p}_{244}-\tilde{p}_{222}\tilde{p}_{444}^2.
    \end{align*}
is a complete intersection that cuts out $CV_T^{F84}$ in an open set containing the no-evolution point. 
Furthermore, $CV_T^{F84}$ is an irreducible component of $X_T^{F84}$
and
$CV_T^{F84}=\overline{X_T^{F84}\backslash \mathcal{V}(\tilde{p}_{222}\tilde{p}_{444})}$.
\end{prop}
\begin{proof}
   The proof is analogous to the case of tripods for the TN93 model in \cite[Proposition 5.3]{CHT}. 
\end{proof}


The following result proves that the equations of the submodels F81 and F84 can be obtained from the equations of TN93 by imposing the corresponding symmetry equations.

\begin{prop}\label{prop:tripodCI}
Let $T$ be a tripod, and $\MM$ be F84 or F81. Then $X_T^{TN93}\cap\F_3^\MM=X_T^\MM$ and $CV_T^{TN93}\cap\F_3^\MM=CV_T^\MM$. 
\end{prop}
\begin{proof}
We can rewrite the $9$ equations of the complete intersection $X_T^{TN93}$ displayed in \cite[Proposition 5.3]{CHT} in coordinates $\tilde{p}_{i_1i_2i_3}$ as follows:

\begin{align*}
&g_1=\tilde{p}_{222}\tilde{p}_{441}-\tilde{p}_{221}\tilde{p}_{442},\qquad
g_2=\tilde{p}_{222}\tilde{p}_{414}-\tilde{p}_{212}\tilde{p}_{424},\qquad
g_3=\tilde{p}_{222}\tilde{p}_{144}-\tilde{p}_{122}\tilde{p}_{244}, \\
&g_4=\tilde{p}_{332}\tilde{p}_{441}-\tilde{p}_{331}\tilde{p}_{442},\qquad
g_5=\tilde{p}_{323}\tilde{p}_{414}-\tilde{p}_{313}\tilde{p}_{424},\qquad
g_6=\tilde{p}_{233}\tilde{p}_{144}-\tilde{p}_{133}\tilde{p}_{244},\\
&g_7=\tilde{p}_{144}\tilde{p}_{414}\tilde{p}_{441}-\tilde{p}_{111}\tilde{p}_{444}^2,  \qquad 
g_8=\tilde{p}_{133}\tilde{p}_{313}\tilde{p}_{331}-\tilde{p}_{111}\tilde{p}_{333}^2,\qquad
g_9=\tilde{p}_{332}\tilde{p}_{323}\tilde{p}_{233}-\tilde{p}_{222}\tilde{p}_{333}^2.
    \end{align*}

It is straightforward to check when applying the symmetry equations defining $\F^{F84}_3$ to the previous list, polynomials $g_4,g_5,g_6$ vanish, and $g_7=g_8$. The surviving $5$ equations cut out the complete intersection $X_T^{F84}$ in \Cref{prop:CItripod}. If we additionally add the symmetry equations cutting out $\F_3^{81}$, only $g_7$ remains. A similar check can be performed for all equations defining $CV_T^{TN93}$ in \cite{CHT}[SM1].
\end{proof}

\subsection{Local complete intersections for quartets}

Now we turn our attention to the phylogenetic variety for quartets. Note that we no longer have a monomial parametrization, so we need to find alternative strategies to tackle the problem. As in the case of TN93 in \cite{CHT}, we are no longer able to explicitly compute all equations defining $CV_T^{F84}$. 

Analogously to \cite{CHT} for TN93, we will use rank conditions to obtain the equations defining a complete intersection locally around the no-evolution point. The rest of the section is devoted to imposing the symmetry equations in $\F_4^\MM$ to the flattening matrices (as in \Cref{sec:rank}) to obtain the relevant equations cutting the phylogenetic variety. We assume that $T=T_{12}$ for the rest of the section and consider $A=\Flat_{12\mid 34}(\bar{p})$.

\begin{thm}\label{thm:CIF84}
  Consider the tree $T=T_{12}$ evolving under the F84 model. Let $X_T^{F84}\subset\E_T^{F84}$ be the variety cut out by the following $38$ equations ($26$ quadrics, $10$ cubics, and $2$ quartics):
  \begin{itemize}
      \item [(a)] $10$ equations ($6$ quadrics and $4$ cubics) arising from tripods: 
      \begin{itemize}
          \item[(a.1)]$\tilde{p}_{1222}\tilde{p}_{1441}-\tilde{p}_{1221}\tilde{p}_{1442}$, $\tilde{p}_{1222}\tilde{p}_{1414}-\tilde{p}_{1212}\tilde{p}_{1424}$, $\tilde{p}_{1222}\tilde{p}_{1144}-\tilde{p}_{1122}\tilde{p}_{1244}$,\\ $\tilde{p}_{1144}\tilde{p}_{1414}\tilde{p}_{1441}-\tilde{p}_{1111}\tilde{p}_{1444}^2$, $\tilde{p}_{1442}\tilde{p}_{1424}\tilde{p}_{1244}-\tilde{p}_{1222}\tilde{p}_{1444}^2$,
         \item[(a.2)] $\tilde{p}_{2212}\tilde{p}_{4411}-\tilde{p}_{2211}\tilde{p}_{4412}$, $\tilde{p}_{2212}\tilde{p}_{4114}-\tilde{p}_{2112}\tilde{p}_{4214}$, $\tilde{p}_{2212}\tilde{p}_{1414}-\tilde{p}_{1212}\tilde{p}_{2414}$,\\ $\tilde{p}_{1414}\tilde{p}_{4114}\tilde{p}_{4411}-\tilde{p}_{1111}\tilde{p}_{4414}^2$, $\tilde{p}_{4412}\tilde{p}_{4214}\tilde{p}_{2414}-\tilde{p}_{2212}\tilde{p}_{4414}^2$;
      \end{itemize}
       \item [(b)] $15$ quadrics from rank 1 conditions:
      \begin{itemize}
          \item $\left|A_{14,ab}^{14,ij}\right|$ with $ab\in\{41,24,42,44\}$, $ij\in\{41,24,42\}$,
          \item $\left|A_{12,cd}^{12,21}\right|$ with $cd\in\{44,22,21\}$;
      \end{itemize}
         \item [(c)] $4$ equations (1 quadric and 3 cubics) from rank 2 conditions: 
      \begin{itemize}
          \item $1$ quadric: $(1/\tilde{p}_{1111})\left|A_{11,12,21}^{11,12,22}\right|$,
          \item $3$ cubics: $\left|A_{11,12,ab}^{11,12,22}\right|$ with $ab\in\{22,33,44\}$; 
      \end{itemize}
            \item [(d)] $9$ equations ($4$ quadrics, $3$ cubics, and $2$ quartics) from rank $3$ conditions: 
       \begin{itemize}
          \item $4$ quadrics: $(1/\tilde{p}_{1111}\tilde{p}_{1414})\left\vert A_{11,12,13,21}^{11,12,13,33}\right\vert$, $(1/\tilde{p}_{1111}\tilde{p}_{1212})\left\vert A_{11,12,13,ab}^{11,12,13,33}\right\vert$ with $ab\in\{31,23,32\}$,
          \item $3$ cubics: $(1/\tilde{p}_{1414})\left\vert A_{11,12,13,ab}^{11,12,13,33}\right\vert$ with $ab\in\{22,33\}$, $(1/\tilde{p}_{1414})\left|A_{11,12,14,22}^{11,12,14,44}\right|$,
        \item $2$ quartics: $\left|A_{11,12,1i,ii}^{11,12,1i,ii}\right|$ with $i\in\{3,4\}$.
      \end{itemize}

  \end{itemize}
 Then, $CV_T^{F84}$ is an irreducible component of $X_T^{F84}$ and $CV_T^{F84}=X_T^{F84}$ around the no-evolution point. 
\end{thm}
\begin{proof}
We first prove that $X_T^{F84}$ contains $CV_T^{F84}$ by arguing that the polynomials in the statement vanish on $CV_T^{F84}$. 
Equations in (a) hold for any quartet because marginalizing over leaf $l_1$ in (a.1) or leaf $l_3$ in (a.2) produces the tripod equations in \Cref{prop:CItripod}; see \cite[Lemma 4.2, Lemma 5.9]{CHT} for further details.

Consider the vanishing minors of the flattening matrix $A=\mathrm{Flat}_{12\mid 34}(\bar{p})$ of sizes $2, 3$ and $4$, respectively, that contain a diagonal submatrix with non-vanishing determinant if we assume that $\tilde{p}_{1111}\tilde{p}_{1212}\tilde{p}_{1414}\neq 0$. They yield equations (b), (c), and (d), up to factoring out the previous variables. \Cref{tab:flatF84rk1,tab:flatF84rk2a,tab:flatF84rk2b,tab:flatF84rk3} display the rank blocks with the distinguished diagonal submatrices.

Since $CV_T^{F84}$ is a $11$-dimensional space (Proposition \ref{prop:dim}) that lives inside the $49$-dimensional linear space $\E_T^{F84}$, its codimension in $\E_T^{F84}$ is $38$. 
Using {\fontfamily{qcr}\selectfont{Sage}} \cite{sage}, we see that the rank of the Jacobian of the $38$ above equations at the no-evolution point is $38$. Hence, locally, they define the same variety. A copy of the code used can be found in Appendix \ref{app:code84}. 

Now, since $CV_T^{84}$ is irreducible and the irreducible component of $X_T^{84}$ containing the no-evolution point $\rho$ has dimension $11$, by the above computations, we see that $CV_T^{84}$ is an irreducible component of $X_T^{84}$.
\end{proof}

\begin{definition}
$I_A^\MM$ denotes the ideal of $2, 3$ and $4$-minors of blocks of $A$ of ranks $1, 2$ and $3$, respectively, highlighted in the flattening matrices displayed in \Cref{app:flat} and \Cref{app:flatF84}. 
\end{definition}

\begin{cor}\label{cor:satF84}
Let $J_A^{F84}$ be the ideal generated by polynomials (b), (c) and (d) in \Cref{thm:CIF84} and polynomials in \Cref{prop:4quartetEq84}. Then
$$[J_A^{F84}:\left(\tilde{p}_{1111}\tilde{p}_{1212}\tilde{p}_{1414}\right)^\infty]=[I_A^{F84}:\left(\tilde{p}_{1111}\tilde{p}_{1212}\tilde{p}_{1414}\right)^\infty]\subsetneq I_T^{F84}. $$
\end{cor}
\begin{proof}
Equality of $I_A^{F84}$ and $J_A^{F84}$ up to saturation by $\left(\tilde{p}_{1111}\tilde{p}_{1212}\tilde{p}_{1414}\right)$ follows from the proof of \Cref{thm:CIF84} and \Cref{rem:nonzero_F84}. The last inclusion follows directly from \Cref{thm:ARflat}.
\end{proof}

\begin{thm}\label{thm:CI}
  The phylogenetic variety $CV_T^{F81} \subset \mathcal{E}_T^{F81}$ is cut out by $10$ quadrics and $23$ cubics. Locally around the no-evolution point, it is defined by the following seven equations of degrees $2$ and $3$, for $T=T_{12}$: 
  \begin{itemize}
      \item [(a)] $\tilde{p}_{1144}\tilde{p}_{1414}\tilde{p}_{1441} - \tilde{p}_{1111}\tilde{p}_{1444}^2$; 
      \item [(b)] $\tilde{p}_{1441}\tilde{p}_{4141}\tilde{p}_{4411} - \tilde{p}_{1111}\tilde{p}_{4441}^2$; 
      \item [(c)]
    $(\pi_3^3+\pi_4^3)\tilde{p}_{1111}\tilde{p}_{3333}-(\pi_3^3+\pi_4^3-\pi_3\pi_4\pi_{34}^2)\tilde{p}_{1111}\tilde{p}_{2444}-\pi_3\pi_4\pi_{34}^2\tilde{p}_{1144}\tilde{p}_{4411}=0$;
      \item [(d)]
      $\tilde{p}_{1414}\tilde{p}_{4141} - \tilde{p}_{4114}\tilde{p}_{1441}$;
      \item [(e)]
      $\tilde{p}_{1414}\tilde{p}_{4144} - \tilde{p}_{4114}\tilde{p}_{1444}$;
      \item [(f)]
      $\tilde{p}_{1414}\tilde{p}_{4441} - \tilde{p}_{4414}\tilde{p}_{1441}$;
      \item [(g)]
      $\tilde{p}_{1414}\tilde{p}_{2444} - \tilde{p}_{4414}\tilde{p}_{1444}$.
  \end{itemize}
  Moreover, $CV_T^{F81}$ is an irreducible component of the algebraic variety $X_T^{F81}$ that is cut out by these equations. More precisely, $CV_T^{F81}=\overline{X_T^{F81}\backslash \mathcal{V}(\tilde{p}_{1111}\tilde{p}_{1414}\tilde{p}_{4141})}$.
\end{thm}

\begin{proof}
The proof is analogous to \Cref{thm:CIF84}. Indeed, (a) and (b) hold in $CV_T^{F84}$ due to the marginalization procedure, (c) is given in \Cref{prop:non-binomialEq} and (d)-(g) are derived from rank 1 conditions $A^{12,21}_{12,21},A^{13,23}_{13,31},A^{12,21}_{12,33},A^{13,23}_{13,33}$. As the variety $CV_T^{F81}$ has dimension $6$ (Proposition \ref{prop:dim}) and lives inside the 13-dimensional space $\mathcal{E}_T^{F81}$, its codimension is 7. A copy of the code we used to check that the rank of the Jacobian at the no-evolution point is 7 can be found in Appendix \ref{app:code81}. In the case of the F81 model, the vanishing ideal $I_T^{F81}$ of $CV_T^{F81}$ can be explicitly computed with \texttt{Macaulay2} and it can be checked that $\left[\mathcal{I}(X_T^{F81}):\left(\tilde{p}_{1111}\tilde{p}_{1414}\tilde{p}_{4141}\right)^\infty\right]=I_T^{F81}$.
\end{proof}

\begin{remark}\rm Note that for the Jukes-Cantor model, that is, uniform $\pi$ (which is not a generic $\pi$ and some of our results would not hold for it), the ideal is also generated by $10$ quadrics and $23$ cubics; see the Small Trees Webpage \cite{Smalltrees}.
    
\end{remark}

\begin{cor}\label{cor:satF81}
Let $J_A^{F81}$ be the ideal generated by polynomials (c)-(g) in \Cref{thm:CI} and linear polynomials in Propositions \ref{prop:4quartetEq} and \ref{prop:non-binomialEq}. Then
$$[J_A^{F81}:\left(\tilde{p}_{1111}\tilde{p}_{1414}\right)^\infty]=[I_A^{F81}:\left(\tilde{p}_{1111}\tilde{p}_{1414}\right)^\infty]\subsetneq I_T^{F81},$$
and the two ideals coincide up to degree two.
\end{cor}
\begin{proof} 
Equality of $I_A^{F81}$ and $J_A^{F81}$ up to saturation by $(\tilde{p}_{1111}\tilde{p}_{1414})$ follows from \Cref{rem:nonzero_F81}. \texttt{Macaulay2} computations prove the last inclusion. 
\end{proof}

\begin{remark}\rm 
Note that $\tilde{p}_{1111}\tilde{p}_{1414}\neq 0$ (or $\tilde{p}_{1111}\tilde{p}_{1212}\tilde{p}_{1414}\neq 0$ in the case of F84) is equivalent to the diagonal submatrix $A_{1234}^{1234}$ having non-zero determinant. This is precisely the matrix one obtains by marginalizing the tensor $p$ over leaves $l_1$ and $l_3$, and it is analogous to the result obtained in \cite[Theorem 5.3]{casfer2024}.
Moreover, we can equivalently define $I_A^\MM$ as $I_A^{TN93}+\mathcal{I}(\F_4^\MM)$.
Therefore, Corollaries \ref{cor:satF84} and \ref{cor:satF81} provide an analogous statement to the result stated in \cite[Corollary 5.8]{casfer2024} for equivariant models: $$[I_A^{TN93}+\mathcal{I}(\F_4^\MM):\left(\det\,A_{1234}^{1234}\right)^\infty]\subset I_T^\MM.$$  
\end{remark}

\begin{cor}\label{cor:quartetCI}
Let $T$ be a quartet and let $CV_T^{TN93}$ be the phylogenetic variety for the $TN93$ model. Then, for $\MM=$ F81 or F84, $X_T^\MM=X_T^{TN93}\cap \F_4^\MM$ and $CV_T^{\MM}$ is an irreducible component of $CV_T^{TN93}\cap \F_4^{\MM}.$ 
\end{cor}

\begin{proof}
Let $X_T^{TN93}$ be the complete intersection obtained in \cite[Theorem 5.15]{CHT} for the quartet evolving under $TN93$. 
The equations defining $X_T^{TN93}$ were rank conditions on certain minors together with $18$ equations coming from the tripods. In Sections \ref{sec:rank} and \ref{sec:mixtures} we have seen that imposing the symmetries of $\F_4^\MM$ on these rank conditions, we obtain equations $(c)-(f)$ above and the equations defining $\mathcal{E}_T^{\MM}$. 
On the other hand, $(a)$ and $(b)$ above are also obtained by imposing the symmetries of $\F_4^\MM$ on the equations arising from tripods in $X_T^{TN93}$. Thus, if $X_T^\MM$ is defined as in Theorem \ref{thm:CIF84} for F84 or in Theorem \ref{thm:CI} for F81, we have $X_T^\MM=X_T^{TN93}\cap \F_4^\MM$. 
Based on the results of \cite[Theorem 5.15]{CHT}, $CV_T^{TN93}$ is the irreducible component of $X_T^{TN93}$ containing the no evolution point $\rho$. Thus, $X_T^{TN93}$ decomposes into irreducible components as $C_1 \cup C_2\dots \cup C_l$ with $C_1=CV_T^{TN93}$ and $X_T^\MM=(C_1\cap \F_4^\MM) \cup \dots \cup (C_l\cap \F_4^\MM)$. 
As in Theorem \ref{thm:CI} we have seen that $CV_T^\MM$ is the irreducible component of $X_T^\MM$ containing $\rho$, we get that $CV_T^\MM$ is an irreducible component of $C_1\cap \F_4^\MM$ as we wanted to prove.
\end{proof}

\section{Discussion}

We have used phylogenetic invariants for tripods and quartets evolving under the TN93 model computed in \cite{CHT} to obtain invariants for F81 and F84. More precisely, we have described the spaces of phylogenetic mixtures and provided local complete intersections that cut out the phylogenetic varieties at the no-evolution point.
For tripods and quartets we have proved that only symmetry equations of the submodel are needed to obtain all equations from those of TN93, even in the case where the space of mixtures requires more involved (non-binomial) linear equations (Proposition \ref{prop:CItripod}, Corollary \ref{cor:quartetCI}). A first natural question to ask is whether this will be true for a larger number of leaves.

For tripods, we prove that the algebraic variety of each submodel coincides with $CV_3^{TN93}$ intersected with the space $\F_3^\MM$ defined by symmetry equations. 
In the case of quartets, we do not know whether $CV_4^{TN93}\cap\F_4^\MM$ is irreducible for $\MM=F81$ or $F84$, but we are able to show that $CV_4^\MM$ is an irreducible component of this intersection. 
{Therefore, another natural question is whether the intersection of the variety $CV_T^{TN93}$ with $\F_n^\MM$ is irreducible only for star trees (thus generalizing the result for tripods), as it is conjectured for equivariant models in \cite[Section 6]{casfer2024}.}

As for practical applications, edge invariants (rank conditions on the flattening matrix associated to a tree topology) for quartets are already useful for phylogenetic reconstruction purposes due to the existence of quartet-based methods and {numerical methods that test for rank}. On the other hand, understanding which phylogneetic invariants describe the space of mixtures is crucial for model selection goals. 

The study of TN93 and its submodels showcases the potential of the ATR framework developed in \cite{CHT} to compute phylogenetic invariants for time-reversible evolutionary models. There is a long way to go to be able to provide model selection and phylogenetic reconstruction methods based on phylogenetic invariants that are competitive with currently available implemented software. Nevertheless, a first step in this direction is to apply these techniques to other models, especially in the case of amino acid sequence data.

Note that, because we wanted to take advantage of the equations of the bigger model, in this paper we have used the unique (up to scaling) $\pi$-orthogonal change of coordinates for TN93 in \cite{CHT}. However, the submodels considered here have other $\pi$-orthogonal bases that may produce linear equations that are easier to generalize to trees with a larger number of leaves or even to models with more states. This will be the topic of a forthcoming paper for the Equal-Input model.

\bibliographystyle{plain}
\bibliography{biblio}

\appendix
\newgeometry{left=3cm,bottom=2.5cm}
\begin{landscape}
\section{Flattening tables F81}\label{app:flat}
\begin{table}[h]
    \tiny
$\begin{NiceArray}{c|c|cccc|c|c|cc|c|cccc|cc}
\CodeBefore
\Body
& (1,1) & (1,4)  & (4,1) & (2,4) & (4,2) & (4,4) & (2,2) & (1,2) & (2,1)& (3,3) & (1,3) & (3,1) & (2,3) & (3,2)\\
\hline
(1,1) & \tilde{p}_{1111} & 0  & 0 & 0 & 0 &  \pp_{1144}\tilde{p}_{1144} & \pp_{1122}\tilde{p}_{1144} & 0&0& \pp_{1133}\tilde{p}_{1144} & 0 & 0 & 0 & 0 \\
\hline
(1,4) & 0 & \pp_{1144}\tilde{p}_{1414} & \pp_{1144}\tilde{p}_{1441} & \pp_{1244}\tilde{p}_{1444} & \pp_{1244}\tilde{p}_{1444} & \pp_{1444}\tilde{p}_{1444} & 0 & 0 & 0 & 0 & 0 & 0 & 0 & 0 \\
(4,1) & 0 & \pp_{1144}\tilde{p}_{4114} & \pp_{1144}\tilde{p}_{4141} &\pp_{1244}\tilde{p}_{4144} & \pp_{1244}\tilde{p}_{4144} &  \pp_{1444}\tilde{p}_{4144} &0 & 0 & 0 & 0 & 0 & 0 & 0 & 0\\
(2,4) & 0 & \pp_{1244}\tilde{p}_{4414} & \pp_{1244}\tilde{p}_{4441}  &\pp_{2244} \tilde{p}_{2424} & \pp_{2244} \tilde{p}_{2442} & \pp_{2444}\tilde{p}_{2444} & 0 & 0 & 0 & 0 & 0 & 0 & 0 & 0 \\
(4,2) & 0 & \pp_{1244}\tilde{p}_{4414} & \pp_{1244}\tilde{p}_{4441} & \pp_{2244}\tilde{p}_{2442} & \pp_{2244}\tilde{p}_{2424}  & \pp_{2444}\tilde{p}_{2444} & 0 & 0 & 0 & 0 & 0 & 0 & 0 & 0\\
\hline
(4,4) & \pp_{1144}\tilde{p}_{4411} & \pp_{1444}\tilde{p}_{4414} & \pp_{1444}\tilde{p}_{4441} & \pp_{2444}\tilde{p}_{2444} &\pp_{2444}\tilde{p}_{2444} & \pp_{4444}\tilde{p}_{4444}  & \pp_{2244}\tilde{p}_{2244} & \pp_{1244}\tilde{p}_{4414} & 
\pp_{1244}\tilde{p}_{4441} &\tilde{p}_{3344} & 0  & 0 & 0 & 0 \\
\hline
(2,2) & \pp_{1122}\tilde{p}_{4411} & 0 & 0 & 0 & 0& 
\pp_{2244}\tilde{p}_{2244} & \pp_{2222}\tilde{p}_{2222} &
\pp_{1222}\tilde{p}_{4414} & \pp_{1222}\tilde{p}_{4441} & \pp_{2233}\tilde{p}_{2233} & 0 & 0 & 0 & 0 \\
\hline
(1,2) & 0 & 0 & 0 & 0 & 0&  \pp_{1244}\tilde{p}_{1444} & \pp_{1222}\tilde{p}_{1444} & \pp_{1122}\tilde{p}_{1414} & \pp_{1122}\tilde{p}_{1441} & \pp_{1233}\tilde{p}_{1444} & 0 & 0 & 0 & 0\\
(2,1) & 0 & 0  & 0 & 0 & 0 & \pp_{1244}\tilde{p}_{4144} & \pp_{1222}\tilde{p}_{4144} & \pp_{1122}\tilde{p}_{4114} &
\pp_{1122}\tilde{p}_{4141} & \pp_{1233}\tilde{p}_{4144} & 0 & 0 & 0 & 0 \\
\hline
(3,3) &\pp_{1133}\tilde{p}_{4411} & 0 & 0 & 0 & 0 & \tilde{p}_{3344} & \pp_{2233}\tilde{p}_{2233} & \pp_{1233}\tilde{p}_{4414}  &  \pp_{1233}\tilde{p}_{4441}  &\pp_{3333}\tilde{p}_{3333} & \pp_{1333}\tilde{p}_{4414} & \pp_{1333}\tilde{p}_{4441} & \pp_{2333}\tilde{p}_{2444} & \pp_{2333}\tilde{p}_{2444} \\
\hline
(1,3) & 0 &  0 & 0 & 0 & 0 & 0 & 0 & 0& 0& \pp_{1333}\tilde{p}_{1444} & \pp_{1133}\tilde{p}_{1414}& 
\pp_{1133}\tilde{p}_{1441} & \pp_{1233}\tilde{p}_{1444} & \pp_{1233}\tilde{p}_{1444}\\
(3,1)  & 0 &  0 & 0 & 0 & 0 & 0 & 0 & 0& 0& \pp_{1333}\tilde{p}_{4144} & \pp_{1133}\tilde{p}_{4114} & \pp_{1133}\tilde{p}_{4141} & 
\pp_{1233}\tilde{p}_{4144} & \pp_{1233}\tilde{p}_{4144} \\
(2,3)& 0 &  0 & 0 & 0 & 0 & 0 & 0 & 0& 0&
\pp_{2333} \tilde{p}_{2444} & 
\pp_{1233}\tilde{p}_{4414} & 
\pp_{1233}\tilde{p}_{4441} & 
\pp_{2233}\tilde{p}_{2323} & 
\pp_{2233}\tilde{p}_{2332} \\
(3,2) & 0 &  0 & 0 & 0 & 0 & 0 & 0 & 0& 0&
\pp_{2333}\tilde{p}_{2444} & 
\pp_{1233}\tilde{p}_{4414} & 
\pp_{1233}\tilde{p}_{4441} & 
\pp_{2233}\tilde{p}_{2332} & 
\pp_{2233}\tilde{p}_{2323} 
\end{NiceArray}$
\caption{Flattening matrix of $\Bar{p}$ for a quartet $T=12\mid 34$ under the F81 model.}
\label{tab:flat}
\end{table}

\begin{table}[h]
    \tiny
$\begin{NiceArray}{c|c|cccc|c|c|cc|c|cccc|cc}
\CodeBefore
\rectanglecolor{gray!30}{3-3}{7-6} 
\rectanglecolor{gray!30}{11-12}{15-15} 
\rectanglecolor{gray!30}{7-9}{11-10} 
\rectanglecolor{gray!80}{3-3}{3-3} 
\rectanglecolor{gray!80}{12-12}{12-12}
\rectanglecolor{gray!80}{9-9}{9-9}
\Body
& (1,1) & (1,4)  & (4,1) & (2,4) & (4,2) & (4,4) & (2,2) & (1,2) & (2,1)& (3,3) & (1,3) & (3,1) & (2,3) & (3,2)\\
\hline
(1,1) & \tilde{p}_{1111} & 0  & 0 & 0 & 0 &  \pp_{1144}\tilde{p}_{1144} & \pp_{1122}\tilde{p}_{1144} & 0&0& \pp_{1133}\tilde{p}_{1144} & 0 & 0 & 0 & 0 \\
\hline
(1,4) & 0 & \pp_{1144}\tilde{p}_{1414} & \pp_{1144}\tilde{p}_{1441} & \pp_{1244}\tilde{p}_{1444} & \pp_{1244}\tilde{p}_{1444} & \pp_{1444}\tilde{p}_{1444} & 0 & 0 & 0 & 0 & 0 & 0 & 0 & 0 \\
(4,1) & 0 & \pp_{1144}\tilde{p}_{4114} & \pp_{1144}\tilde{p}_{4141} &\pp_{1244}\tilde{p}_{4144} & \pp_{1244}\tilde{p}_{4144} &  \pp_{1444}\tilde{p}_{4144} &0 & 0 & 0 & 0 & 0 & 0 & 0 & 0\\
(2,4) & 0 & \pp_{1244}\tilde{p}_{4414} & \pp_{1244}\tilde{p}_{4441}  &\pp_{2244} \tilde{p}_{2424} & \pp_{2244} \tilde{p}_{2442} & \pp_{2444}\tilde{p}_{2444} & 0 & 0 & 0 & 0 & 0 & 0 & 0 & 0 \\
(4,2) & 0 & \pp_{1244}\tilde{p}_{4414} & \pp_{1244}\tilde{p}_{4441} & \pp_{2244}\tilde{p}_{2442} & \pp_{2244}\tilde{p}_{2424}  & \pp_{2444}\tilde{p}_{2444} & 0 & 0 & 0 & 0 & 0 & 0 & 0 & 0\\
\hline
(4,4) & \pp_{1144}\tilde{p}_{4411} & \pp_{1444}\tilde{p}_{4414} & \pp_{1444}\tilde{p}_{4441} & \pp_{2444}\tilde{p}_{2444} &\pp_{2444}\tilde{p}_{2444} & \pp_{4444}\tilde{p}_{4444}  & \pp_{2244}\tilde{p}_{2244} & \pp_{1244}\tilde{p}_{4414} & 
\pp_{1244}\tilde{p}_{4441} &\tilde{p}_{3344} & 0  & 0 & 0 & 0 \\
\hline
(2,2) & \pp_{1122}\tilde{p}_{4411} & 0 & 0 & 0 & 0& 
\pp_{2244}\tilde{p}_{2244} & \pp_{2222}\tilde{p}_{2222} &
\pp_{1222}\tilde{p}_{4414} & \pp_{1222}\tilde{p}_{4441} & \pp_{2233}\tilde{p}_{2233} & 0 & 0 & 0 & 0 \\
\hline
(1,2) & 0 & 0 & 0 & 0 & 0&  \pp_{1244}\tilde{p}_{1444} & \pp_{1222}\tilde{p}_{1444} & \pp_{1122}\tilde{p}_{1414} & \pp_{1122}\tilde{p}_{1441} & \pp_{1233}\tilde{p}_{1444} & 0 & 0 & 0 & 0\\
(2,1) & 0 & 0  & 0 & 0 & 0 & \pp_{1244}\tilde{p}_{4144} & \pp_{1222}\tilde{p}_{4144} & \pp_{1122}\tilde{p}_{4114} &
\pp_{1122}\tilde{p}_{4141} & \pp_{1233}\tilde{p}_{4144} & 0 & 0 & 0 & 0 \\
\hline
(3,3) &\pp_{1133}\tilde{p}_{4411} & 0 & 0 & 0 & 0 & \tilde{p}_{3344} & \pp_{2233}\tilde{p}_{2233} & \pp_{1233}\tilde{p}_{4414}  &  \pp_{1233}\tilde{p}_{4441}  &\pp_{3333}\tilde{p}_{3333} & \pp_{1333}\tilde{p}_{4414} & \pp_{1333}\tilde{p}_{4441} & \pp_{2333}\tilde{p}_{2444} & \pp_{2333}\tilde{p}_{2444} \\
\hline
(1,3) & 0 &  0 & 0 & 0 & 0 & 0 & 0 & 0& 0& \pp_{1333}\tilde{p}_{1444} & \pp_{1133}\tilde{p}_{1414}& 
\pp_{1133}\tilde{p}_{1441} & \pp_{1233}\tilde{p}_{1444} & \pp_{1233}\tilde{p}_{1444}\\
(3,1)  & 0 &  0 & 0 & 0 & 0 & 0 & 0 & 0& 0& \pp_{1333}\tilde{p}_{4144} & \pp_{1133}\tilde{p}_{4114} & \pp_{1133}\tilde{p}_{4141} & 
\pp_{1233}\tilde{p}_{4144} & \pp_{1233}\tilde{p}_{4144} \\
(2,3)& 0 &  0 & 0 & 0 & 0 & 0 & 0 & 0& 0&
\pp_{2333} \tilde{p}_{2444} & 
\pp_{1233}\tilde{p}_{4414} & 
\pp_{1233}\tilde{p}_{4441} & 
\pp_{2233}\tilde{p}_{2323} & 
\pp_{2233}\tilde{p}_{2332} \\
(3,2) & 0 &  0 & 0 & 0 & 0 & 0 & 0 & 0& 0&
\pp_{2333}\tilde{p}_{2444} & 
\pp_{1233}\tilde{p}_{4414} & 
\pp_{1233}\tilde{p}_{4441} & 
\pp_{2233}\tilde{p}_{2332} & 
\pp_{2233}\tilde{p}_{2323} 
\end{NiceArray}$
\caption{Flattening matrix of $\Bar{p}$ for a quartet $T=12\mid 34$. The highlighted submatrices have rank 1 and the entries highlighted in dark grey are generically non-zero. Therefore, the 18 quadratic edge invariants arise from considering all 2-minors of the grey submatrices containing the dark grey minor. }
\label{tab:flatrank2}
\end{table}

\begin{table}[h]
    \tiny
$\begin{NiceArray}{c|c|cccc|c|c|cc|c|cccc|cc}
\CodeBefore
\rectanglecolor{gray!30}{7-2}{11-2} 
\rectanglecolor{gray!30}{2-8}{2-8} 
\rectanglecolor{gray!30}{7-8}{11-9} 
\rectanglecolor{gray!80}{2-2}{2-2} 
\rectanglecolor{gray!80}{9-2}{9-2}
\rectanglecolor{gray!80}{2-9}{2-9}
\rectanglecolor{gray!80}{9-9}{9-9}
\Body
& (1,1) & (1,4)  & (4,1) & (2,4) & (4,2) & (4,4) & (2,2) & (1,2) & (2,1)& (3,3) & (1,3) & (3,1) & (2,3) & (3,2)\\
\hline
(1,1) & \tilde{p}_{1111} & 0  & 0 & 0 & 0 &  \pp_{1144}\tilde{p}_{1144} & \pp_{1122}\tilde{p}_{1144} & 0&0& \pp_{1133}\tilde{p}_{1144} & 0 & 0 & 0 & 0 \\
\hline
(1,4) & 0 & \pp_{1144}\tilde{p}_{1414} & \pp_{1144}\tilde{p}_{1441} & \pp_{1244}\tilde{p}_{1444} & \pp_{1244}\tilde{p}_{1444} & \pp_{1444}\tilde{p}_{1444} & 0 & 0 & 0 & 0 & 0 & 0 & 0 & 0 \\
(4,1) & 0 & \pp_{1144}\tilde{p}_{4114} & \pp_{1144}\tilde{p}_{4141} &\pp_{1244}\tilde{p}_{4144} & \pp_{1244}\tilde{p}_{4144} &  \pp_{1444}\tilde{p}_{4144} &0 & 0 & 0 & 0 & 0 & 0 & 0 & 0\\
(2,4) & 0 & \pp_{1244}\tilde{p}_{4414} & \pp_{1244}\tilde{p}_{4441}  &\pp_{2244} \tilde{p}_{2424} & \pp_{2244} \tilde{p}_{2442} & \pp_{2444}\tilde{p}_{2444} & 0 & 0 & 0 & 0 & 0 & 0 & 0 & 0 \\
(4,2) & 0 & \pp_{1244}\tilde{p}_{4414} & \pp_{1244}\tilde{p}_{4441} & \pp_{2244}\tilde{p}_{2442} & \pp_{2244}\tilde{p}_{2424}  & \pp_{2444}\tilde{p}_{2444} & 0 & 0 & 0 & 0 & 0 & 0 & 0 & 0\\
\hline
(4,4) & \pp_{1144}\tilde{p}_{4411} & \pp_{1444}\tilde{p}_{4414} & \pp_{1444}\tilde{p}_{4441} & \pp_{2444}\tilde{p}_{2444} &\pp_{2444}\tilde{p}_{2444} & \pp_{4444}\tilde{p}_{4444}  & \pp_{2244}\tilde{p}_{2244} & \pp_{1244}\tilde{p}_{4414} & 
\pp_{1244}\tilde{p}_{4441} &\tilde{p}_{3344} & 0  & 0 & 0 & 0 \\
\hline
(2,2) & \pp_{1122}\tilde{p}_{4411} & 0 & 0 & 0 & 0& 
\pp_{2244}\tilde{p}_{2244} & \pp_{2222}\tilde{p}_{2222} &
\pp_{1222}\tilde{p}_{4414} & \pp_{1222}\tilde{p}_{4441} & \pp_{2233}\tilde{p}_{2233} & 0 & 0 & 0 & 0 \\
\hline
(1,2) & 0 & 0 & 0 & 0 & 0&  \pp_{1244}\tilde{p}_{1444} & \pp_{1222}\tilde{p}_{1444} & \pp_{1122}\tilde{p}_{1414} & \pp_{1122}\tilde{p}_{1441} & \pp_{1233}\tilde{p}_{1444} & 0 & 0 & 0 & 0\\
(2,1) & 0 & 0  & 0 & 0 & 0 & \pp_{1244}\tilde{p}_{4144} & \pp_{1222}\tilde{p}_{4144} & \pp_{1122}\tilde{p}_{4114} &
\pp_{1122}\tilde{p}_{4141} & \pp_{1233}\tilde{p}_{4144} & 0 & 0 & 0 & 0 \\
\hline
(3,3) &\pp_{1133}\tilde{p}_{4411} & 0 & 0 & 0 & 0 & \tilde{p}_{3344} & \pp_{2233}\tilde{p}_{2233} & \pp_{1233}\tilde{p}_{4414}  &  \pp_{1233}\tilde{p}_{4441}  &\pp_{3333}\tilde{p}_{3333} & \pp_{1333}\tilde{p}_{4414} & \pp_{1333}\tilde{p}_{4441} & \pp_{2333}\tilde{p}_{2444} & \pp_{2333}\tilde{p}_{2444} \\
\hline
(1,3) & 0 &  0 & 0 & 0 & 0 & 0 & 0 & 0& 0& \pp_{1333}\tilde{p}_{1444} & \pp_{1133}\tilde{p}_{1414}& 
\pp_{1133}\tilde{p}_{1441} & \pp_{1233}\tilde{p}_{1444} & \pp_{1233}\tilde{p}_{1444}\\
(3,1)  & 0 &  0 & 0 & 0 & 0 & 0 & 0 & 0& 0& \pp_{1333}\tilde{p}_{4144} & \pp_{1133}\tilde{p}_{4114} & \pp_{1133}\tilde{p}_{4141} & 
\pp_{1233}\tilde{p}_{4144} & \pp_{1233}\tilde{p}_{4144} \\
(2,3)& 0 &  0 & 0 & 0 & 0 & 0 & 0 & 0& 0&
\pp_{2333} \tilde{p}_{2444} & 
\pp_{1233}\tilde{p}_{4414} & 
\pp_{1233}\tilde{p}_{4441} & 
\pp_{2233}\tilde{p}_{2323} & 
\pp_{2233}\tilde{p}_{2332} \\
(3,2) & 0 &  0 & 0 & 0 & 0 & 0 & 0 & 0& 0&
\pp_{2333}\tilde{p}_{2444} & 
\pp_{1233}\tilde{p}_{4414} & 
\pp_{1233}\tilde{p}_{4441} & 
\pp_{2233}\tilde{p}_{2332} & 
\pp_{2233}\tilde{p}_{2323} 
\end{NiceArray}$
\caption{Flattening matrix of $\Bar{p}$ for a quartet $T=12\mid 34$. The entries highlighted in dark gray form a (generically) non-vanishing 2-minor and the submatrix formed by the entries highlighted in (dark and light) gray has rank 2. Therefore, the 4 cubic edge invariants arise from considering all 3-minors of the grey submatrix containing the dark grey minor. }
\label{tab:flatrank3a}
\end{table}

\begin{table}[h]
    \tiny
$\begin{NiceArray}{c|c|cccc|c|c|cc|c|cccc|cc}
\CodeBefore
\rectanglecolor{gray!30}{7-2}{15-2} 
\rectanglecolor{gray!30}{2-11}{2-11} 
\rectanglecolor{gray!30}{7-9}{15-9} 
\rectanglecolor{gray!30}{7-11}{15-12} 
\rectanglecolor{gray!80}{2-2}{2-2} 
\rectanglecolor{gray!80}{9-2}{9-2}
\rectanglecolor{gray!80}{2-9}{2-9}
\rectanglecolor{gray!80}{9-9}{9-9}
\rectanglecolor{gray!80}{12-2}{12-2}
\rectanglecolor{gray!80}{2-12}{2-12}
\rectanglecolor{gray!80}{12-9}{12-9}
\rectanglecolor{gray!80}{9-12}{9-12}
\rectanglecolor{gray!80}{12-12}{12-12}
\Body
& (1,1) & (1,4)  & (4,1) & (2,4) & (4,2) & (4,4) & (2,2) & (1,2) & (2,1)& (3,3) & (1,3) & (3,1) & (2,3) & (3,2)\\
\hline
(1,1) & \tilde{p}_{1111} & 0  & 0 & 0 & 0 &  \pp_{1144}\tilde{p}_{1144} & \pp_{1122}\tilde{p}_{1144} & 0&0& \pp_{1133}\tilde{p}_{1144} & 0 & 0 & 0 & 0 \\
\hline
(1,4) & 0 & \pp_{1144}\tilde{p}_{1414} & \pp_{1144}\tilde{p}_{1441} & \pp_{1244}\tilde{p}_{1444} & \pp_{1244}\tilde{p}_{1444} & \pp_{1444}\tilde{p}_{1444} & 0 & 0 & 0 & 0 & 0 & 0 & 0 & 0 \\
(4,1) & 0 & \pp_{1144}\tilde{p}_{4114} & \pp_{1144}\tilde{p}_{4141} &\pp_{1244}\tilde{p}_{4144} & \pp_{1244}\tilde{p}_{4144} &  \pp_{1444}\tilde{p}_{4144} &0 & 0 & 0 & 0 & 0 & 0 & 0 & 0\\
(2,4) & 0 & \pp_{1244}\tilde{p}_{4414} & \pp_{1244}\tilde{p}_{4441}  &\pp_{2244} \tilde{p}_{2424} & \pp_{2244} \tilde{p}_{2442} & \pp_{2444}\tilde{p}_{2444} & 0 & 0 & 0 & 0 & 0 & 0 & 0 & 0 \\
(4,2) & 0 & \pp_{1244}\tilde{p}_{4414} & \pp_{1244}\tilde{p}_{4441} & \pp_{2244}\tilde{p}_{2442} & \pp_{2244}\tilde{p}_{2424}  & \pp_{2444}\tilde{p}_{2444} & 0 & 0 & 0 & 0 & 0 & 0 & 0 & 0\\
\hline
(4,4) & \pp_{1144}\tilde{p}_{4411} & \pp_{1444}\tilde{p}_{4414} & \pp_{1444}\tilde{p}_{4441} & \pp_{2444}\tilde{p}_{2444} &\pp_{2444}\tilde{p}_{2444} & \pp_{4444}\tilde{p}_{4444}  & \pp_{2244}\tilde{p}_{2244} & \pp_{1244}\tilde{p}_{4414} & 
\pp_{1244}\tilde{p}_{4441} &\tilde{p}_{3344} & 0  & 0 & 0 & 0 \\
\hline
(2,2) & \pp_{1122}\tilde{p}_{4411} & 0 & 0 & 0 & 0& 
\pp_{2244}\tilde{p}_{2244} & \pp_{2222}\tilde{p}_{2222} &
\pp_{1222}\tilde{p}_{4414} & \pp_{1222}\tilde{p}_{4441} & \pp_{2233}\tilde{p}_{2233} & 0 & 0 & 0 & 0 \\
\hline
(1,2) & 0 & 0 & 0 & 0 & 0&  \pp_{1244}\tilde{p}_{1444} & \pp_{1222}\tilde{p}_{1444} & \pp_{1122}\tilde{p}_{1414} & \pp_{1122}\tilde{p}_{1441} & \pp_{1233}\tilde{p}_{1444} & 0 & 0 & 0 & 0\\
(2,1) & 0 & 0  & 0 & 0 & 0 & \pp_{1244}\tilde{p}_{4144} & \pp_{1222}\tilde{p}_{4144} & \pp_{1122}\tilde{p}_{4114} &
\pp_{1122}\tilde{p}_{4141} & \pp_{1233}\tilde{p}_{4144} & 0 & 0 & 0 & 0 \\
\hline
(3,3) &\pp_{1133}\tilde{p}_{4411} & 0 & 0 & 0 & 0 & \tilde{p}_{3344} & \pp_{2233}\tilde{p}_{2233} & \pp_{1233}\tilde{p}_{4414}  &  \pp_{1233}\tilde{p}_{4441}  &\pp_{3333}\tilde{p}_{3333} & \pp_{1333}\tilde{p}_{4414} & \pp_{1333}\tilde{p}_{4441} & \pp_{2333}\tilde{p}_{2444} & \pp_{2333}\tilde{p}_{2444} \\
\hline
(1,3) & 0 &  0 & 0 & 0 & 0 & 0 & 0 & 0& 0& \pp_{1333}\tilde{p}_{1444} & \pp_{1133}\tilde{p}_{1414}& 
\pp_{1133}\tilde{p}_{1441} & \pp_{1233}\tilde{p}_{1444} & \pp_{1233}\tilde{p}_{1444}\\
(3,1)  & 0 &  0 & 0 & 0 & 0 & 0 & 0 & 0& 0& \pp_{1333}\tilde{p}_{4144} & \pp_{1133}\tilde{p}_{4114} & \pp_{1133}\tilde{p}_{4141} & 
\pp_{1233}\tilde{p}_{4144} & \pp_{1233}\tilde{p}_{4144} \\
(2,3)& 0 &  0 & 0 & 0 & 0 & 0 & 0 & 0& 0&
\pp_{2333} \tilde{p}_{2444} & 
\pp_{1233}\tilde{p}_{4414} & 
\pp_{1233}\tilde{p}_{4441} & 
\pp_{2233}\tilde{p}_{2323} & 
\pp_{2233}\tilde{p}_{2332} \\
(3,2) & 0 &  0 & 0 & 0 & 0 & 0 & 0 & 0& 0&
\pp_{2333}\tilde{p}_{2444} & 
\pp_{1233}\tilde{p}_{4414} & 
\pp_{1233}\tilde{p}_{4441} & 
\pp_{2233}\tilde{p}_{2332} & 
\pp_{2233}\tilde{p}_{2323} 
\end{NiceArray}$
\caption{Flattening matrix of $\Bar{p}$ for a quartet $T=12\mid 34$. The entries highlighted in dark gray form a (generically) non-vanishing 3-minor and the submatrix formed by the entries highlighted in (dark and light) gray has rank 3. Therefore, the 7 quartic edge invariants arise from considering all 4-minors of the grey submatrix containing the dark grey minor. }
\label{tab:flatrank3b}
\end{table}

\begin{table}[h]
    \tiny
$\begin{NiceArray}{c|c|cccc|c|c|cc|c|cccc|cc}
\CodeBefore
\rectanglecolor{gray!30}{2-2}{11-3} 
\rectanglecolor{gray!30}{2-7}{11-7} 
\rectanglecolor{gray!30}{2-9}{11-9} 
\rectanglecolor{gray!80}{2-2}{3-3} 
\rectanglecolor{gray!80}{9-9}{9-9}
\rectanglecolor{gray!80}{2-9}{3-9}
\rectanglecolor{gray!80}{9-2}{9-3}
\Body
& (1,1) & (1,4)  & (4,1) & (2,4) & (4,2) & (4,4) & (2,2) & (1,2) & (2,1)& (3,3) & (1,3) & (3,1) & (2,3) & (3,2)\\
\hline
(1,1) & \tilde{p}_{1111} & 0  & 0 & 0 & 0 &  \pp_{1144}\tilde{p}_{1144} & \pp_{1122}\tilde{p}_{1144} & 0&0& \pp_{1133}\tilde{p}_{1144} & 0 & 0 & 0 & 0 \\
\hline
(1,4) & 0 & \pp_{1144}\tilde{p}_{1414} & \pp_{1144}\tilde{p}_{1441} & \pp_{1244}\tilde{p}_{1444} & \pp_{1244}\tilde{p}_{1444} & \pp_{1444}\tilde{p}_{1444} & 0 & 0 & 0 & 0 & 0 & 0 & 0 & 0 \\
(4,1) & 0 & \pp_{1144}\tilde{p}_{4114} & \pp_{1144}\tilde{p}_{4141} &\pp_{1244}\tilde{p}_{4144} & \pp_{1244}\tilde{p}_{4144} &  \pp_{1444}\tilde{p}_{4144} &0 & 0 & 0 & 0 & 0 & 0 & 0 & 0\\
(2,4) & 0 & \pp_{1244}\tilde{p}_{4414} & \pp_{1244}\tilde{p}_{4441}  &\pp_{2244} \tilde{p}_{2424} & \pp_{2244} \tilde{p}_{2442} & \pp_{2444}\tilde{p}_{2444} & 0 & 0 & 0 & 0 & 0 & 0 & 0 & 0 \\
(4,2) & 0 & \pp_{1244}\tilde{p}_{4414} & \pp_{1244}\tilde{p}_{4441} & \pp_{2244}\tilde{p}_{2442} & \pp_{2244}\tilde{p}_{2424}  & \pp_{2444}\tilde{p}_{2444} & 0 & 0 & 0 & 0 & 0 & 0 & 0 & 0\\
\hline
(4,4) & \pp_{1144}\tilde{p}_{4411} & \pp_{1444}\tilde{p}_{4414} & \pp_{1444}\tilde{p}_{4441} & \pp_{2444}\tilde{p}_{2444} &\pp_{2444}\tilde{p}_{2444} & \pp_{4444}\tilde{p}_{4444}  & \pp_{2244}\tilde{p}_{2244} & \pp_{1244}\tilde{p}_{4414} & 
\pp_{1244}\tilde{p}_{4441} &\tilde{p}_{3344} & 0  & 0 & 0 & 0 \\
\hline
(2,2) & \pp_{1122}\tilde{p}_{4411} & 0 & 0 & 0 & 0& 
\pp_{2244}\tilde{p}_{2244} & \pp_{2222}\tilde{p}_{2222} &
\pp_{1222}\tilde{p}_{4414} & \pp_{1222}\tilde{p}_{4441} & \pp_{2233}\tilde{p}_{2233} & 0 & 0 & 0 & 0 \\
\hline
(1,2) & 0 & 0 & 0 & 0 & 0&  \pp_{1244}\tilde{p}_{1444} & \pp_{1222}\tilde{p}_{1444} & \pp_{1122}\tilde{p}_{1414} & \pp_{1122}\tilde{p}_{1441} & \pp_{1233}\tilde{p}_{1444} & 0 & 0 & 0 & 0\\
(2,1) & 0 & 0  & 0 & 0 & 0 & \pp_{1244}\tilde{p}_{4144} & \pp_{1222}\tilde{p}_{4144} & \pp_{1122}\tilde{p}_{4114} &
\pp_{1122}\tilde{p}_{4141} & \pp_{1233}\tilde{p}_{4144} & 0 & 0 & 0 & 0 \\
\hline
(3,3) &\pp_{1133}\tilde{p}_{4411} & 0 & 0 & 0 & 0 & \tilde{p}_{3344} & \pp_{2233}\tilde{p}_{2233} & \pp_{1233}\tilde{p}_{4414}  &  \pp_{1233}\tilde{p}_{4441}  &\pp_{3333}\tilde{p}_{3333} & \pp_{1333}\tilde{p}_{4414} & \pp_{1333}\tilde{p}_{4441} & \pp_{2333}\tilde{p}_{2444} & \pp_{2333}\tilde{p}_{2444} \\
\hline
(1,3) & 0 &  0 & 0 & 0 & 0 & 0 & 0 & 0& 0& \pp_{1333}\tilde{p}_{1444} & \pp_{1133}\tilde{p}_{1414}& 
\pp_{1133}\tilde{p}_{1441} & \pp_{1233}\tilde{p}_{1444} & \pp_{1233}\tilde{p}_{1444}\\
(3,1)  & 0 &  0 & 0 & 0 & 0 & 0 & 0 & 0& 0& \pp_{1333}\tilde{p}_{4144} & \pp_{1133}\tilde{p}_{4114} & \pp_{1133}\tilde{p}_{4141} & 
\pp_{1233}\tilde{p}_{4144} & \pp_{1233}\tilde{p}_{4144} \\
(2,3)& 0 &  0 & 0 & 0 & 0 & 0 & 0 & 0& 0&
\pp_{2333} \tilde{p}_{2444} & 
\pp_{1233}\tilde{p}_{4414} & 
\pp_{1233}\tilde{p}_{4441} & 
\pp_{2233}\tilde{p}_{2323} & 
\pp_{2233}\tilde{p}_{2332} \\
(3,2) & 0 &  0 & 0 & 0 & 0 & 0 & 0 & 0& 0&
\pp_{2333}\tilde{p}_{2444} & 
\pp_{1233}\tilde{p}_{4414} & 
\pp_{1233}\tilde{p}_{4441} & 
\pp_{2233}\tilde{p}_{2332} & 
\pp_{2233}\tilde{p}_{2323} 
\end{NiceArray}$
\caption{Flattening matrix of $\Bar{p}$ for a quartet $T=12\mid 34$. The entries highlighted in dark gray form a (generically) non-vanishing 3-minor and the submatrix formed by the entries highlighted in (dark and light) gray has rank 3. Therefore, the 7 quartic edge invariants arise from considering all 4-minors of the grey submatrix containing the dark grey minor. }
\label{tab:flatrank4}
\end{table}
\end{landscape}

\begin{landscape}
    
\section{Flattening matrix F84}\label{app:flatF84}
\begin{table}[h]
    \tiny
$\begin{NiceArray}{c|c|cccc|c|c|cc|c|cccc|cc}
\CodeBefore
\Body
& (1,1) & (1,4)  & (4,1) & (2,4) & (4,2) & (4,4) & (2,2) & (1,2) & (2,1)& (3,3) & (1,3) & (3,1) & (2,3) & (3,2)\\
\hline
(1,1) & \tilde{p}_{1111} & 0  & 0 & 0 & 0 &  \pp_{1144}\tilde{p}_{1144} & \pp_{1122}\tilde{p}_{1122} & 0&0& \pp_{1133}\tilde{p}_{1144} & 0 & 0 & 0 & 0 \\
\hline
(1,4) & 0 & \pp_{1144}\tilde{p}_{1414} & \pp_{1144}\tilde{p}_{1441} & \pp_{1244}\tilde{p}_{1424} & \pp_{1244}\tilde{p}_{1442} & \pp_{1444}\tilde{p}_{1444} & 0 & 0 & 0 & 0 & 0 & 0 & 0 & 0 \\
(4,1) & 0 & \pp_{1144}\tilde{p}_{4114} & \pp_{1144}\tilde{p}_{4141} &\pp_{1244}\tilde{p}_{4124} & \pp_{1244}\tilde{p}_{4142} &  \pp_{1444}\tilde{p}_{4144} &0 & 0 & 0 & 0 & 0 & 0 & 0 & 0\\
(2,4) & 0 & \pp_{1244}\tilde{p}_{2414} & \pp_{1244}\tilde{p}_{2441}  &\pp_{2244} \tilde{p}_{2424} & \pp_{2244} \tilde{p}_{2442} & \pp_{2444}\tilde{p}_{2444} & 0 & 0 & 0 & 0 & 0 & 0 & 0 & 0 \\
(4,2) & 0 & \pp_{1244}\tilde{p}_{4214} & \pp_{1244}\tilde{p}_{4241} & \pp_{2244}\tilde{p}_{4224} & \pp_{2244}\tilde{p}_{4242}  & \pp_{2444}\tilde{p}_{4244} & 0 & 0 & 0 & 0 & 0 & 0 & 0 & 0\\
\hline
(4,4) & \pp_{1144}\tilde{p}_{4411} & \pp_{1444}\tilde{p}_{4414} & \pp_{1444}\tilde{p}_{4441} & \pp_{2444}\tilde{p}_{4424} & \pp_{2444}\tilde{p}_{4442} & \pp_{4444}\tilde{p}_{4444}  & \pp_{2244}\tilde{p}_{4422} & \pp_{1244}\tilde{p}_{4412} & 
\pp_{1244}\tilde{p}_{4421} &\tilde{p}_{3344} & 0  & 0 & 0 & 0 \\
\hline
(2,2) & \pp_{1122}\tilde{p}_{2211} & 0 & 0 & 0 & 0& 
\pp_{2244}\tilde{p}_{2244} & \pp_{2222}\tilde{p}_{2222} &
\pp_{1222}\tilde{p}_{2212} & \pp_{1222}\tilde{p}_{2221} & \pp_{2233}\tilde{p}_{2233} & 0 & 0 & 0 & 0 \\
\hline
(1,2) & 0 & 0 & 0 & 0 & 0&  \pp_{1244}\tilde{p}_{1244} & \pp_{1222}\tilde{p}_{1222} & \pp_{1122}\tilde{p}_{1212} & \pp_{1122}\tilde{p}_{1221} & \pp_{1233}\tilde{p}_{1244} & 0 & 0 & 0 & 0\\
(2,1) & 0 & 0  & 0 & 0 & 0 & \pp_{1244}\tilde{p}_{2144} & \pp_{1222}\tilde{p}_{2122} & \pp_{1122}\tilde{p}_{2112} &
\pp_{1122}\tilde{p}_{2121} & \pp_{1233}\tilde{p}_{2144} & 0 & 0 & 0 & 0 \\
\hline
(3,3) &\pp_{1133}\tilde{p}_{4411} & 0 & 0 & 0 & 0 & \tilde{p}_{3344} & \pp_{2233}\tilde{p}_{3322} & \pp_{1233}\tilde{p}_{4412}  &  \pp_{1233}\tilde{p}_{4421}  &\pp_{3333}\tilde{p}_{3333} & \pp_{1333}\tilde{p}_{4414} & \pp_{1333}\tilde{p}_{4441} & \pp_{2333}\tilde{p}_{4424} & \pp_{2333}\tilde{p}_{4442} \\
\hline
(1,3) & 0 &  0 & 0 & 0 & 0 & 0 & 0 & 0& 0& \pp_{1333}\tilde{p}_{1444} & \pp_{1133}\tilde{p}_{1414}& 
\pp_{1133}\tilde{p}_{1441} & \pp_{1233}\tilde{p}_{1424} & \pp_{1233}\tilde{p}_{1442}\\
(3,1)  & 0 &  0 & 0 & 0 & 0 & 0 & 0 & 0& 0& \pp_{1333}\tilde{p}_{4144} & \pp_{1133}\tilde{p}_{4114} & \pp_{1133}\tilde{p}_{4141} & 
\pp_{1233}\tilde{p}_{4124} & \pp_{1233}\tilde{p}_{4142} \\
(2,3)& 0 &  0 & 0 & 0 & 0 & 0 & 0 & 0& 0&
\pp_{2333} \tilde{p}_{2444} & 
\pp_{1233}\tilde{p}_{2414} & 
\pp_{1233}\tilde{p}_{2441} & 
\pp_{2233}\tilde{p}_{2323} & 
\pp_{2233}\tilde{p}_{2332} \\
(3,2) & 0 &  0 & 0 & 0 & 0 & 0 & 0 & 0& 0&
\pp_{2333}\tilde{p}_{4244} & 
\pp_{1233}\tilde{p}_{4214} & 
\pp_{1233}\tilde{p}_{4241} & 
\pp_{2233}\tilde{p}_{3223} & 
\pp_{2233}\tilde{p}_{3232} 
\end{NiceArray}$
\caption{Flattening matrix of $\Bar{p}$ for a quartet $T=12\mid 34$ under the F84 model.}
\label{tab:flatF84}
\end{table}

\begin{table}[h]
    \tiny
$\begin{NiceArray}{c|c|cccc|c|c|cc|c|cccc|cc}
\CodeBefore
\rectanglecolor{gray!30}{3-3}{7-6} 
\rectanglecolor{gray!30}{11-12}{15-15} 
\rectanglecolor{gray!30}{7-9}{11-10} 
\rectanglecolor{gray!80}{3-3}{3-3} 
\rectanglecolor{gray!80}{12-12}{12-12}
\rectanglecolor{gray!80}{9-9}{9-9}
\Body
& (1,1) & (1,4)  & (4,1) & (2,4) & (4,2) & (4,4) & (2,2) & (1,2) & (2,1)& (3,3) & (1,3) & (3,1) & (2,3) & (3,2)\\
\hline
(1,1) & \tilde{p}_{1111} & 0  & 0 & 0 & 0 &  \pp_{1144}\tilde{p}_{1144} & \pp_{1122}\tilde{p}_{1122} & 0&0& \pp_{1133}\tilde{p}_{1144} & 0 & 0 & 0 & 0 \\
\hline
(1,4) & 0 & \pp_{1144}\tilde{p}_{1414} & \pp_{1144}\tilde{p}_{1441} & \pp_{1244}\tilde{p}_{1424} & \pp_{1244}\tilde{p}_{1442} & \pp_{1444}\tilde{p}_{1444} & 0 & 0 & 0 & 0 & 0 & 0 & 0 & 0 \\
(4,1) & 0 & \pp_{1144}\tilde{p}_{4114} & \pp_{1144}\tilde{p}_{4141} &\pp_{1244}\tilde{p}_{4124} & \pp_{1244}\tilde{p}_{4142} &  \pp_{1444}\tilde{p}_{4144} &0 & 0 & 0 & 0 & 0 & 0 & 0 & 0\\
(2,4) & 0 & \pp_{1244}\tilde{p}_{2414} & \pp_{1244}\tilde{p}_{2441}  &\pp_{2244} \tilde{p}_{2424} & \pp_{2244} \tilde{p}_{2442} & \pp_{2444}\tilde{p}_{2444} & 0 & 0 & 0 & 0 & 0 & 0 & 0 & 0 \\
(4,2) & 0 & \pp_{1244}\tilde{p}_{4214} & \pp_{1244}\tilde{p}_{4241} & \pp_{2244}\tilde{p}_{4224} & \pp_{2244}\tilde{p}_{4242}  & \pp_{2444}\tilde{p}_{4244} & 0 & 0 & 0 & 0 & 0 & 0 & 0 & 0\\
\hline
(4,4) & \pp_{1144}\tilde{p}_{4411} & \pp_{1444}\tilde{p}_{4414} & \pp_{1444}\tilde{p}_{4441} & \pp_{2444}\tilde{p}_{4424} &\pp_{2444}\tilde{p}_{4442} & \pp_{4444}\tilde{p}_{4444}  & \pp_{2244}\tilde{p}_{4422} & \pp_{1244}\tilde{p}_{4412} & 
\pp_{1244}\tilde{p}_{4421} &\tilde{p}_{3344} & 0  & 0 & 0 & 0 \\
\hline
(2,2) & \pp_{1122}\tilde{p}_{2211} & 0 & 0 & 0 & 0& 
\pp_{2244}\tilde{p}_{2244} & \pp_{2222}\tilde{p}_{2222} &
\pp_{1222}\tilde{p}_{2212} & \pp_{1222}\tilde{p}_{2221} & \pp_{2233}\tilde{p}_{2233} & 0 & 0 & 0 & 0 \\
\hline
(1,2) & 0 & 0 & 0 & 0 & 0&  \pp_{1244}\tilde{p}_{1244} & \pp_{1222}\tilde{p}_{1222} & \pp_{1122}\tilde{p}_{1212} & \pp_{1122}\tilde{p}_{1221} & \pp_{1233}\tilde{p}_{1244} & 0 & 0 & 0 & 0\\
(2,1) & 0 & 0  & 0 & 0 & 0 & \pp_{1244}\tilde{p}_{2144} & \pp_{1222}\tilde{p}_{2122} & \pp_{1122}\tilde{p}_{2112} &
\pp_{1122}\tilde{p}_{2121} & \pp_{1233}\tilde{p}_{2144} & 0 & 0 & 0 & 0 \\
\hline
(3,3) &\pp_{1133}\tilde{p}_{4411} & 0 & 0 & 0 & 0 & \tilde{p}_{3344} & \pp_{2233}\tilde{p}_{3322} & \pp_{1233}\tilde{p}_{4412}  &  \pp_{1233}\tilde{p}_{4421}  &\pp_{3333}\tilde{p}_{3333} & \pp_{1333}\tilde{p}_{4414} & \pp_{1333}\tilde{p}_{4441} & \pp_{2333}\tilde{p}_{4424} & \pp_{2333}\tilde{p}_{4442} \\
\hline
(1,3) & 0 &  0 & 0 & 0 & 0 & 0 & 0 & 0& 0& \pp_{1333}\tilde{p}_{1444} & \pp_{1133}\tilde{p}_{1414}& 
\pp_{1133}\tilde{p}_{1441} & \pp_{1233}\tilde{p}_{1424} & \pp_{1233}\tilde{p}_{1442}\\
(3,1)  & 0 &  0 & 0 & 0 & 0 & 0 & 0 & 0& 0& \pp_{1333}\tilde{p}_{4144} & \pp_{1133}\tilde{p}_{4114} & \pp_{1133}\tilde{p}_{4141} & 
\pp_{1233}\tilde{p}_{4124} & \pp_{1233}\tilde{p}_{4142} \\
(2,3)& 0 &  0 & 0 & 0 & 0 & 0 & 0 & 0& 0&
\pp_{2333} \tilde{p}_{2444} & 
\pp_{1233}\tilde{p}_{2414} & 
\pp_{1233}\tilde{p}_{2441} & 
\pp_{2233}\tilde{p}_{2323} & 
\pp_{2233}\tilde{p}_{2332} \\
(3,2) & 0 &  0 & 0 & 0 & 0 & 0 & 0 & 0& 0&
\pp_{2333}\tilde{p}_{4244} & 
\pp_{1233}\tilde{p}_{4214} & 
\pp_{1233}\tilde{p}_{4241} & 
\pp_{2233}\tilde{p}_{3223} & 
\pp_{2233}\tilde{p}_{3232} 
\end{NiceArray}$
\caption{Flattening matrix of $\bar{p}$ for a quartet $T=12\vert 34$. The highlighted submatrices gave rank 1 and the entries highlighted in dark grey are generically non-zero. Therefore, the quadratic edge invariants arise from considering all $2$-minors of the grey submatrices containing the dark grey minor.}
\label{tab:flatF84rk1}
\end{table}

\begin{table}[h]
    \tiny
$\begin{NiceArray}{c|c|cccc|c|c|cc|c|cccc|cc}
\CodeBefore
\rectanglecolor{gray!30}{7-2}{11-2} 
\rectanglecolor{gray!30}{2-8}{2-8} 
\rectanglecolor{gray!30}{7-8}{11-9} 
\rectanglecolor{gray!80}{2-2}{2-2} 
\rectanglecolor{gray!80}{9-2}{9-2}
\rectanglecolor{gray!80}{2-9}{2-9}
\rectanglecolor{gray!80}{9-9}{9-9}
\Body
& (1,1) & (1,4)  & (4,1) & (2,4) & (4,2) & (4,4) & (2,2) & (1,2) & (2,1)& (3,3) & (1,3) & (3,1) & (2,3) & (3,2)\\
\hline
(1,1) & \tilde{p}_{1111} & 0  & 0 & 0 & 0 &  \pp_{1144}\tilde{p}_{1144} & \pp_{1122}\tilde{p}_{1122} & 0&0& \pp_{1133}\tilde{p}_{1144} & 0 & 0 & 0 & 0 \\
\hline
(1,4) & 0 & \pp_{1144}\tilde{p}_{1414} & \pp_{1144}\tilde{p}_{1441} & \pp_{1244}\tilde{p}_{1424} & \pp_{1244}\tilde{p}_{1442} & \pp_{1444}\tilde{p}_{1444} & 0 & 0 & 0 & 0 & 0 & 0 & 0 & 0 \\
(4,1) & 0 & \pp_{1144}\tilde{p}_{4114} & \pp_{1144}\tilde{p}_{4141} &\pp_{1244}\tilde{p}_{4124} & \pp_{1244}\tilde{p}_{4142} &  \pp_{1444}\tilde{p}_{4144} &0 & 0 & 0 & 0 & 0 & 0 & 0 & 0\\
(2,4) & 0 & \pp_{1244}\tilde{p}_{2414} & \pp_{1244}\tilde{p}_{2441}  &\pp_{2244} \tilde{p}_{2424} & \pp_{2244} \tilde{p}_{2442} & \pp_{2444}\tilde{p}_{2444} & 0 & 0 & 0 & 0 & 0 & 0 & 0 & 0 \\
(4,2) & 0 & \pp_{1244}\tilde{p}_{4214} & \pp_{1244}\tilde{p}_{4241} & \pp_{2244}\tilde{p}_{4224} & \pp_{2244}\tilde{p}_{4242}  & \pp_{2444}\tilde{p}_{4244} & 0 & 0 & 0 & 0 & 0 & 0 & 0 & 0\\
\hline
(4,4) & \pp_{1144}\tilde{p}_{4411} & \pp_{1444}\tilde{p}_{4414} & \pp_{1444}\tilde{p}_{4441} & \pp_{2444}\tilde{p}_{4424} &\pp_{2444}\tilde{p}_{4442} & \pp_{4444}\tilde{p}_{4444}  & \pp_{2244}\tilde{p}_{4422} & \pp_{1244}\tilde{p}_{4412} & 
\pp_{1244}\tilde{p}_{4421} &\tilde{p}_{3344} & 0  & 0 & 0 & 0 \\
\hline
(2,2) & \pp_{1122}\tilde{p}_{2211} & 0 & 0 & 0 & 0& 
\pp_{2244}\tilde{p}_{2244} & \pp_{2222}\tilde{p}_{2222} &
\pp_{1222}\tilde{p}_{2212} & \pp_{1222}\tilde{p}_{2221} & \pp_{2233}\tilde{p}_{2233} & 0 & 0 & 0 & 0 \\
\hline
(1,2) & 0 & 0 & 0 & 0 & 0&  \pp_{1244}\tilde{p}_{1244} & \pp_{1222}\tilde{p}_{1222} & \pp_{1122}\tilde{p}_{1212} & \pp_{1122}\tilde{p}_{1221} & \pp_{1233}\tilde{p}_{1244} & 0 & 0 & 0 & 0\\
(2,1) & 0 & 0  & 0 & 0 & 0 & \pp_{1244}\tilde{p}_{2144} & \pp_{1222}\tilde{p}_{2122} & \pp_{1122}\tilde{p}_{2112} &
\pp_{1122}\tilde{p}_{2121} & \pp_{1233}\tilde{p}_{2144} & 0 & 0 & 0 & 0 \\
\hline
(3,3) &\pp_{1133}\tilde{p}_{4411} & 0 & 0 & 0 & 0 & \tilde{p}_{3344} & \pp_{2233}\tilde{p}_{3322} & \pp_{1233}\tilde{p}_{4412}  &  \pp_{1233}\tilde{p}_{4421}  &\pp_{3333}\tilde{p}_{3333} & \pp_{1333}\tilde{p}_{4414} & \pp_{1333}\tilde{p}_{4441} & \pp_{2333}\tilde{p}_{4424} & \pp_{2333}\tilde{p}_{4442} \\
\hline
(1,3) & 0 &  0 & 0 & 0 & 0 & 0 & 0 & 0& 0& \pp_{1333}\tilde{p}_{1444} & \pp_{1133}\tilde{p}_{1414}& 
\pp_{1133}\tilde{p}_{1441} & \pp_{1233}\tilde{p}_{1424} & \pp_{1233}\tilde{p}_{1442}\\
(3,1)  & 0 &  0 & 0 & 0 & 0 & 0 & 0 & 0& 0& \pp_{1333}\tilde{p}_{4144} & \pp_{1133}\tilde{p}_{4114} & \pp_{1133}\tilde{p}_{4141} & 
\pp_{1233}\tilde{p}_{4124} & \pp_{1233}\tilde{p}_{4142} \\
(2,3)& 0 &  0 & 0 & 0 & 0 & 0 & 0 & 0& 0&
\pp_{2333} \tilde{p}_{2444} & 
\pp_{1233}\tilde{p}_{2414} & 
\pp_{1233}\tilde{p}_{2441} & 
\pp_{2233}\tilde{p}_{2323} & 
\pp_{2233}\tilde{p}_{2332} \\
(3,2) & 0 &  0 & 0 & 0 & 0 & 0 & 0 & 0& 0&
\pp_{2333}\tilde{p}_{4244} & 
\pp_{1233}\tilde{p}_{4214} & 
\pp_{1233}\tilde{p}_{4241} & 
\pp_{2233}\tilde{p}_{3223} & 
\pp_{2233}\tilde{p}_{3232} 
\end{NiceArray}$
\caption{Flattening matrix of $\bar{p}$ for a quartet $T=12\vert 34$. The entries highlighted in dark grey form a generically non-vanishing $2$-minor and the submatrix formed by the entries highlighted in grey has rank 2. Therefore, the 4 cubic edge invariants arise from considering all 3-minors of the grey submatrix containing the dark grey minor.}
\label{tab:flatF84rk2a}
\end{table}

\begin{table}[h]
    \tiny
$\begin{NiceArray}{c|c|cccc|c|c|cc|c|cccc|cc}
\CodeBefore
\rectanglecolor{gray!30}{7-2}{15-2} 
\rectanglecolor{gray!30}{2-11}{2-11} 
\rectanglecolor{gray!30}{7-9}{15-9} 
\rectanglecolor{gray!30}{7-11}{15-12} 
\rectanglecolor{gray!80}{2-2}{2-2} 
\rectanglecolor{gray!80}{9-2}{9-2}
\rectanglecolor{gray!80}{2-9}{2-9}
\rectanglecolor{gray!80}{9-9}{9-9}
\rectanglecolor{gray!80}{12-2}{12-2}
\rectanglecolor{gray!80}{2-12}{2-12}
\rectanglecolor{gray!80}{12-9}{12-9}
\rectanglecolor{gray!80}{9-12}{9-12}
\rectanglecolor{gray!80}{12-12}{12-12}
\Body
& (1,1) & (1,4)  & (4,1) & (2,4) & (4,2) & (4,4) & (2,2) & (1,2) & (2,1)& (3,3) & (1,3) & (3,1) & (2,3) & (3,2)\\
\hline
(1,1) & \tilde{p}_{1111} & 0  & 0 & 0 & 0 &  \pp_{1144}\tilde{p}_{1144} & \pp_{1122}\tilde{p}_{1122} & 0&0& \pp_{1133}\tilde{p}_{1144} & 0 & 0 & 0 & 0 \\
\hline
(1,4) & 0 & \pp_{1144}\tilde{p}_{1414} & \pp_{1144}\tilde{p}_{1441} & \pp_{1244}\tilde{p}_{1424} & \pp_{1244}\tilde{p}_{1442} & \pp_{1444}\tilde{p}_{1444} & 0 & 0 & 0 & 0 & 0 & 0 & 0 & 0 \\
(4,1) & 0 & \pp_{1144}\tilde{p}_{4114} & \pp_{1144}\tilde{p}_{4141} &\pp_{1244}\tilde{p}_{4124} & \pp_{1244}\tilde{p}_{4142} &  \pp_{1444}\tilde{p}_{4144} &0 & 0 & 0 & 0 & 0 & 0 & 0 & 0\\
(2,4) & 0 & \pp_{1244}\tilde{p}_{2414} & \pp_{1244}\tilde{p}_{2441}  &\pp_{2244} \tilde{p}_{2424} & \pp_{2244} \tilde{p}_{2442} & \pp_{2444}\tilde{p}_{2444} & 0 & 0 & 0 & 0 & 0 & 0 & 0 & 0 \\
(4,2) & 0 & \pp_{1244}\tilde{p}_{4214} & \pp_{1244}\tilde{p}_{4241} & \pp_{2244}\tilde{p}_{4224} & \pp_{2244}\tilde{p}_{4242}  & \pp_{2444}\tilde{p}_{4244} & 0 & 0 & 0 & 0 & 0 & 0 & 0 & 0\\
\hline
(4,4) & \pp_{1144}\tilde{p}_{4411} & \pp_{1444}\tilde{p}_{4414} & \pp_{1444}\tilde{p}_{4441} & \pp_{2444}\tilde{p}_{4424} &\pp_{2444}\tilde{p}_{4442} & \pp_{4444}\tilde{p}_{4444}  & \pp_{2244}\tilde{p}_{4422} & \pp_{1244}\tilde{p}_{4412} & 
\pp_{1244}\tilde{p}_{4421} &\tilde{p}_{3344} & 0  & 0 & 0 & 0 \\
\hline
(2,2) & \pp_{1122}\tilde{p}_{2211} & 0 & 0 & 0 & 0& 
\pp_{2244}\tilde{p}_{2244} & \pp_{2222}\tilde{p}_{2222} &
\pp_{1222}\tilde{p}_{2212} & \pp_{1222}\tilde{p}_{2221} & \pp_{2233}\tilde{p}_{2233} & 0 & 0 & 0 & 0 \\
\hline
(1,2) & 0 & 0 & 0 & 0 & 0&  \pp_{1244}\tilde{p}_{1244} & \pp_{1222}\tilde{p}_{1222} & \pp_{1122}\tilde{p}_{1212} & \pp_{1122}\tilde{p}_{1221} & \pp_{1233}\tilde{p}_{1244} & 0 & 0 & 0 & 0\\
(2,1) & 0 & 0  & 0 & 0 & 0 & \pp_{1244}\tilde{p}_{2144} & \pp_{1222}\tilde{p}_{2122} & \pp_{1122}\tilde{p}_{2112} &
\pp_{1122}\tilde{p}_{2121} & \pp_{1233}\tilde{p}_{2144} & 0 & 0 & 0 & 0 \\
\hline
(3,3) &\pp_{1133}\tilde{p}_{4411} & 0 & 0 & 0 & 0 & \tilde{p}_{3344} & \pp_{2233}\tilde{p}_{3322} & \pp_{1233}\tilde{p}_{4412}  &  \pp_{1233}\tilde{p}_{4421}  &\pp_{3333}\tilde{p}_{3333} & \pp_{1333}\tilde{p}_{4414} & \pp_{1333}\tilde{p}_{4441} & \pp_{2333}\tilde{p}_{4424} & \pp_{2333}\tilde{p}_{4442} \\
\hline
(1,3) & 0 &  0 & 0 & 0 & 0 & 0 & 0 & 0& 0& \pp_{1333}\tilde{p}_{1444} & \pp_{1133}\tilde{p}_{1414}& 
\pp_{1133}\tilde{p}_{1441} & \pp_{1233}\tilde{p}_{1424} & \pp_{1233}\tilde{p}_{1442}\\
(3,1)  & 0 &  0 & 0 & 0 & 0 & 0 & 0 & 0& 0& \pp_{1333}\tilde{p}_{4144} & \pp_{1133}\tilde{p}_{4114} & \pp_{1133}\tilde{p}_{4141} & 
\pp_{1233}\tilde{p}_{4124} & \pp_{1233}\tilde{p}_{4142} \\
(2,3)& 0 &  0 & 0 & 0 & 0 & 0 & 0 & 0& 0&
\pp_{2333} \tilde{p}_{2444} & 
\pp_{1233}\tilde{p}_{2414} & 
\pp_{1233}\tilde{p}_{2441} & 
\pp_{2233}\tilde{p}_{2323} & 
\pp_{2233}\tilde{p}_{2332} \\
(3,2) & 0 &  0 & 0 & 0 & 0 & 0 & 0 & 0& 0&
\pp_{2333}\tilde{p}_{4244} & 
\pp_{1233}\tilde{p}_{4214} & 
\pp_{1233}\tilde{p}_{4241} & 
\pp_{2233}\tilde{p}_{3223} & 
\pp_{2233}\tilde{p}_{3232} 
\end{NiceArray}$
\caption{Flattening matrix of $\bar{p}$ for a quartet $T=12\vert 34$. The entries highlighted in dark grey form a generically non-vanishing $2$-minor and the submatrix formed by the entries highlighted in grey has rank 2. Therefore, the 7 cubic edge invariants arise from considering all 3-minors of the grey submatrix containing the dark grey minor.}
\label{tab:flatF84rk2b}
\end{table}

\begin{table}[h]
    \tiny
$\begin{NiceArray}{c|c|cccc|c|c|cc|c|cccc|cc}
\CodeBefore
\rectanglecolor{gray!30}{2-2}{11-3} 
\rectanglecolor{gray!30}{2-7}{11-7} 
\rectanglecolor{gray!30}{2-9}{11-9} 
\rectanglecolor{gray!80}{2-2}{3-3} 
\rectanglecolor{gray!80}{9-9}{9-9}
\rectanglecolor{gray!80}{2-9}{3-9}
\rectanglecolor{gray!80}{9-2}{9-3}
\Body
& (1,1) & (1,4)  & (4,1) & (2,4) & (4,2) & (4,4) & (2,2) & (1,2) & (2,1)& (3,3) & (1,3) & (3,1) & (2,3) & (3,2)\\
\hline
(1,1) & \tilde{p}_{1111} & 0  & 0 & 0 & 0 &  \pp_{1144}\tilde{p}_{1144} & \pp_{1122}\tilde{p}_{1122} & 0&0& \pp_{1133}\tilde{p}_{1144} & 0 & 0 & 0 & 0 \\
\hline
(1,4) & 0 & \pp_{1144}\tilde{p}_{1414} & \pp_{1144}\tilde{p}_{1441} & \pp_{1244}\tilde{p}_{1424} & \pp_{1244}\tilde{p}_{1442} & \pp_{1444}\tilde{p}_{1444} & 0 & 0 & 0 & 0 & 0 & 0 & 0 & 0 \\
(4,1) & 0 & \pp_{1144}\tilde{p}_{4114} & \pp_{1144}\tilde{p}_{4141} &\pp_{1244}\tilde{p}_{4124} & \pp_{1244}\tilde{p}_{4142} &  \pp_{1444}\tilde{p}_{4144} &0 & 0 & 0 & 0 & 0 & 0 & 0 & 0\\
(2,4) & 0 & \pp_{1244}\tilde{p}_{2414} & \pp_{1244}\tilde{p}_{2441}  &\pp_{2244} \tilde{p}_{2424} & \pp_{2244} \tilde{p}_{2442} & \pp_{2444}\tilde{p}_{2444} & 0 & 0 & 0 & 0 & 0 & 0 & 0 & 0 \\
(4,2) & 0 & \pp_{1244}\tilde{p}_{4214} & \pp_{1244}\tilde{p}_{4241} & \pp_{2244}\tilde{p}_{4224} & \pp_{2244}\tilde{p}_{4242}  & \pp_{2444}\tilde{p}_{4244} & 0 & 0 & 0 & 0 & 0 & 0 & 0 & 0\\
\hline
(4,4) & \pp_{1144}\tilde{p}_{4411} & \pp_{1444}\tilde{p}_{4414} & \pp_{1444}\tilde{p}_{4441} & \pp_{2444}\tilde{p}_{4424} &\pp_{2444}\tilde{p}_{4442} & \pp_{4444}\tilde{p}_{4444}  & \pp_{2244}\tilde{p}_{4422} & \pp_{1244}\tilde{p}_{4412} & 
\pp_{1244}\tilde{p}_{4421} &\tilde{p}_{3344} & 0  & 0 & 0 & 0 \\
\hline
(2,2) & \pp_{1122}\tilde{p}_{2211} & 0 & 0 & 0 & 0& 
\pp_{2244}\tilde{p}_{2244} & \pp_{2222}\tilde{p}_{2222} &
\pp_{1222}\tilde{p}_{2212} & \pp_{1222}\tilde{p}_{2221} & \pp_{2233}\tilde{p}_{2233} & 0 & 0 & 0 & 0 \\
\hline
(1,2) & 0 & 0 & 0 & 0 & 0&  \pp_{1244}\tilde{p}_{1244} & \pp_{1222}\tilde{p}_{1222} & \pp_{1122}\tilde{p}_{1212} & \pp_{1122}\tilde{p}_{1221} & \pp_{1233}\tilde{p}_{1244} & 0 & 0 & 0 & 0\\
(2,1) & 0 & 0  & 0 & 0 & 0 & \pp_{1244}\tilde{p}_{2144} & \pp_{1222}\tilde{p}_{2122} & \pp_{1122}\tilde{p}_{2112} &
\pp_{1122}\tilde{p}_{2121} & \pp_{1233}\tilde{p}_{2144} & 0 & 0 & 0 & 0 \\
\hline
(3,3) &\pp_{1133}\tilde{p}_{4411} & 0 & 0 & 0 & 0 & \tilde{p}_{3344} & \pp_{2233}\tilde{p}_{3322} & \pp_{1233}\tilde{p}_{4412}  &  \pp_{1233}\tilde{p}_{4421}  &\pp_{3333}\tilde{p}_{3333} & \pp_{1333}\tilde{p}_{4414} & \pp_{1333}\tilde{p}_{4441} & \pp_{2333}\tilde{p}_{4424} & \pp_{2333}\tilde{p}_{4442} \\
\hline
(1,3) & 0 &  0 & 0 & 0 & 0 & 0 & 0 & 0& 0& \pp_{1333}\tilde{p}_{1444} & \pp_{1133}\tilde{p}_{1414}& 
\pp_{1133}\tilde{p}_{1441} & \pp_{1233}\tilde{p}_{1424} & \pp_{1233}\tilde{p}_{1442}\\
(3,1)  & 0 &  0 & 0 & 0 & 0 & 0 & 0 & 0& 0& \pp_{1333}\tilde{p}_{4144} & \pp_{1133}\tilde{p}_{4114} & \pp_{1133}\tilde{p}_{4141} & 
\pp_{1233}\tilde{p}_{4124} & \pp_{1233}\tilde{p}_{4142} \\
(2,3)& 0 &  0 & 0 & 0 & 0 & 0 & 0 & 0& 0&
\pp_{2333} \tilde{p}_{2444} & 
\pp_{1233}\tilde{p}_{2414} & 
\pp_{1233}\tilde{p}_{2441} & 
\pp_{2233}\tilde{p}_{2323} & 
\pp_{2233}\tilde{p}_{2332} \\
(3,2) & 0 &  0 & 0 & 0 & 0 & 0 & 0 & 0& 0&
\pp_{2333}\tilde{p}_{4244} & 
\pp_{1233}\tilde{p}_{4214} & 
\pp_{1233}\tilde{p}_{4241} & 
\pp_{2233}\tilde{p}_{3223} & 
\pp_{2233}\tilde{p}_{3232} 
\end{NiceArray}$
\caption{Flattening matrix of $\bar{p}$ for a quartet $T=12\vert 34$. The entries highlighted in dark grey form a generically non-vanishing $3$-minor and the submatrix formed by the entries highlighted in grey has rank 3. Therefore, the 7 quartic edge invariants arise from considering all 4-minors of the grey submatrix containing the dark grey minor.}
\label{tab:flatF84rk3}
\end{table}

\end{landscape}

\section{SageMath Code for F84}\label{app:code84}
\begin{lstlisting}
#define the p-tilde variables
p1111, p1114, p1141, p1411, p4111, p1112, p1121, p1121, p1211,
p2111, p4141, p4114, p4411, p1414, p1144, p2121, p2112, p2211, 
p1212, p1122, p4441, p4414, p4144, p1444, p2221, p2212, p2122, 
p1222 = var('p1111, p1114, p1141, p1411, p4111, p1112, p1121, 
p1121, p1211, p2111, p4141, p4114, p4411, p1414, p1144, p2121,
p2112, p2211, p1212, p1122, p4441, p4414, p4144, p1444, p2221, 
p2212, p2122, p1222')
p4412, p4214, p1424, p1422, p2414, p2441, p4241, p4214, p2424, 
p3322, p2444, p4244, p3344, p3322, p2333, p3333 = var('p4412, 
p4214, p1424, p1422, p2414, p2441, p4241, p4214, p2424, p3322,
p2444, p4244, p3344, p3322, p2333, p3333')
p1441, p1221, p1244, p1442 ,p4124, p4224, p4444, p4142, p2442, 
p4242, p4422, p2222, p2233, p2144, p2244, p4421= var('p1441, 
p1221, p1244, p1442, p4124, p4224, p4444, p4142, p2442, p4242, 
p4422, p2222, p2233, p2144, p2244, p4421')

#write the 38 equations 
f1 = p1222*p1441 - p1221*p1422
f2 = p1222*p1414 - p1212*p1424
f3 = p1222*p1144 - p1122*p1244
f4 = p1144*p1414*p1441 - p1111*p1444^2
f5 = p1442*p1424*p1244 - p1222*p1444^2
f6 = p2212*p4411 - p2211*p4412
f7 = p2212*p4114 - p2112*p4214
f8 = p2212*p1414 - p1212*p2414
f9 = p1414*p4114*p4411 - p1111*p4414^2
f10 = p4412*p4214*p2414 - p2212*p4414^2
g1 = p1414*p4141 - p1441*p4114
g2 = p1414*p2441 - p1441*p2414
g3 = p1414*p4241 - p1441*p4214
g4 = p1414*p4441 - p1441*p4414
g5 = p1414*p4124 - p1424*p4114
g6 = p1414*p2424 - p1424*p2414
g7 = p1414*p4224 - p1424*p4214
g8 = p1414*p4444 - p1424*p4414
g9 = p1414*p4142 - p1442*p4114
g10 = p1414*p2442 - p1442*p2414
g11 = p1414*p4242 - p1442*p4214
g12 = p1414*p2444 - p1442*p4414
g13 = p1212*p4421 - p1221*p4412
g14 = p1212*p2221 - p1221*p2212
g15 = p1212*p2121 - p1221*p2112
h1 = matrix([[p1111, p1122, 0], [0, p1222, p1212], 
[p4411, p4422, p4412]]).determinant()
h2 = matrix([[p1111, p1122, 0], [0, p1222, p1212], 
[p2211, p2222, p2212]]).determinant()
h3 = matrix([[p1111, p1122, 0], [0, p1222, p1212], 
[0, p2122, p2112]]).determinant()
h4 = matrix([[p1111, p1122, 0], [0, p1222, p1212], 
[p4411, p3322, p4412]]).determinant()
i1 = matrix([[p1111, 0, p1144, 0], [0, p1212, p1244, 0], 
[0, 0, p1444, p1414], [p4411, p4412, p3344, 0]]).determinant()
i2 = matrix([[p1111, 0, p1144, 0], [0, p1212, p1244, 0], 
[0, 0, p1444, p1414], [p2211, p2212, p2233, 0]]).determinant()
i3 = matrix([[p1111, 0, p1144, 0], [0, p1212, p1244, 0], 
[0, 0, p1444, p1414], [p4411, p4412, p3333, p4414]]).determinant()
i4 = matrix([[p1111, 0, p1144, 0], [0, p1212, p1244, 0], 
[0, 0, p1444, p1414], [0, 0, p4144, p4114]]).determinant()
i5 = matrix([[p1111, 0, p1144, 0], [0, p1212, p1244, 0], 
[0, 0, p1444, p1414], [0, 0, p2444, p2414]]).determinant()
i6 = matrix([[p1111, 0, p1144, 0], [0, p1212, p1244, 0], 
[0, 0, p1444, p1414], [0, 0, p4244, p4214]]).determinant()
i7 = matrix([[p1111, 0, p1144, 0], [0, p1212, p1244, 0], 
[0, 0, p1444, p1414], [0, p2112, p2144, 0]]).determinant()
i8 =  matrix([[p1111, 0, p1144, 0], [0, p1414, p1444, 0], 
[0, 0, p1244, p1212], [p4411, p4414, p4444, p4412]]).determinant()
i9 =  matrix([[p1111, 0, p1144, 0], [0, p1414, p1444, 0], 
[0, 0, p1244, p1212], [p2211, 0, p2244, p2212]]).determinant()

#compute the Jacobian for a generic point
J(p1111, p1114, p1141, p1411, p4111, p1112, p1121, p1121, p1211,
p2111, p4141, p4114, p4411, p1414, p1144, p2121, p2112, p2211, 
p1212, p1122, p4441, p4414, p4144, p1444, p2221, p2212, p2122, 
p1222, p4412, p4214, p1424, p1422, p2414, p2441, p4241, p4214, 
p2424, p3322, p2444, p4244, p3344, p3322, p2333, p3333, p1441, 
p1221, p1244, p1442 ,p4124, p4224, p4444, p4142, p2442, p4242, 
p4422, p2222, p2233, p2144, p2244, p4421)  = 
jacobian((f1,f2,f3,f4,f5,f6,f7,f8,f9,f10,g1,g2,g3,g4,g5,g6,g7,
g8,g9,g10,g11,g12,g13,g14,g15,h1,h2,h3,h4,i1,i2,i3,i4,i5,i6,i7,
i8,i9), (p1111, p1114, p1141, p1411, p4111, p1112, p1121, p1121, 
p1211, p2111, p4141, p4114, p4411, p1414, p1144, p2121, p2112, 
p2211, p1212, p1122, p4441, p4414, p4144, p1444, p2221, p2212, 
p2122, p1222,p4412, p4214, p1424, p1422, p2414, p2441, p4241, 
p4214, p2424, p3322, p2444, p4244, p3344, p3322, p2333, p3333, 
p1441, p1221, p1244, p1442 ,p4124, p4224, p4444, p4142, p2442, 
p4242, p4422, p2222, p2233, p2144, p2244, p4421))

#define the Jacobian at the no-evolution point
nep_jac = J(ones_matrix(1,60))

#compute the rank of the Jacobian at the no-evolution point
nep_jac.rank()


\end{lstlisting}
\section{SageMath Code for F81}\label{app:code81}
\begin{lstlisting}
#define the p-tilde variables and the three constants 
p1111, p1212, p1221, p1222, p2121, p2211, p2221, p2112, p2122, p3333, p2232, 
p2333 = var('p1111, p1212, p1221, p1222, p2121, p2211, p2221, p2112, 
p2122, p3333, p2232, p2333');
c1,c2,c3=var('c1,c2,c3');

#write the 7 equations
f1 = p1122*p1212*p1221 - p1111*p1222;
f2 = p1221*p2121*p2211 - p1111*p2221;
f3 = c1*p1111*p3333 - c2*p1111*p2332 - c3*p1122*p2211;
f4 = p1212*p2121 - p2112*p1221;
f5 = p1212*p2122 - p2112*p1222;
f6 = p1212*p2221 - p2212*p1221;
f7 = p1212*p2333 - p2212*p1222;

#compute the Jacobian for a generic point
J(p1111, p1212, p1221, p1222, p2121, p2211, p2221, p2112, p2122, p3333, p2232, 
p2333) = jacobian((f1,f2,f3,f4,f5,f6,f7), (p1111, p1212, p1221, p1222, p2121, 
p2211, p2221, p2112, p2122, p3333, p2232, p2333))

#define the Jacobian at the no-evolution point
nep_jac = J(1,1,1,1,1,1,1,1,1,1,1,1)

#compute the rank of the Jacobian at the no evolution point
nep_jac.rank()
\end{lstlisting}
\end{document}